\newif\ifJournal

\ifJournal

\else
\documentclass[11pt]{article}
\RequirePackage{amsthm,amsmath,amsfonts,amssymb}
\RequirePackage{graphicx}
\RequirePackage[numbers]{natbib}
\usepackage{mathtools}

\usepackage{authblk}
\usepackage[dvipsnames]{xcolor}
\usepackage[plainpages=false,pdfpagelabels,
colorlinks=true,linkcolor=RedOrange,
urlcolor=RedOrange,citecolor=MidnightBlue]{hyperref}
\providecommand{\keywords}[1]{\small \textbf{\textit{Keywords---}} #1}
\usepackage{fullpage}
\usepackage{tikz-cd}

\usepackage{my-macros}
\usepackage{marginnote} 

\usepackage{subfig}
\begin{document}

\title{Reversible Gromov-Monge Sampler for Simulation-Based Inference}

\author{YoonHaeng Hur}
\author{Wenxuan Guo}
\author{Tengyuan Liang
}
\affil{University of Chicago}

\maketitle
\begin{abstract}
	This paper introduces a new simulation-based inference procedure to model and sample from multi-dimensional probability distributions given access to i.i.d.\ samples, circumventing the usual approaches of explicitly modeling the density function or designing Markov chain Monte Carlo. Motivated by the seminal work on distance and isomorphism between metric measure spaces, we propose a new notion called the Reversible Gromov-Monge (RGM) distance and study how RGM can be used to design new transform samplers to perform simulation-based inference. Our RGM sampler can also estimate optimal alignments between two heterogeneous metric measure spaces $(\cX, \mu, c_{\cX})$ and $(\cY, \nu, c_{\cY})$ from empirical data sets, with estimated maps that approximately push forward one measure $\mu$ to the other $\nu$, and vice versa. We study the analytic properties of the RGM distance and derive that under mild conditions, RGM equals the classic Gromov-Wasserstein distance. Curiously, drawing a connection to Brenier's polar factorization, we show that the RGM sampler induces bias towards strong isomorphism with proper choices of $c_{\cX}$ and $c_{\cY}$. Statistical rate of convergence, representation, and optimization questions regarding the induced sampler are studied. Synthetic and real-world examples showcasing the effectiveness of the RGM sampler are also demonstrated.
\end{abstract}

\fi


\section{Introduction}
\label{sec:intro}
One of the central tasks in statistics is to model and sample from a multi-dimensional probability distribution. Classic statistics approaches this problem by fitting a model to the target distribution and then sampling from a fitted model via Markov Chain Monte Carlo (MCMC) techniques. Although such model-based methods are widely used, MCMC sampling often entails several technicalities. Beyond diagnosing whether the chain mixes, obtaining i.i.d.\ samples from MCMC methods is complex as one has to control correlations between successive samples or run parallel chains.

An alternative approach available in statistics, reserved for the one-dimensional case, is usually referred to as the (inverse) \textit{transform sampling}. Such an approach circumvents the calling for a parametric or nonparametric density and directly designs a sampler by transforming a simple uniform distribution. The idea is simple: one can transform a uniform measure $\mu = {\rm Unif}([0,1])$ to any one-dimensional target probability measure $\nu$ leveraging the following monotonic transformation $T \colon [0, 1] \rightarrow \R$ called the inverse Cumulative Distribution Function (CDF),
\begin{align}
	\label{eqn:inv-cdf}
	T(x)= \inf \{y \in \R ~:~ \nu((-\infty, y]) \ge x\} \;.
\end{align}
Define the pushforward measure $T_\#\mu$ by $T_\#\mu(S) = \mu(\{ x : T(x) \in S \})$ for any Borel set $S\subseteq \R$, then one can easily check that $T_{\#} \mu = \nu$; namely, with a draw from the one-dimensional uniform distribution $x\sim \mu$, the transformed sample $T(x)$ has the target probability distribution $\nu$.

The transform sampling idea can be extended to the multi-dimensional setting: given a target probability measure $\nu$ supported on $\cY$, one can specify a probability measure $\mu$ on $\cX$, which is easy to sample from such as a multivariate Gaussian, and then find a measurable map $T \colon \cX \to \cY$ such that $T_{\#} \mu = \nu$, where the pushforward measure $T_{\#} \mu$ is defined analogously to the one-dimensional case above. Such a map $T$, named as \textit{transport map} from $\mu$ to $\nu$, transforms i.i.d.\ samples from $\mu$ into i.i.d.\ samples from $\nu$. Over the past few years, the generative modeling literature has been actively employing such transform sampling ideas by identifying $T_{\#} \mu = \nu$ through the following minimization:
\begin{equation}
	\label{eqn:transport_minimization}
	\min_{T \in \cF} \cL(T_{\#} \mu, \nu) \;,
\end{equation}
where $\cF$ is a class of maps from $\cX$ to $\cY$ parametrized by neural networks and $\cL$ measures certain discrepancies between two distributions. Different choices of $\cL$ have led to various models such as the Jensen-Shannon divergence for Generative Adversarial Networks (GANs) \citep{goodfellow_2014}, the Wasserstein-$1$ distance for Wasserstein-GAN \citep{arjovsky_2017}, and the Maximum Mean Discrepancy (MMD) for MMD-GAN \citep{dziugaite_roy_ghahramani_2015, li_swersky_zemel_2015}. One caveat is that there can be infinitely many transport maps from $\mu$ to $\nu$; for instance, when $\mu = \nu = \mathrm{Unif}([0, 1])$, define $T \colon [0, 1] \to [0, 1]$ by $T(x) = |2 x - 1|$, then the $n$-fold compositions of $T$ are valid transport maps for all $n \in \N$. In other words, finding a map $T$ satisfying $T_{\#} \mu = \nu$ is an over-identified problem, where \eqref{eqn:transport_minimization} has infinitely many minimizers. Though all minimizers are equivalent in terms of transform sampling, not all are equally preferred in light of Occam's razor principle: one wishes to select simple, desirable transport maps among the over-identified set $\{ T: T_{\#} \mu = \nu \}$.

Inductive biases tackle the aforementioned over-identified problem by restricting the search to transport maps with desirable properties. In this context, there has been meaningful progress based on optimal transport (OT) theory \citep{taghvaei_2-wasserstein_2019,makkuva_optimal_2020}. The OT theory aims to identify an optimal transformation $T$, quantified by the transportation cost of moving mass from $\mu$ to $\nu$; for instance, when $\mu$ and $\nu$ lie in the same space $\R^d$, each transport map $T$ is associated with the transport cost $C(T) := \int_{\R^d} \|x - T(x)\|^2 \dd{\mu}(x)$. Brenier \citep{brenier1991} proved that, under mild regularity conditions, there exists a unique minimizer $T^{\star}$ of $C$ among all transport maps, namely, 
\begin{equation}
	\label{eqn:optimal_transport}
	T^\star = \argmin_{T_\# \mu = \nu} C(T) \;.
\end{equation}
More importantly, $T^{\star}$ is the gradient of some convex function. On the one hand, Brenier's result extends the one-dimensional (inverse) transform sampling to the multi-dimensional case. When $d = 1$ and $\mu = \mathrm{Unif}([0, 1])$, the inverse CDF map in \eqref{eqn:inv-cdf} turns out to be exactly $T^\star$; for $d > 1$, the multi-dimensional map $T^\star \colon \R^d \rightarrow \R^d$ is the gradient of a convex function, generalizing monotonic functions on the real line to multi-dimensions. On the other hand, Brenier's result naturally initiates an inductive bias in transform sampling: instead of searching any transport map, one may find $T^\star$, the optimal one with the smallest cost. To contrast this with the plain transform sampling \eqref{eqn:transport_minimization}, let us rewrite \eqref{eqn:optimal_transport} using a suitable Lagrangian multiplier $\lambda > 0$ to enforce the equality constraint $T_{\#} \mu = \nu$:
\begin{equation}
	\label{eqn:optimal_transport_Lagrangian}
	\min_{T \in \cF} C(T) + \lambda \cdot \cL(T_\# \mu, \nu) \;.
\end{equation}
Now, we can see that \eqref{eqn:optimal_transport_Lagrangian} incorporates an additional objective function of $T$---transport cost---in \eqref{eqn:transport_minimization}, thereby introducing an inductive bias towards the optimal transport map, the minimum of $C$.

Such an OT-based approach, however, can be inefficient in practice if the target $\nu$ is a high-dimensional embedding of some low-dimensional distribution. For instance, let $\nu$ be the distribution of handwritten digit images from the MNIST data set on $\R^{784}$.\footnote{Images are normalized and fit into a $28 \times 28$ pixel bounding box, hence defined on $\R^{28 \times 28} \equiv \R^{784}$ \citep{lecun1998gradient}.} To use the above OT-based approach, one must choose $\mu$ on $\R^{784}$ and find a map $T \colon \R^{784} \to \R^{784}$. However, the support of $\nu$ is intrinsically low-dimensional (roughly $\R^{15}$ as in \citep{facco2017estimating}); hence, other transform samplers with $\cX = \R^{15}$ yielding $T \colon \R^{15} \to \R^{784}$ are more efficient than the OT-based method in terms of estimating $T$ and computing $T(X)$ for $X \sim \mu$.

In this paper, we propose and study a transform sampler, combining the best of both worlds: it introduces beneficial inductive biases like the OT approach, while operating when $\cX$ and $\cY$ are heterogeneous spaces. The key to our approach is to utilize a notion of isomorphism and the Gromov-Wasserstein (GW) distance between $\mu$ and $\nu$. Given two cost functions $c_{\cX} \colon \cX \times \cX \rightarrow \R$ and $c_{\cY} \colon \cY \times \cY \rightarrow \R$, the GW distance \citep{memoli_2011,chowdhury_memoli_2019} is 
\begin{align}
	\mathrm{GW}(\mu, \nu) \coloneqq \inf_{\gamma \in \Pi(\mu, \nu)} \left( \int_{\cX \times \cY} \int_{\cX \times \cY} \big( c_{\cX}(x, x') - c_{\cY}(y, y') \big)^2 \dd{\gamma} (x, y) \dd{\gamma}(x', y') \right)^{1/2} \;,
\end{align}
where $\Pi(\mu, \nu)$ is the set of couplings between $\mu$ and $\nu$. GW aims to match the cost functions defined on two heterogeneous spaces, intending to identify an isomorphism, namely, a transport map $T$ such that $c_{\cX}(x, x') = c_{\cY}(T(x), T(x'))$ for all $x, x' \in \cX$. Inspired by these, one can define the following objective function of $T$ to replace the transport cost $C$:
\begin{equation*}
	Q(T) := \int_{\cX} \int_{\cX} (c_{\cX}(x, x') - c_{\cY}(T(x), T(x'))^2 \dd{\mu} (x) \dd{\mu}(x') \;.
\end{equation*}
Roughly speaking, the new sampler takes a form of \eqref{eqn:optimal_transport_Lagrangian}, but with $Q$ instead of $C$, thereby introducing an inductive bias towards isomorphisms (the target when $\min_{T} Q(T) = 0$). Remark that GW is a Quadratic Program (QP) in $\gamma$, which is known to be computationally hard \citep{cela1998QuadraticAssignment}; similarly, the objective function $Q$ is quadratic in $\mu$. More importantly, designing and analyzing the plug-in estimation of the quadratic objects $\mathrm{GW}(\mu, \nu)$ and $Q$ based on finite i.i.d.\ samples from $\mu$ and $\nu$ are not obvious.

\paragraph{Main contributions}
This paper considers computational and statistical questions regarding Gromov-Wasserstein outlined above, and aims to design a new transform sampler as an approach to model and sample from multi-dimensional probability distributions given access to i.i.d.\ samples, circumventing the usual ways of modeling the density function or MCMC. Our transform sampler can also estimate good alignments between two heterogeneous metric measure spaces $(\cX, \mu, c_{\cX})$ and $(\cY, \nu, c_{\cY})$ from empirical data sets, with estimated maps that approximately pushforward one measure $\mu$ to the other $\nu$, and vice versa. Towards reaching these goals, we made the following specific contributions.
\begin{itemize}
	\item We introduce a new notion, Reversible Gromov-Monge (RGM) distance, on metric measure spaces that majorizes the usual Gromov-Wasserstein distance. Moreover, we show several analytic properties possessed by GW naturally carry over to RGM; in particular, RGM induces a valid metric between metric measure spaces up to an isomorphism. Furthermore, under mild assumptions, we derive that RGM equals GW. Finally, we illustrate how RGM induces an inductive bias favoring strong isomorphisms through a new insight from Brenier's polar factorization \citep{brenier1991}.
	
	\item Our RGM formulation induces a transform sampler, modifying the usual GW formulation via decoupling and binding. Rather than solving a QP which is quadratic in the coupling $\gamma \in \Pi(\mu, \nu)$, we decouple the pair as $(\mathrm{Id}, F)_{\#} \mu$ and $(B, \mathrm{Id})_{\#} \nu$ with $F \colon \cX \rightarrow \cY$ and $B \colon \cY \rightarrow \cX$, respectively, and then bind them later via the constraint $(\mathrm{Id}, F)_{\#} \mu \approx (B, \mathrm{Id})_{\#} \nu$. Such a decoupling and binding idea will prove suitable for the statistical estimation problem based on finite i.i.d.\ samples. We will also show, from an operator viewpoint, such a decoupling and binding idea  can relax our RGM to an infinite-dimensional convex program in $F, B$ that admits a simple representation theorem, as opposed to the otherwise intractable infinite-dimensional QP in GW.
	
	\item We derive non-asymptotic rates of convergence for the proposed RGM sampler using tools from empirical processes, for generic classes modeling the measurable maps $F$ and $B$. Based on our non-asymptotic results, concrete upper bounds can be easily spelled out in the cases where $F$ and $B$ are parametrized by deep neural networks. As mentioned earlier, the RGM sampler also identifies good alignments between metric measure spaces, and learns approximate isomorphism when possible. We demonstrate such a point using numerical experiments on MNIST.
\end{itemize}

\paragraph{Organization}
The rest of the paper is organized as follows. First, we briefly review other related studies omitted in the discussion above. Then, in Section~\ref{sec:background}, preliminary background on optimal transport and Gromov-Wasserstein distance is outlined. Next, Section~\ref{sec:summary-of-results} summarizes the primary methodology and theory regarding our proposed Reversible Gromov-Monge sampler. Synthetic and real-world examples showcasing the effectiveness of the RGM sampler are demonstrated in Section~\ref{sec:numerical} as a proof of concept. The supplementary material collects details of the results in Sections \ref{sec:summary-of-results} and \ref{sec:numerical} along with extensive discussions.

\subsection{Related Literature}

Inferring the underlying probability distributions from data has been a central problem in statistics and unsupervised machine learning since the invention of histograms by Pearson a century ago. Classic mathematical statistics explicitly models the density function in a parametric or a nonparametric way \citep{silverman_1986}, and studies the minimax optimality of directly estimating such density functions \citep{stone_1982}. It is also unclear how to proceed to sample from a possibly improper\footnote{Here we mean that the estimated density is not always non-negative and integrates to one.} density estimator, even with an optimal estimator at hand. One may employ Markov Chain Monte Carlo (MCMC) techniques for sampling from specific models. However, on the computational front, it is highly non-trivial how to ensure the mixing properties of MCMC for a designed sampler \citep{robert_casella_1999}.

A recent trend in unsupervised machine learning is to learn complex, high-dimensional distributions via (deep) generative models, either explicitly by parametrizing the sufficient statistics of the exponential families \citep{doersch2016TutorialVariationalAutoencoders,kingma2013auto}, or implicitly by parametrizing the pushforward map transporting distributions \citep{dziugaite_roy_ghahramani_2015,goodfellow_2014}, with a focus on tractability in computation. Surprisingly, though lacking theoretical underpinning and optimality, the generative models' approach performs well empirically in large-scale applications where classical statistical procedures are destined to fail. There has been a growing literature on understanding distribution estimation with the implicit framework, with more general metrics and target distribution classes, to name a few, \cite{muandet_fukumizu_sriperumbudur_scholkopf_2017, li_swersky_zemel_2015, dziugaite_roy_ghahramani_2015} on MMDs, \cite{sriperumbudur2012empirical, liang2019EstimatingCertain} on integral probability metrics, and \cite{mroueh2017sobolev, arora2017generalization,liang2018HowWell,singh2018minimax,bai2018approximability, weed2019estimation, lei2019sgd, chen2020statistical} on generative adversarial networks. Last but not least, we emphasize that an alternative implicit distribution estimation approach using the simulated method of moments has been formulated in the econometrics literature since \cite{mcfadden1989MethodSimulated,pakes1989SimulationAsymptotics} and \cite{gourieroux1997SimulationbasedEconometric}. 

Originally introduced as a tool for comparing objects in computer graphics, analytic properties of the Gromov-Wasserstein distance have been studied extensively \citep{memoli_2011,sturm_2012}; the most important one is that it defines a distance between metric measure spaces, namely, metric spaces endowed with probability measures. Since many real-world data sets can be modeled as metric measure spaces, the GW distance has been utilized in various problems such as shape correspondence \citep{solomon_etal_2016}, graph matching \citep{xu2019scalable}, and protein comparison \citep{gellert2019substrate}. Certain statistical aspects of comparing metric measure spaces have been studied in \cite{brecheteau2019statistical,weitkamp2020gromov}.

Computation of the GW distance amounts to a relaxation of the quadratic assignment problem \citep{koopmans1957AssignmentProblems}; both are known to be NP-hard \citep{cela1998QuadraticAssignment} in the worst case. Several approaches have been proposed for the approximate computation of the GW distance. \cite{memoli_2011} studies lower bounds on the GW distance that are easier to compute. \cite{peyre_etal_2016} adds an entropic regularization term to the GW distance, which leads to a fast iterative algorithm; \cite{scetbon2021linear} further modifies this by imposing a low-rank constraint on couplings. \cite{titouan_etal_2019} proposes the Sliced Gromov-Wasserstein distance defined by integrating GW distances over one-dimensional projections. Last but not least, recent papers \citep{xu2019scalable,blumberg2020mrec,chowdhury2021quantized} study scalable partitioning schemes to approximately compute GW distances.

\section{Background}
\label{sec:background}
In this section, we provide background on the Optimal Transport (OT) theory and the Gromov-Wasserstein distance. First, we start with some notations. Let $\|A\|$ denote the Frobenius norm of a matrix $A$ and $\|x\|$ denote the Euclidean norm of a vector $x$. Given a set $\cX$ and a function $f \colon \cX \to \R$, let $\|f\|_\infty = \sup_{x \in \cX} |f(x)|$ denote the sup norm. For an integer $n \in \N$, we define $[n] = \{1, \ldots, n\}$. For a metric space $\cX$, we denote its metric as $d_{\cX}$ and write $\cP(\cX)$ to denote the collection of all Borel probability measures on $\cX$. We call a pair $(\cX, \mu)$ a Polish probability space if $\cX$ is a metric space that is complete and separable and $\mu \in \cP(\cX)$. Given two Polish probability spaces $(\cX, \mu)$ and $(\cY, \nu)$, the collection of all transport maps from $\mu$ to $\nu$ is denoted as $\cT(\mu, \nu) \coloneqq \{ T \colon \cX \rightarrow \cY ~|~ T_\# \mu = \nu \}$; we call $\gamma \in \cP(\cX \times \cY)$ a coupling between $\mu$ and $\nu$ if $\gamma(A \times \cY) = \mu(A)$ and $\gamma(\cX \times B) = \nu(B)$ for all Borel subsets $A \subset \cX$ and $B \subset \cY$, and we denote the collection of all such couplings as $\Pi(\mu, \nu)$. For a sequence of numbers $a(n), b(n) \in \R$, we use $a(n)\precsim b(n)$ to denote the relationship that $a(n)/b(n)\leq C, \forall n$ with some universal constant $C>0$.

\subsection{A Brief Overview of Optimal Transport Theory}
A major goal of OT is minimizing the cost associated with the transport map between two Polish probability spaces, say $(\cX, \mu)$ and $(\cY, \nu)$. Consider a measurable function $c \colon \cX \times \cY \to \R_{+}$; we view $c(x, y)$ as the cost associated with $x \in \cX$ and $y \in \cY$. For each transport map $T \in \cT(\mu, \nu)$, we interpret $c(x, T(x))$ as a unit cost incurred by mapping each $x \in \cX$ to $T(x) \in \cY$. We define the average cost incurred by the transport map $T$ as the integration of all the unit costs with respect to $\mu$, that is, $\int_{\cX} c(x, T(x)) \dd{\mu}(x)$. Minimizing the cost over $\cT(\mu, \nu)$ is referred to as the Monge problem named after Gaspard Monge. We call $T^\star$ an optimal transport map if $T^\star$ is minimizer, that is,
\begin{equation*}
	T^\star \in \argmin_{T \in \cT(\mu, \nu)} \int_{\cX} c(x, T(x)) \dd{\mu}(x)\;,
\end{equation*}

Another important OT problem is minimizing the cost given by couplings. We define the average cost incurred by a coupling $\gamma \in \Pi(\mu, \nu)$ as the integration of the cost $c$ with respect to $\gamma$, namely, $\int_{\cX \times \cY} c(x, y) \dd{\gamma}(x, y)$. Minimizing this cost over $\Pi(\mu, \nu)$ is called the Kantorovich problem credited to Leonid Kantorovich. We call $\gamma^\star$ an optimal coupling if
\begin{equation*}
	\gamma^\star \in \argmin_{\gamma \in \Pi(\mu, \nu)} \int_{\cX \times \cY} c(x, y) \dd{\gamma}(x, y) \;,
\end{equation*}

The two OT problems are closely related: the Kantorovich problem is a relaxation of the Monge problem. To see this, for each $T \in \cT(\mu, \nu)$, define a map $(\mathrm{Id}, T) \colon \cX \to \cX \times \cY$ by $(\mathrm{Id}, T)(x) = (x, T(x))$. One can verify $(\mathrm{Id}, T)_{\#} \mu \in \Pi(\mu, \nu)$. Therefore, if we define $\Pi_{\cT} \coloneqq \{(\mathrm{Id}, T)_{\#} \mu : T \in \cT(\mu, \nu)\}$, then $\Pi_{\cT} \subset \Pi(\mu, \nu)$ and thus
\begin{equation*}
	\inf_{T \in \cT(\mu, \nu)} \int_{\cX} c(x, T(x)) \dd{\mu}(x)
	= \inf_{\gamma \in \Pi_{\cT}} \int_{\cX \times \cY} c(x, y) \dd{\gamma}(x, y)
	\ge \inf_{\gamma \in \Pi(\mu, \nu)} \int_{\cX \times \cY} c(x, y) \dd{\gamma}(x, y)\;,
\end{equation*}
where the first equality follows from change-of-variables. In other words, two OT problems share the same objective function as a function of couplings; however, the Kantorovich problem has a larger constraint set.

Unlike the Monge problem, the Kantorovich problem has favorable properties. First, the objective function is linear in $\gamma$. Moreover, $\Pi(\mu, \nu)$ is compact in the weak topology of Borel probability measures defined on $\cX \times \cY$. This suggests that we can view the Kantorovich problem as an infinite-dimensional linear program.

Besides seeking optimal transport maps or couplings, another interesting aspect of OT problems is that the least possible cost can endow a metric structure among Polish probability spaces. If $\cX = \cY$ and $c = d_{\cX}^2$, the square root of the solution of the Kantorovich problem defines a distance between $\mu$ and $\nu$, known as the Wasserstein distance.

\begin{definition}
	Given a metric space $\cX$ that is complete and separable, we call
	\begin{equation*}
		W_2(\mu, \nu) = \inf_{\gamma \in \Pi(\mu, \nu)} \left(\int_{\cX \times \cX} d_{\cX}^2(x, y) \dd{\gamma}(x, y)\right)^{1/2}
	\end{equation*}
	the Wasserstein-2 distance\footnote{One can define the Wasserstein-$p$ distance by replacing the exponent $2$ above with $p \in [1, \infty]$.} between $\mu, \nu \in \cP(\cX)$.
\end{definition}

\subsection{Gromov-Wasserstein and Gromov-Monge Distances}
Although OT problems can be defined between arbitrary Polish probability spaces, in practice, it is unclear how to design a function $c \colon \cX \times \cY \to \R_+$ to represent meaningful cost associated with $x \in \cX$ and $y \in \cY$ in two heterogeneous spaces.
For instance, if $\cX = \R^p$ and $\cY = \R^q$ with $p \neq q$, there is no simple choice for a cost function $c$ over $\R^p \times \R^q$. As a result, classic OT theory (including Brenier's result) cannot be directly used for comparing heterogeneous Polish probability spaces.

M{\'e}moli's pioneering work \cite{memoli_2011} resolved this issue by considering a quadratic objective function of $\gamma$:
\begin{equation*}
	\int_{\cX \times \cY} c(x, y) \dd{\gamma}(x, y) \Rightarrow \int_{\cX \times \cY}\int_{\cX \times \cY} (c_{\cX}(x, x') - c_{\cY}(y, y'))^2 \dd{\gamma}(x, y) \dd{\gamma}(x', y') \;,
\end{equation*}
where $c_{\cX}$ and $c_{\cY}$ are defined over $\cX \times \cX$ and $\cY \times \cY$, respectively. For instance, one can specify $c_{\cX} = d_{\cX}$ and $c_{\cY} = d_{\cY}$. Rather than considering a unit cost corresponding to each pair $(x, y) \in \cX \times \cY$, we associate two pairs $(x, y), (x', y') \in \cX \times \cY$ with the discrepancy of intra-space quantities $c_{\cX}(x, x')$ and $c_{\cY}(y, y')$. In summary, by switching from the integration $\dd{\gamma}$ to the double integration $\dd{\gamma} \dd{\gamma}$, we no longer need an otherwise inter-space quantity $c \colon \cX \times \cY \to \R_+$. Therefore, we can always define this objective function whenever we have proper $c_{\cX}$ and $c_{\cY}$ in each individual space, leading to the following definition.

\begin{definition}
	\label{def:mms}
	A triple $(\cX, \mu, c_{\cX})$ is called a network space if $(\cX, \mu)$ is a Polish probability space such that $\mathrm{supp}(\mu) = \cX$ and $c_{\cX} \colon \cX \times \cX \to \R$ is measurable. The Gromov-Wasserstein distance between network spaces $(\cX, \mu, c_{\cX})$ and $(\cY, \nu, c_{\cY})$ is defined as
	\begin{equation*}
		\mathrm{GW}(\mu, \nu) = \inf_{\gamma \in \Pi(\mu, \nu)} \left( \int_{\cX \times \cY} \int_{\cX \times \cY} (c_{\cX}(x, x') - c_{\cY}(y, y'))^2 \dd{\gamma}(x, y) \dd{\gamma}(x', y') \right)^{1/2}\;.
	\end{equation*}
\end{definition}

\begin{remark} \rm
	We adopt the network space definition introduced in \cite{chowdhury_memoli_2019}. A network space $(\cX, \mu, c_{\cX})$ is called a metric measure space if $c_{\cX} = d_{\cX}$ as introduced in \cite{memoli_2011} and \cite{sturm_2012}. In short, a network space is a generalization of a metric measure space.
\end{remark}

Like the Wasserstein distance, the GW distance has metric properties; it satisfies symmetry and the triangle inequality, and $\mathrm{GW}(\mu, \nu) = 0$ if $(\cX, \mu, c_{\cX}) = (\cY, \nu, c_{\cY})$. However, the converse of this last statement does not hold in general: for its validity, a suitable equivalence relation needs to be defined on the collection of network spaces.

\begin{definition}
	\label{def:iso}
	Network spaces $(\cX, \mu, c_{\cX})$ and $(\cY, \nu, c_{\cY})$ are strongly isomorphic if there exists $T \in \cT(\mu, \nu)$ such that $T \colon \cX \to \cY$ is bijective and $c_{\cX}(x, x') = c_{\cY}(T(x), T(x'))$ for all $x, x' \in \cX$. In this case, we write $(\cX, \mu, c_{\cX}) \cong (\cY, \nu, c_{\cY})$ and such a transport map $T$ is called a strong isomorphism.
\end{definition}

One can easily check that $\cong$ is indeed an equivalence relation on the collection of network spaces. The following theorem states that the GW distance satisfies all metric axioms on the quotient space---under the equivalence relation $\cong$---of metric measure spaces.

\begin{theorem}[Lemma 1.10 of \cite{sturm_2012}]
	\label{thm:1}
	Let $\cM$ be the collection of all network spaces $(\cX, \mu, c_{\cX})$ such that $c_{\cX} = d_{\cX}$. Also, let $\cM/_{\cong}$ be the collection of all equivalence classes of $\cM$ induced by $\cong$. Then, GW satisfies the three metric axioms on $\cM/_{\cong}$.
\end{theorem}

Recall that the Monge problem is a restricted version of the Kantorovich problem with an additional constraint that couplings are given by a transport map; replacing $\Pi(\mu, \nu)$ in the Kantorovich problem with $\Pi_{\cT}$ yields the Monge problem. Imposing the same constraint on the definition of GW leads to the Gromov-Monge distance.
\begin{definition}
	The Gromov-Monge distance between network spaces $(\cX, \mu, c_{\cX})$ and $(\cY, \nu, c_{\cY})$ is defined as
	\begin{equation*}
		\mathrm{GM}(\mu, \nu) = \inf_{T \in \cT(\mu, \nu)} \left( \int_{\cX} \int_{\cX} (c_{\cX}(x, x') - c_{\cY}(T(x), T(x')))^2 \dd{\mu}(x) \dd{\mu}(x') \right)^{1/2} \;.
	\end{equation*}
\end{definition}
Loosely speaking, computing GM amounts to finding a transport map $T$ such that $c_{\cX}(x, x')$ best matches $c_{\cY}(T(x), T(x'))$ on average; we can view such a map $T$ as a surrogate for an isomorphism. See Section 2.4 of \cite{memoli2018distance} for more details of GM.

\section{Summary of Results}
\label{sec:summary-of-results}
Inspired by the Gromov-Wasserstein and Gromov-Monge distances, we propose a new metric---the reversible Gromov-Monge distance---between network spaces in this paper. Our formulation seeks a pair of transport maps $F \in \cT(\mu, \nu)$ and $B \in \cT(\nu, \mu)$ best approximating isomorphic relations between network spaces. We propose a novel transform sampling method that uses $F$ as a push-forward map to obtain i.i.d.\ samples from a target distribution $\nu$. We present two optimization formulations solving for such a pair $(F, B)$ in order: a potentially non-convex formulation that employs the standard gradient descent method to optimize, and an infinite-dimensional convex formulation where global optima can be found efficiently. For the former, we analyze the statistical rate of convergence for generic classes $\cF \times \cB$ parametrizing $(F, B)$. For the latter, we derive a new representer theorem on a suitable reproducing kernel Hilbert space (RKHS).

\subsection{Metric Properties of Reversible Gromov-Monge}
Our formulation is based on the following observation: for a coupling $\gamma$ such that $\gamma = (\mathrm{Id}, F)_{\#}\mu = (B, \mathrm{Id})_{\#} \nu$, which presents a binding constraint, we can simplify the objective function of GW as
\begin{equation*}
	\int_{\cX \times \cY} (c_{\cX}(x, B(y)) - c_{\cY}(F(x), y))^2 \dd{\mu \otimes \nu} \;,
\end{equation*}
where $\dd{\mu \otimes \nu} \coloneqq \dd{\mu}(x) \dd{\nu}(y)$ denotes the product measure of $\mu$ and $\nu$.
Imposing the binding constraint on the definition of GW leads to the following definition.
\begin{definition}\label{def:RGM}
	For network spaces $(\cX, \mu, c_{\cX})$ and $(\cY, \nu, c_{\cY})$, we write $(F, B) \in \cI(\mu, \nu)$ if measurable maps $F \colon \cX \to \cY$ and $B \colon \cY \to \cX$ satisfy the binding constraint $(\mathrm{Id}, F)_{\#}\mu = (B, \mathrm{Id})_{\#} \nu$. We define the reversible Gromov-Monge (RGM) distance between $(\cX, \mu, c_{\cX})$ and $(\cY, \nu, c_{\cY})$ as
	\begin{equation}
		\label{eq:RGM}
		\mathrm{RGM}(\mu, \nu) \coloneqq \inf_{(F, B) \in \cI(\mu, \nu)} \left(\int_{\cX \times \cY} (c_{\cX}(x, B(y)) - c_{\cY}(F(x), y))^2 \dd{\mu \otimes \nu}\right)^{1/2}\;.
	\end{equation}
\end{definition}

\begin{remark}
	\rm
	A few remarks are in place for the binding constraint.
	If $(\mathrm{Id}, F)_{\#}\mu = (B, \mathrm{Id})_{\#} \nu$, then $F_{\#} \mu = \nu$ and $B_{\#} \nu = \mu$ follow due to marginal conditions. However, the converse is not true in general. To see this, let $\mu = \nu = \mathrm{Unif}([0, 1])$, then $F_{\#} \mu = \nu$ and $B_{\#} \nu = \mu$ hold for $F(x) = B(x) = |2 x - 1|$. However, $(\mathrm{Id}, F)_{\#}\mu \neq (B, \mathrm{Id})_{\#} \nu$ because $(\mathrm{Id}, F)_{\#}\mu$ is a uniform measure on $\{(x, |2x-1|): x \in [0, 1]\}$, whereas $(B, \mathrm{Id})_{\#} \nu$ is a uniform measure on $\{(|2y-1|, y): y \in [0, 1]\}$. Lastly, note that $\cI(\mu, \nu)$ might be empty, for instance, if $\mu$ and $\nu$ are discrete and their supports have different cardinality; say, $\mu = \delta_x$ and $\nu = (\delta_{y_1} + \delta_{y_2}) / 2$, namely, Dirac measures supported on $x \in \cX$ and $y_1, y_2 \in \cY$; in such a case, $\mathrm{RGM}(\mu, \nu) = \infty$.
\end{remark}

Roughly speaking, computing RGM consists in finding a pair $(F, B) \in \cI(\mu, \nu)$ such that $c_{\cX}(x, B(y))$ best matches $c_{\cY}(F(x), y)$ on average. Like a strong isomorphism, we can view such a pair as jointly capturing an isomorphic relation of $(\cX, \mu, c_{\cX})$ and $(\cY, \nu, c_{\cY})$. We will use this observation later to build a transform sampling method.

We will prove that RGM possesses metric properties similar to the Gromov-Wasserstein. Motivated by Theorem \ref{thm:1}, we derive the following result.
\begin{theorem}
	\label{thm:metric}
	Let $h \colon \R_+ \to \R$ be a continuous and strictly monotone function and $\cN^{h}$ be a collection of all network spaces $(\cX, \mu, c_{\cX})$ such that $c_{\cX} = h(d_{\cX})$. Then RGM satisfies the three metric axioms on $\cN^h/_{\cong}$, the collection of all equivalence classes of $\cN^h$ induced by $\cong$.
\end{theorem}
\begin{remark}
	\rm
	Suppose $\cX$ is a Euclidean space and $d_{\cX}$ is the standard Euclidean distance. If $h(x) = \exp(- \alpha x^2)$ with $\alpha > 0$, then $h(d_{\cX})$ is the radial basis function (RBF) kernel on $\cX$; we will use this in numerical experiments.
\end{remark}

Readers may wonder about the generic relations among three distances GW, GM, and RGM, which will be established in the next proposition.
\begin{proposition}\label{prop:1}
	For network spaces $(\cX, \mu, c_{\cX})$ and $(\cY, \nu, c_{\cY})$ as in Definition~\ref{def:mms},
	\begin{equation}
		\label{eq:GW<GM<RGM}
		\mathrm{GW}(\mu, \nu) \le \mathrm{GM}(\mu, \nu) \le \mathrm{RGM}(\mu, \nu)\;.
	\end{equation}
\end{proposition}
Interestingly, under mild conditions, the above inequalities in Proposition~\ref{prop:1} hold as equality, thus showing that RGM provides the exact metric as GW. The proof is inspired by a construction in \cite{brenier2003, brenier1991}.
\begin{theorem}
	\label{thm:equality}
	Let $(\cX, \mu, c_\cX)$ and $(\cY, \nu, c_\cY)$ be two network spaces. Assume that $c_\cX$ and $c_\cY$ are bounded and $\mu(\{x\}) = \nu(\{y\}) = 0$ for any $(x, y) \in \cX \times \cY$. Then, $\mathrm{GW}(\mu, \nu) = \mathrm{GM}(\mu, \nu) =  \mathrm{RGM}(\mu, \nu)$.
\end{theorem}
We defer the proof details of Theorem \ref{thm:metric} and Proposition \ref{prop:1} to Section \ref{sec:metric-and-basic-properties}, and Theorem \ref{thm:equality} to Section \ref{sec:gw=rgm}. 

Finally, we conclude this section by pointing out a connection between inductive biases in RGM and Brenier's polar factorization \citep{brenier1991}. Given two Polish probability spaces $(\cX, \mu)$ and $(\cY, \nu)$, there exist cost functions $c_{\cX}$ and $c_{\cY}$ (that depend on $\mu, \nu$) such that the resulting network spaces $(\cX, \mu, c_{\cX})$ and $(\cY, \nu, c_{\cY})$ are strongly isomorphic. More importantly, among (possibly) infinitely many pairs $(F, B)$'s in $\cI(\mu, \nu)$, which are all valid for transform sampling, the optimal pair $(F^\star, B^\star)$ minimizing the RGM term \eqref{eq:RGM} achieves the strong isomorphism
\begin{align}
	\mathrm{RGM}(\mu, \nu) =  \int_{\cX \times \cY} (c_{\cX}(x, B^\star(y)) - c_{\cY}(F^\star(x), y))^2 \dd{\mu \otimes \nu} = 0 \;,
\end{align}
and thus $c_{\cX}(x, B^\star(y)) = c_{\cY}(F^\star(x), y)$ almost surely. In plain language, the RGM introduces an inductive bias favoring strong isomorphisms, in the same spirit as the Wasserstein-$2$ metric favors the transport map with the optimal cost seen in the introduction. 
The detailed discussions are deferred to Section \ref{sec:brenier-polar}.

\subsection{Transform Sampling via RGM}
\label{sec:RGM-sampler}
With the proposed notion of RGM, we design a transform sampling method in this section. The transform sampler is based on finding a minimizing pair $(F, B)$ of RGM, which can capture isomorphic relations between network spaces. To implement this method, we need to estimate $(F, B)$ using only i.i.d.\ samples from $\mu$ and $\nu$. Leveraging the Lagrangian form, we derive a minimization problem that can be implemented based on finite samples.

First, we rewrite the population minimization problem with the binding constraint as follows,
\begin{equation}		
	\label{eq:constrained_form}
	\begin{aligned}
		\min_{\substack{F \colon \cX \to \cY \\ B \colon \cY \to \cX}} \quad & \int_{\cX \times \cY} (c_{\cX}(x, B(y)) - c_{\cY}(F(x), y))^2 \dd{\mu \otimes \nu} \\
		\mathrm{s.t.} \quad & \cL_{\cX \times \cY}((\mathrm{Id}, F)_{\#} \mu, (B, \mathrm{Id})_{\#} \nu) = 0 \;. 
	\end{aligned}
\end{equation}
Here, $\cL_{\cX \times \cY}$ is a suitable discrepancy measure on $\cP(\cX \times \cY)$ so that the constraint of \eqref{eq:constrained_form} is a surrogate for the original constraint $(\mathrm{Id}, F)_{\#}\mu = (B, \mathrm{Id})_{\#} \nu$. In practice, we do not require that $\cL_{\cX \times \cY} = 0$ implies $(\mathrm{Id}, F)_{\#}\mu = (B, \mathrm{Id})_{\#} \nu$; in fact, the former constraint can be a relaxation of the latter. The choice of $\cL_{\cX \times \cY}$ will be specified later. To solve this minimization problem, we propose utilizing the Lagrangian:
\begin{equation*}
	\min_{\substack{F \colon \cX \to \cY \\ B \colon \cY \to \cX}} \int_{\cX \times \cY} (c_{\cX}(x, B(y)) - c_{\cY}(F(x), y))^2 \dd{\mu \otimes \nu} + \lambda \cdot \cL_{\cX \times \cY}((\mathrm{Id}, F)_{\#} \mu, (B, \mathrm{Id})_{\#} \nu)\;.
\end{equation*}

Given i.i.d.\ samples $\{x_i\}_{i=1}^{m}$ and $\{y_j\}_{j=1}^{n}$ from $\mu$ and $\nu$, respectively, we replace the population objective with its empirical estimates:
\begin{equation*}
	\min_{\substack{F \colon \cX \to \cY \\ B \colon \cY \to \cX}} \frac{1}{m n} \sum_{i = 1}^{m} \sum_{j = 1}^{n} (c_{\cX}(x_i, B(y_j)) - c_{\cY}(F(x_i), y_j))^2 + \lambda \cdot \cL_{\cX \times \cY}((\mathrm{Id}, F)_{\#} \widehat{\mu}_m, (B, \mathrm{Id})_{\#} \widehat{\nu}_n)\;,
\end{equation*}
where $\widehat{\mu}_m$ and $\widehat{\nu}_n$ are the empirical measures based on $\{x_i\}_{i=1}^{m}$ and $\{y_j\}_{j=1}^{n}$, respectively. Empirically, we find that adding the following extra terms often enhances empirical results:
\begin{equation*}
	\begin{split}
		\min_{\substack{F \colon \cX \to \cY \\ B \colon \cY \to \cX}} \quad & \frac{1}{m n} \sum_{i = 1}^{m} \sum_{j = 1}^{n} (c_{\cX}(x_i, B(y_j)) - c_{\cY}(F(x_i), y_j))^2 + \lambda_1 \cdot \cL_{\cX \times \cY}((\mathrm{Id}, F)_{\#} \widehat{\mu}_m, (B, \mathrm{Id})_{\#} \widehat{\nu}_n) \\
		& + \lambda_2 \cdot \cL_{\cX}(\widehat{\mu}_m, B_{\#} \widehat{\nu}_n) + \lambda_3 \cdot \cL_{\cY}(F_{\#} \widehat{\mu}_m, \widehat{\nu}_n)\;.
	\end{split}
\end{equation*}
Like $\cL_{\cX \times \cY}$, we utilize suitable discrepancy measures $\cL_{\cX}$ and $\cL_{\cY}$ so that these additional terms help matching the marginals of $(\mathrm{Id}, F)_{\#} \widehat{\mu}_m$ and $(B, \mathrm{Id})_{\#} \widehat{\nu}_n$.

Lastly, we discuss the choice of $\cL_{\cX}, \cL_{\cY}$, and $\cL_{\cX \times \cY}$. We use the square of Maximum Mean Discrepancy (MMD) as the leading example.\footnote{This is merely a proof of concept. One may use other quantities in practice, described in Section \ref{sec:numerical}.} MMD between two measures is a distance between their embeddings in some reproducing kernel Hilbert space (RKHS), which is indeed a metric under mild conditions \citep{muandet_fukumizu_sriperumbudur_scholkopf_2017}. Also, MMD is representable via the reproducing kernel of the RKHS, hence one may simply choose a kernel function to define it. Concretely, for any kernel $K_{\cX}$ on $\cX$, the square of MMD between $\widehat{\mu}_m$ and $B_{\#} \widehat{\nu}_n$ is
\begin{equation*}
	\frac{1}{m^2} \sum_{i, i'} K_{\cX}(x_i, x_{i'})	+ \frac{1}{n^2} \sum_{j, j'} K_{\cX}(B(y_j), B(y_{j'}))	- \frac{2}{mn} \sum_{i, j} K_{\cX}(x_i, B(y_j))\;.
\end{equation*}
To utilize such a convenient closed form, we specify $\cL_{\cX}, \cL_{\cY}, \cL_{\cX \times \cY}$ as the square of corresponding MMDs by choosing kernels $K_{\cX}, K_{\cY}, K_{\cX \times \cY}$ on $\cX, \cY, \cX \times \cY$. For the kernel $K_{\cX \times \cY}$ on the product space, we use the tensor product kernel $K_{\cX} \otimes K_{\cY}$ given as
\begin{equation*}
	K_{\cX} \otimes K_{\cY}((x, y), (x', y')) = K_{\cX}(x, x') K_{\cY}(y, y')\;.
\end{equation*}
The tensor product notation is employed since the kernel on the product space inherits the feature map as the tensor product of two individual feature maps w.r.t. $K_{\cX}$ and $K_{\cY}$.

Denoting the MMD associated with a kernel $K$ as $\mathrm{MMD}_{K}$, we obtain the following minimization problem:
\begin{equation}
	\label{eqn:1}
	\begin{split}
		\min_{\substack{F \colon \cX \to \cY \\ B \colon \cY \to \cX}} \quad & \frac{1}{m n} \sum_{i = 1}^{m} \sum_{j = 1}^{n} (c_{\cX}(x_i, B(y_j)) - c_{\cY}(F(x_i), y_j))^2 \\
		& + \lambda_1 \cdot \mathrm{MMD}_{K_{\cX} \otimes K_{\cY}}^2((\mathrm{Id}, F)_{\#} \widehat{\mu}_m, (B, \mathrm{Id})_{\#} \widehat{\nu}_n)\\
		& + \lambda_2 \cdot \mathrm{MMD}_{K_{\cX}}^2(\widehat{\mu}_m, B_{\#} \widehat{\nu}_n) + \lambda_3 \cdot \mathrm{MMD}_{K_{\cY}}^2(F_{\#} \widehat{\mu}_m, \widehat{\nu}_n) \;.
	\end{split}
\end{equation}

Once we solve the problem above, the solution $\widehat{F} \colon \cX\to\cY$ will serve as an approximate isomorphism and facilitate transform sampling of the target $\nu$ from a known distribution $\mu$. The map $\widehat{B}$ possesses similar properties as $\widehat{F}$, whereas the map $\widehat{F}$ is of our primary interest for sampling purposes. The reverse map $\widehat{B} \colon \cY \rightarrow \cX$ also embeds point clouds in $\cY$ into $\cX$, with approximate isomorphism properties in the sense of Gromov-Monge.

\subsection{Statistical Rate of Convergence}
\label{subsec:stat_rate}

Like other transform sampling approaches for generative models, we consider \eqref{eqn:1} using vector-valued function classes $\cF$ and $\cB$ parametrized by neural networks, and then optimize using a gradient descent algorithm. We emphasize this minimization problem is much simpler than adversarial formulations as in GANs: variational problems of GANs consist of minimization over a class of generators and maximization over a class of discriminators, which requires complex saddle-point dynamics \citep{daskalakis2017training, liang2018interaction}. In contrast, our RGM only solves a single minimization problem in network parameters. Although generally non-convex in nature, the parameter minimization problem in neural networks can often be efficiently optimized by stochastic gradient descent, and can even provably achieve the global optima if the loss satisfies certain Polyak-\L ojasiewicz conditions \citep{bassily_belkin_ma_2018}.

We investigate the statistical rate of convergence for this minimization problem, assuming the empirical problem \eqref{eqn:1} can be solved accurately. First, define
\begin{equation}
	\label{eqn:cost-def}
	\begin{split}
		C(\mu, \nu, F, B) \coloneqq & \int (c_{\cX}(x, B(y)) - c_{\cY}(F(x), y))^2 \dd{\mu \otimes \nu} \\
		& + \lambda_1 \cdot \mathrm{MMD}_{K_{\cX} \otimes K_{\cY}}^2((\mathrm{Id}, F)_{\#} \mu, (B, \mathrm{Id})_{\#} \nu)\\
		& + \lambda_2 \cdot \mathrm{MMD}_{K_{\cX}}^2(\mu, B_{\#} \nu) + \lambda_3 \cdot \mathrm{MMD}_{K_{\cY}}^2(F_{\#} \mu, \nu) \;.
	\end{split}
\end{equation}
Then, the objective function of \eqref{eqn:1} is a plug-in estimator $C(\widehat{\mu}_m, \widehat{\nu}_n, F, B)$. We consider solving \eqref{eqn:1} over the transformation class $\cF\times\cB$ given as follows, for which we will state our non-asymptotic results in full generality. From now on, let $\cX$ and $\cY$ be subsets of Euclidean spaces of dimensions $\mathrm{dim}(\cX)$ and $\mathrm{dim}(\cY)$, respectively. $\cF$ (resp.\ $\cB$) is a collection of vector-valued measurable functions from $\cX$ to $\cY$ (resp.\ from $\cY$ to $\cX$). For each $F \in \cF$ and $k \in [\mathrm{dim}(\cY)]$, we write $F_k(x)$ to denote the $k$-th coordinate of $F(x)$. Accordingly, we define $\cF_k =\{F_k: \cX \rightarrow \R ~|~ F \in \cF\}$, namely, a collection of real-valued measurable functions defined on $\cX$ that are given as the $k$-th coordinate of $F \in \cF$. For $\ell \in [\mathrm{dim}(\cX)]$, we define $B_\ell$ and $\cB_{\ell} = \{B_\ell: \cY \rightarrow \R ~|~ B \in \cB\}$ analogously.

Then, solving \eqref{eqn:1} over $\cF \times \cB$ is written as $\min_{(F, B) \in \cF \times \cB} C(\widehat{\mu}_m, \widehat{\nu}_n, F, B)$. We prove that the empirical solution leads to an approximate infimum of $(F, B) \mapsto C(\mu, \nu, F, B)$ evaluated with the population measures $\mu, \nu$, with sufficiently large sample sizes $m$ and $n$.

\paragraph{Overview of assumptions}
Before stating the next theorem, we present an overview of the assumptions. The complete statement of the assumptions and key definitions are designated to Sections~\ref{sec:statistical-theory}-\ref{sec:representation} in the supplementary material due to space constraints. Assumptions~\ref{a:bounded1} and \ref{a:lip_of_C} require the boundedness and Lipschitzness of the cost functions $c_{\cX}$ and $c_{\cY}$. Similarly, boundedness and Lipschitzness of the kernel functions $K_{\cX}, K_{\cY}$ corresponding to the MMD term are stated in Assumptions~\ref{a:bounded_kernels} and \ref{a:lip_kernels}, respectively. The last two assumptions are imposed on the set of transformations $F \colon \cX\rightarrow \cY$ and $B \colon \cY \rightarrow \cX$: Assumption~\ref{a:uniform_boundedness} requires the transformation class is bounded, and Assumption~\ref{a:separation} states that the classes should contain non-trivial maps. We shall employ a notion of combinatorial dimension to measure the complexity of real-valued function classes---the pseudo-dimension---formally stated in Definition~\ref{def:pdim}. 

\begin{theorem}
	\label{thm:stat}
	Let $(\widehat{F}, \widehat{B})$ be a solution to the empirical RGM problem 
	\begin{align*}
		& (\widehat{F}, \widehat{B}) \in \argmin_{(F, B) \in \cF \times \cB} C(\widehat{\mu}_m, \widehat{\nu}_n, F, B)\;,
	\end{align*}
	with $C \colon \cP(\cX) \times \cP(\cY) \times \cF \times \cB \rightarrow \R$ defined in \eqref{eqn:cost-def}. Under Assumptions \ref{a:bounded1}-\ref{a:separation}, the following inequality holds with probability $1- \delta$ on $\{x_i\}_{i = 1}^m$ and $\{y_j\}_{j = 1}^n$
	\begin{align}
		C(\mu, \nu, \widehat{F}, \widehat{B}) - \inf_{(F, B) \in \cF \times \cB} C(\mu, \nu, F, B)  \precsim \cM(\cF, \cB, m, n, \delta)\;.
	\end{align}
	Here, $\cM(\cF, \cB, m, n, \delta)$ denotes a complexity measure of $(\cF, \cB)$ given in terms of pseudo-dimensions (Pdim) of $\cF_k$ and $\cB_\ell$ defined in Definition~\ref{def:pdim}:
	\begin{equation*}
		\cM(\cF, \cB, m, n, \delta)
		\coloneqq
		\sqrt{\frac{\log(\tfrac{m \vee n}{\delta})}{m\wedge n}} + \sqrt{\frac{\log(m \vee n)}{m\wedge n} \left( \sum_{k =1}^{{\rm dim}(\cY)} {\rm Pdim}(\cF_{k}) + \sum_{\ell = 1}^{{\rm dim}(\cX)} {\rm Pdim}(\cB_{\ell}) \right)}\;.
	\end{equation*}
\end{theorem}
We provide required assumptions and the full proof of Theorem \ref{thm:stat} in Section \ref{sec:statistical-theory} along with the definition of the pseudo-dimension (Definition~\ref{def:pdim}). When $\cF$ and $\cB$ are parametrized by neural network classes (the ones we will use for numerical demonstrations in Section~\ref{sec:numerical}), tight pseudo-dimension bounds established in \cite{anthony_bartlett_1999, harvey2017NearlytightVCdimension} can be plugged in Theorem \ref{thm:stat} for concrete non-asymptotic rates.

\subsection{Convex Formulation and Representer Theorem}
\label{subsec:cvx-representer}

As the last bit of our contributions, we study a convex formulation of solving \eqref{eqn:1} by relaxing and lifting it to an infinite-dimensional space. There are two reasons behind our convex formulation: first, as a computational alternative to the possibly non-convex optimization; second, to point out a connection with the Nadaraya-Watson estimator in classic nonparametric statistics. The crux lies in relaxing optimizing over the map $F \colon \cX \rightarrow \cY$ to optimizing over its induced (dual) linear operator $\bF \colon L^2_\cY \rightarrow L^2_\cX$ that maps functions on $\cY$ to functions on $\cX$, where $L^2_{\cX}$ is the collection of real-valued measurable functions $f$ defined on $\cX$ such that $\int_{\cX} f^2 \dd{\pi_{\cX}} < \infty$ given a Borel measure $\pi_\cX$ on $\cX$; similarly, define $L^2_{\cY}$ given a Borel measure $\pi_{\cY}$ on $\cY$. Then, for a measurable map $F \colon \cX \to \cY$, we can define $\bF \colon L^2_{\cY} \to L^2_{\cX}$ by letting $\bF(g) = g \circ F$ for all $g \in L^2_{\cY}$. Similarly, we define $\bB \colon L^2_{\cX} \to L^2_{\cY}$ for each measurable map $B \colon \cY \to \cX$. We will see $\bF$ and $\bB$ are well-defined bounded linear operators in Section \ref{sec:representation} under a mild assumption. 

To state the representer theorem, consider \eqref{eqn:1} with $c_{\cX} = K_{\cX}$ and $c_{\cY} = K_{\cY}$, same as kernel functions specified in MMD terms. We show that this problem can be reduced to a finite-dimensional convex optimization by proving a representer theorem. Since finite-dimensional convex optimization can be optimized globally with provable guarantees, such a formulation can be solved numerically in an efficient way.

Let us lay out more details to state the result. Due to Mercer's theorem, let $\{ \phi_k \in L^2_{\cX} \}_{k \in \N} $ and $\{ \psi_\ell \in L^2_{\cY} \}_{\ell \in \N}$ be countable orthonormal bases of $L^2_{\cX}$ and $L^2_{\cY}$ where the kernels admit the following spectral decompositions:
\begin{align}
	\label{eqn:spectral}
	K_{\cX}(x, x') = \sum_k \lambda_k \phi_k(x) \phi_k(x') \;, \quad   K_{\cY}(y, y') = \sum_{\ell} \gamma_\ell \psi_\ell(y) \psi_\ell(y') \;,
\end{align}
with positive eigenvalues $\lambda_k, \gamma_\ell >0$.
Since $\bF \colon L^2_{\cY} \rightarrow L^2_{\cX}$ defines a bounded linear operator, one can represent $\bF$ (correspondingly $\bB$) under the orthonormal bases
\begin{align}
	\bF[\psi_{\ell}] = \sum_{k = 1}^{\infty} \bF_{k \ell} \phi_k\;, \quad \bB[\phi_{k}] = \sum_{\ell = 1}^{\infty} \bB_{\ell k} \psi_\ell\;.
\end{align}
Here, $[\bF_{k \ell}]$ is a semi-infinite matrix with each column describing the $L^2_{\cX}$ representation of $\bF[\psi_{\ell}]$ under the basis $\{ \phi_k \in L^2_{\cX} \}_{k \in \N}$. With a slight abuse of notation, we will write $\bF$ and $\bB$ to denote these matrices $[\bF_{k \ell}]$ and $[\bB_{\ell k}]$. Then, we will prove in Section \ref{sec:representation} that the objective function in \eqref{eqn:1} with $c_{\cX} = K_{\cX}$ and $c_{\cY} = K_{\cY}$ is 
\begin{equation*}
	\begin{split}
		\Omega(\bF, \bB)
		&\coloneqq
		\frac{1}{mn} \sum_{i, j} (\Psi_{y_j}^\top \bB \Lambda \Phi_{x_i} - \Phi_{x_i}^\top \bF \Gamma \Psi_{y_j})^2 \\
		&+ \lambda_1 \cdot \Bigg(\frac{1}{m^2} \sum_{i, i'} \Phi_{x_i}^\top \Lambda \Phi_{x_i'} \Phi_{x_i}^\top \bF \Gamma \bF^\top \Phi_{x_{i'}} + \frac{1}{n^2} \sum_{j, j'} \Psi_{y_j}^\top \Gamma \Psi_{y_{j'}} \Psi_{y_j}^\top \bB \Lambda \bB^\top \Psi_{y_{j'}} \\
		&\hspace{200pt}- \frac{2}{m n} \sum_{i, j} \Psi_{y_j}^\top \bB \Lambda \Phi_{x_i} \Phi_{x_i}^\top \bF \Gamma \Psi_{y_j} \Bigg) \\
		&+ \lambda_2 \cdot \left(\frac{1}{m^2} \sum_{i, i'} \Phi_{x_i}^\top \Lambda \Phi_{x_i'} + \frac{1}{n^2} \sum_{j, j'} \Psi_{y_j}^\top \bB \Lambda \bB^\top \Psi_{y_{j'}} - \frac{2}{m n} \sum_{i, j} \Psi_{y_j}^\top \bB \Lambda \Phi_{x_i}\right) \\
		&+ \lambda_3 \cdot \left(\frac{1}{m^2} \sum_{i, i'} \Phi_{x_i}^\top \bF \Gamma \bF^\top \Phi_{x_{i'}} + \frac{1}{n^2} \sum_{j, j'} \Psi_{y_j}^\top \Gamma \Psi_{y_{j'}} - \frac{2}{m n} \sum_{i, j} \Phi_{x_i}^\top \bF \Gamma \Psi_{y_j}\right) \;.
	\end{split}
\end{equation*}
Here, $\bF$ and $\bB$ are the matrices denoting the operators induced by $F$ and $B$, respectively, $\Phi_x = [\cdots, \phi_k(x), \cdots ]^\top \in \R^\infty$ and $\Psi_y = [\cdots, \psi_\ell(y), \cdots]^\top \in \R^\infty$ for any $x \in \cX$ and $y \in \cY$, and $\Lambda = \mathrm{diag}(\lambda_1, \lambda_2, \dots)$ and $\Gamma = \mathrm{diag}(\gamma_1, \gamma_2, \dots)$ are diagonal matrices. Hence, \eqref{eqn:1} can be lifted to an infinite-dimensional optimization problem
\begin{equation}
	\label{eqn:op}
	\min_{(\bF, \bB) \in \cC} ~ \Omega(\bF, \bB)\;,
\end{equation}
where $\cC$ denotes the constraint set implying that $\bF$ and $\bB$ are matrices corresponding to bounded linear operators induced by some maps $F \colon \cX \to \cY$ and $B \colon \cY \to \cX$.

We will relax this problem by removing the constraint set $\cC$, namely, by considering all matrices in $\R^{\infty \times \infty}$ as the decision variables, 
\begin{equation}
	\label{eqn:relaxed}
	\min_{\bF, \bB \in \R^{\infty \times \infty}} ~ \Omega(\bF, \bB)\;.
\end{equation}
In other words, this relaxed problem minimizes $\Omega$ over any pair of infinite-dimensional matrices. The next result, which we refer to as the representer theorem, shows that \eqref{eqn:relaxed} boils down to a finite-dimensional convex program. 

\begin{theorem}
	\label{thm:representer}
	Consider the optimization \eqref{eqn:op} under the assumptions in Proposition \ref{prop:representer}. Then, for any minimizer $(\bF^\star, \bB^\star)$ to the relaxed problem \eqref{eqn:relaxed}, we can find finite-dimensional matrices $\mathsf{F}_{m,n}^{\star} \in \R^{m \times n}$ and $\mathsf{B}_{n, m}^{\star} \in \R^{n \times m}$ such that
	\begin{align*}
		\bF^{\star}  = \Lambda \Phi_m \mathsf{F}_{m,n}^{\star} \Psi_n^\top \;, \quad \bB^{\star}  = \Gamma \Psi_n \mathsf{B}_{n,m}^{\star} \Phi_m^\top \;,
	\end{align*}
	where $\Lambda = \mathrm{diag}(\lambda_1, \lambda_2, \dots)$, $\Gamma = \mathrm{diag}(\gamma_1, \gamma_2, \dots)$, and $\Phi_{m} \in \R^{\infty \times m}$ and $\Psi_{n} \in \R^{\infty \times n}$ are matrices whose elements are $\phi_k(x_i)$ and $\psi_{\ell}(y_j)$, as defined in \eqref{eqn:spectral}. In this case, $\Omega(\bF^\star, \bB^\star)$ can be rewritten as $\omega(\mathsf{F}_{m, n}^{\star}, \mathsf{B}_{n, m}^{\star})$ for some convex function $\omega$ defined over $\R^{m \times n} \times \R^{n \times m}$. Hence, by minimizing $\omega$ over $\R^{m \times n} \times \R^{n \times m}$, we obtain a relaxation of \eqref{eqn:relaxed}, that is, 
	\begin{equation*}
		\min_{\bF, \bB \in \R^{\infty \times \infty}} \Omega(\bF, \bB)
		\ge
		\min_{\substack{\mathsf{F}_{m, n} \in \R^{m \times n} \\ \mathsf{B}_{n, m} \in \R^{n \times m}}} \omega(\mathsf{F}_{m, n}, \mathsf{B}_{n, m}) \;.
	\end{equation*}
	In particular, the RHS is a finite-dimensional convex optimization. Lastly, this relaxation is tight, that is,  
	\begin{equation*}
		\min_{\bF, \bB \in \R^{\infty \times \infty}} \Omega(\bF, \bB)
		=
		\min_{\substack{\mathsf{F}_{m, n} \in \R^{m \times n} \\ \mathsf{B}_{n, m} \in \R^{n \times m}}} \omega(\mathsf{F}_{m, n}, \mathsf{B}_{n, m}) \;,
	\end{equation*}
	if kernel matrices $\bK_{\cX}$ and $\bK_{\cY}$ whose elements are $K_{\cX}(x_i, x_{i'})$ and $K_{\cY}(y_j, y_{j'})$, are positive definite. 
\end{theorem}

\begin{remark}
	\rm
	Looking inside the proof of Theorem \ref{thm:representer}, we know the solution to the infinite-dimensional optimization is an operator taking form of $\bF^\star = \Lambda \Phi_m \mathsf{F}_{m,n}^\star \Psi_n^\top$, with a finite-dimensional matrix $\mathsf{F}_{m,n}^\star\in \R^{m \times n}$. Therefore, for any $g \in L^2_\cY$, we can deduce
	\begin{align}
		\label{eqn:nw-connection}
		\bF^\star[g](x) =  \underbrace{K_\cX(x, X_m)}_{1\times m} \underbrace{\mathsf{F}_{m,n}^\star}_{m\times n} \underbrace{g(Y_n)}_{n\times 1} \;,
	\end{align}
	where $K_{\cX}(x, X_m)$ maps each $x \in \cX$ to a row vector whose $i$-th element is $K_{\cX}(x, x_i)$ and $g(Y_n)$ denotes a column vector whose $j$-th element is $g(y_j)$. 
	
	Now let's draw a connection between the classic Nadaraya-Watson estimator and \eqref{eqn:nw-connection}. For now consider a special case: $(x_i, y_i)$'s are paired with $m = n$. In such a case, Nadaraya-Watson estimator takes the form
	\begin{align}
		\sum_{i,j} K_{\cX}(x, x_i) \cdot \tfrac{1}{m} \delta_{i=j} \cdot g(y_j) \; ;
	\end{align}
	Namely, for a new point $x$, the corresponding function value $g(y)$ evaluated on its coupled $y = F(x)$ is a weighted average of $g(y_j)$'s according to the affinity $K_{\cX}(x, x_i)$. Our solution \eqref{eqn:nw-connection} extends the above nonparametric smoothing idea to the decoupled data case, where the coupling weights $\mathsf{F}_{m,n}^\star$ is based on a solution to a convex program, with
	\begin{align}
		\eqref{eqn:nw-connection} = \sum_{i,j} K_{\cX}(x, x_i) \cdot \mathsf{F}_{m,n}^\star[i,j] \cdot g(y_j) \;.
	\end{align}
	
	Lastly, we draw another connection to the Monte-Carlo integration. One downstream task after learning the distribution $\nu$ is to perform numerical integration of $g \in L^2_\cY$ under the measure $\nu \in \cP(\cY)$. In our transform sampling framework, this amounts to evaluating $\E_{y \sim F^\star_\# \mu}[ g(y)] = \E_{x \sim \mu}[ g\circ F^\star (x)]$. The integration, casted in the induced operator form, has the expression
	\begin{align}
		\label{eqn:weights-MC}
		\E_{x \sim \mu} \big[ \bF^\star[g](x) \big] = \E_{x \sim \mu}  \big[ \underbrace{K_\cX(x, X_m) \mathsf{F}_{m,n}^\star}_{=: W(x) \in \R^n } g(Y_n) \big] = \E_{x \sim \mu} \big[ \sum_{j=1}^n W_j(x) g(y_j) \big]
	\end{align}
	where $W(x)$ can be interpreted as the importance weights in the Monte-Carlo integration. We conclude with one more remark: if plug in instead $x \sim \widehat{\mu}_m$ in \eqref{eqn:weights-MC}, one can verify that under mild conditions, 
	\begin{align}
		\E_{x \sim \widehat{\mu}_m} \big[ \bF^\star[g](x) \big] = \frac{1}{n} \sum_{j=1}^n  g(y_j) \;.
	\end{align}
	That is, with the empirical measure as input, \eqref{eqn:weights-MC} outputs the simple sample average. 
\end{remark}

\section{Experiments}
\label{sec:numerical}

This section examines the empirical performance of the reversible Gromov-Monge sampler. Following Section \ref{subsec:stat_rate}, we find a minimum $(\widehat{F}, \widehat{B})$ of \eqref{eqn:1} over a suitable class $\cF \times \cB$ via gradient descent; we inspect the quality of transform sampling $(\widehat{F}_{\#} \mu \approx \nu$) and space isomorphism. Complete technical details of the experiments are deferred to Section~\ref{sec:appendix2}.

\paragraph{Gaussian distributions}	
Consider two strongly isomorphic Gaussian distributions on $\cX = \cY = \R^2$: the base measure $\mu = N(0, I_2)$ and the target distribution $\nu = N(0, \Sigma)$, where $I_2$ is the identity matrix and the entries of $\Sigma$ are $\Sigma_{1 1} = \Sigma_{2 2} = 1$ and $\Sigma_{1 2} = \Sigma_{2 1} = 0.7$. We let $c_{\cX}(x, x') =  x^\top x'$ and $c_\cY(y, y') =  y^\top \Sigma^{-1} y'$, then two network spaces are strongly isomorphic by design; indeed, any pair $(F, B)$ given by $F(x) = \Sigma^{1/2} Q x$ and $B(y) = Q^{\top} \Sigma^{-1/2} y$ for $Q \in O(2)$, where $O(2)$ is the orthogonal group, yields $c_{\cX}(x, B(y)) = c_{\cY}(F(x), y)$ for all $x, y \in \R^2$, hence $F$ and $B$ are strong isomorphisms. We aim at obtaining such a pair of (linear) isomorphisms by letting $\cF = \cB = \{x \mapsto W x : W \in \R^{2 \times 2}\}$, that is, the collection of all linear maps from $\R^2$ to $\R^2$. We set $K_{\cX} = K_{\cY}$ as a degree-2 polynomial kernel that maps $(x, y)$ to $(x^\top y + 1)^2$; the resulting MMD compares distributions by matching the first two moments, which is sufficient to distinguish Gaussian distributions. The resulting linear maps are given by $\widehat{F}(x) = \mathbf{F} x$ and $\widehat{B}(y) = \mathbf{B} y$ for some $\mathbf{F}, \mathbf{B} \in \R^{2 \times 2}$ satisfying
\begin{equation*}
	\mathbf{F} \mathbf{F}^\top
	= 
	\begin{pmatrix}
		1.035 & 0.751 \\
		0.751 & 1.094	
	\end{pmatrix}\;,
	\quad \mathbf{B} \Sigma \mathbf{B}^\top 
	=
	\begin{pmatrix}
		0.940 & 0.001 \\
		0.001 & 0.944
	\end{pmatrix} \; ,
	\quad \mathbf{F} \mathbf{B}
	=
	\begin{pmatrix}
		0.966 & 0.029 \\
		-0.020 & 1.029
	\end{pmatrix} \; .
\end{equation*}
Since $\mathbf{F} \mathbf{F}^\top \approx \Sigma$, $\mathbf{B} \Sigma \mathbf{B}^\top \approx I_2$, and $\mathbf{F} \mathbf{B} \approx I_2$, the pair $(\widehat{F}, \widehat{B})$ can be seen as an instance of the pair of strong isomorphisms described above. Figure \ref{fig:Gaussian} illustrates that $\widehat{F}$ is a strong isomorphism (Definition \ref{def:iso}): (a) shows that $\widehat{F}_{\#} \mu \approx \nu$, that is, $\widehat{F}$ is roughly a transport map, and (b) implies that $c_{\cX}(x, x') \approx c_{\cY}(\widehat{F}(x), \widehat{F}(x'))$ holds.
\begin{figure}[ht]
    \centering
    \subfloat[\centering]{{\includegraphics[width=8cm]{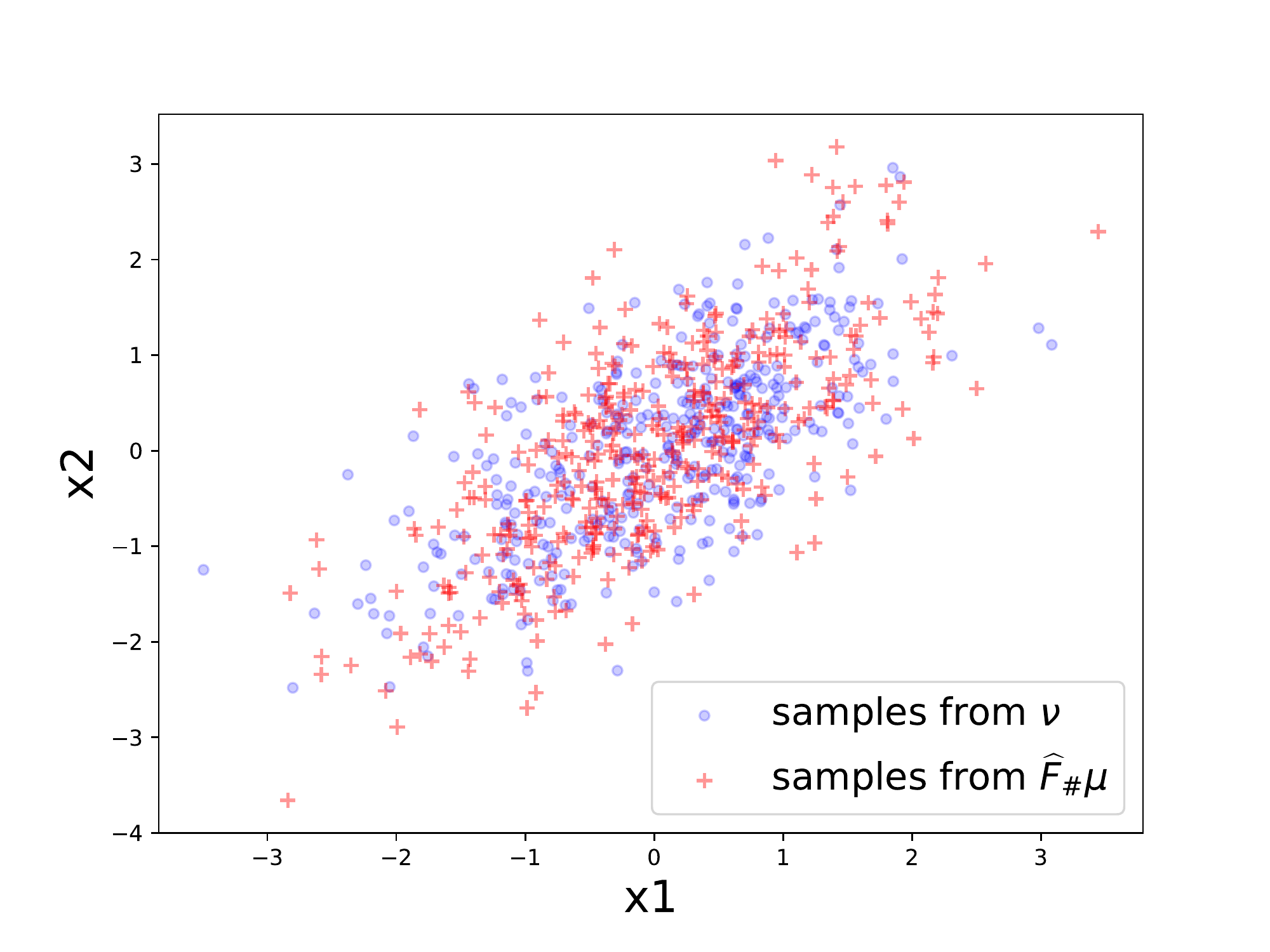}}}
    \subfloat[\centering]{{\includegraphics[width=8cm]{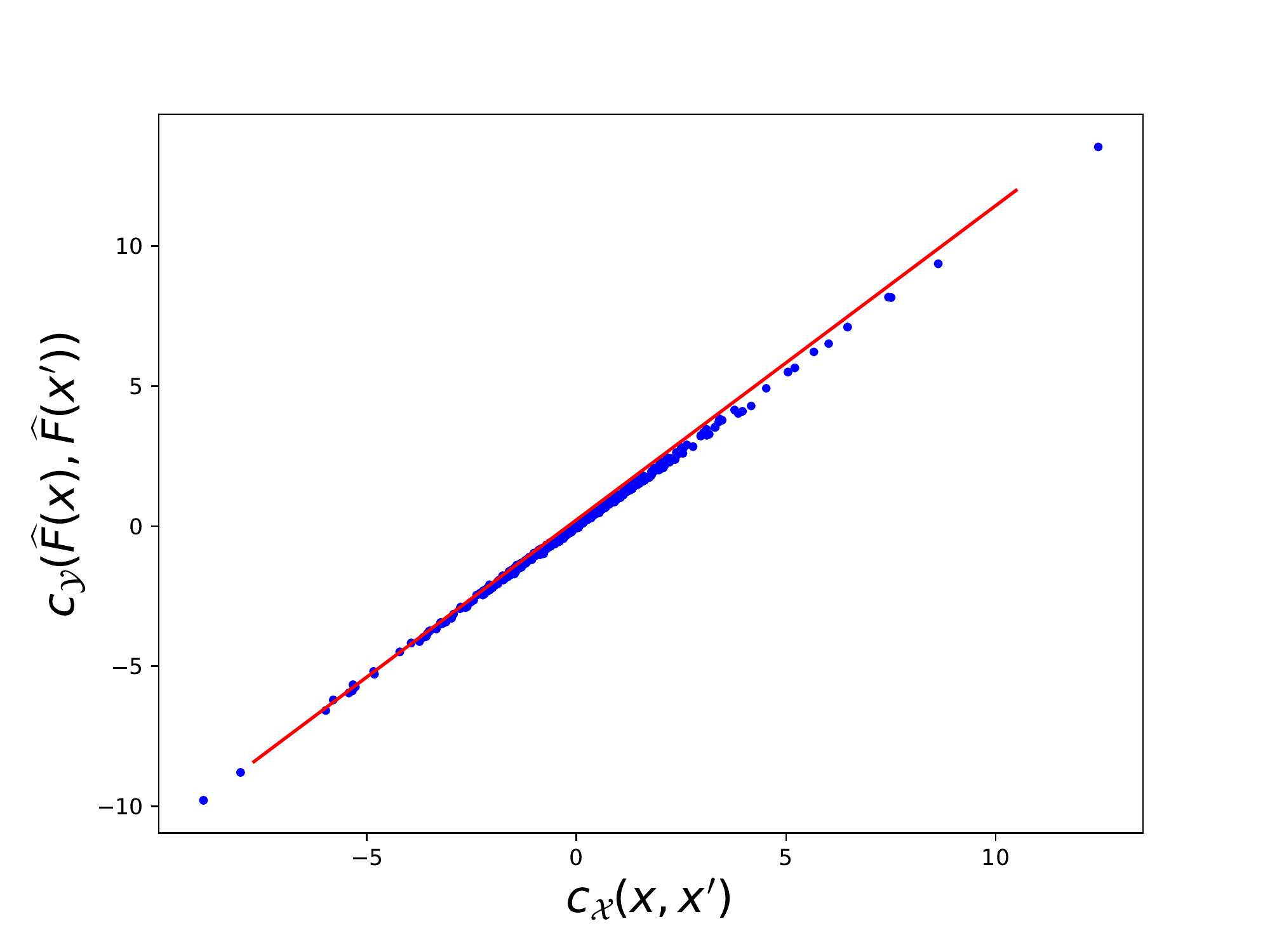}}}
    \caption{Gaussian experiment: $m = n = 1000$ and $\lambda_1 = \lambda_2 = \lambda_3 = 1$. (a) shows $\{\tilde{y}_j\}_{j = 1}^{400}$ versus $\{\widehat{F}(\tilde{x}_i)\}_{i = 1}^{400}$, where $\{\tilde{y}_j\}_{j = 1}^{400}$ and $\{\tilde{x}_i\}_{i = 1}^{400}$ are i.i.d.\ from $\nu = N(0, \Sigma)$ and $\mu = N(0, I_2)$, respectively; they are new samples independent from $\{y_j\}_{j = 1}^{1000}$ and $\{x_i\}_{i = 1}^{1000}$ used in \eqref{eqn:1}. (b) shows the points $\{(c_{\cX}(\tilde{x}_i, \tilde{x}_{i'}), c_{\cY}(\widehat{F}(\tilde{x}_i), \widehat{F}(\tilde{x}_{i'})))\}_{i, i' = 1}^{40}$ and a straight line $y = x$.}
    \label{fig:Gaussian}
\end{figure}
\paragraph{MNIST}
Next, let $\nu$ be a distribution of images corresponding to four digits (2, 4, 6, 7) from the MNIST data set, which is supported on $\R^{784}$. Recall from Section \ref{sec:intro} that the support $\cY$ of $\nu$ is low-dimensional \citep{facco2017estimating}; hence, choosing $\cX = \R^{d}$ with $d \ll 784$ is reasonable. Here, for visualization, we try an extreme embedding task with $d = 2$ and $\mu = N(0, I_2)$, that is, generate MNIST images by transforming two-dimensional Gaussian samples. 

Unlike the Gaussian example where we design the cost functions in advance to make the two spaces strongly isomorphic, specifying them can be more complicated in general cases, which might affect the quality of the RGM sampler. Here, we briefly discuss some of the most commonly used cost functions: given a fixed exponent $p \in \N$ or constant $\alpha > 0$,
\begin{equation*}
	(x, y) \quad \mapsto \quad
	\underbrace{\|x - y\|^p}_{\text{distance-based}}
	\quad \text{or} \quad \underbrace{\exp(-\alpha \|x - y\|^2)}_{\text{RBF kernel}} \;.
\end{equation*}
Clearly, $\|x - y\|^p$ is the most straightforward choice in Euclidean cases; $p = 1$ and $p = 2$ are indeed widely used in the literature \citep{peyre_etal_2016}. The RBF kernel, also referred to as the heat kernel, is a common choice in the object matching literature \citep{solomon_etal_2016}. In this MNIST example, we have found that these cost functions provide reasonable performance once they are scaled properly. Here, we will present the results based on the RBF kernel. Concretely, first define the RBF kernel $K_d(x, y) = \exp(-\|x - y\|^2 / d)$ for $d \in \N$ and $x, y \in \R^d$; here, the constant $(1 /d)$ serves as a scaling factor. Then, we define the cost functions as $c_\cX = (K_{2} - m_\cX) / \mathrm{sd}_\cX$ and $c_\cY = (K_{784} - m_\cY) / \mathrm{sd}_\cY$, where $m_\cX$ and $\mathrm{sd}_\cX$ are the median and the standard error of $\{K_\cX(x_i, x_{i'})\}_{i, i' = 1}^m$, respectively; $m_\cY$ and $\mathrm{sd}_\cY$ are defined analogously. This additional standardization process helps aligning the cost functions.

In the same vein, $K_{\cX}$ and $K_{\cY}$ must be properly specified; comparing the first two moments using the degree-2 polynomial kernel is no longer sufficient as the target distribution is non-Gaussian. We suggest using RBF kernels for the MMD terms as well; let $K_\cX = K_2$ and $K_\cY = K_{784}$. The MMD induced by the RBF kernel indeed defines a metric between distributions under mild assumptions \citep{muandet_fukumizu_sriperumbudur_scholkopf_2017}, which allows the resulting MMD terms to represent the original constraint of the RGM distance as mentioned in Section \ref{sec:RGM-sampler}. 

For the function classes $\cF$ and $\cB$, we need richer classes instead of the linear maps used in the Gaussian case. To this end, we will use fully connected neural networks with three hidden layers, each of which consists of 50 neurons. Lastly, we let $m = n = 20000$ and $\lambda_1 = \lambda_2 = \lambda_3 = 100$. Figure \ref{fig:MNIST-images}(a) shows the images generated by applying the resulting map $\widehat{F}$ to new i.i.d.\ samples from $\mu = N(0, I_2)$. Though not perfect, we see that recognizable images can be generated by transforming two-dimensional Gaussian samples, efficient in computation.\footnote{Computational cost for obtaining $\widehat{F} \colon \R^2 \to \R^{784}$ and computing $\widehat{F}(X)$ from $X \sim \mu$ is far less than that of the OT-based sampler as explained in Section \ref{sec:intro}.}
\begin{figure}[ht]
	\centering
	\subfloat[MMD ($\R^2$)]{{\includegraphics[width=5cm]{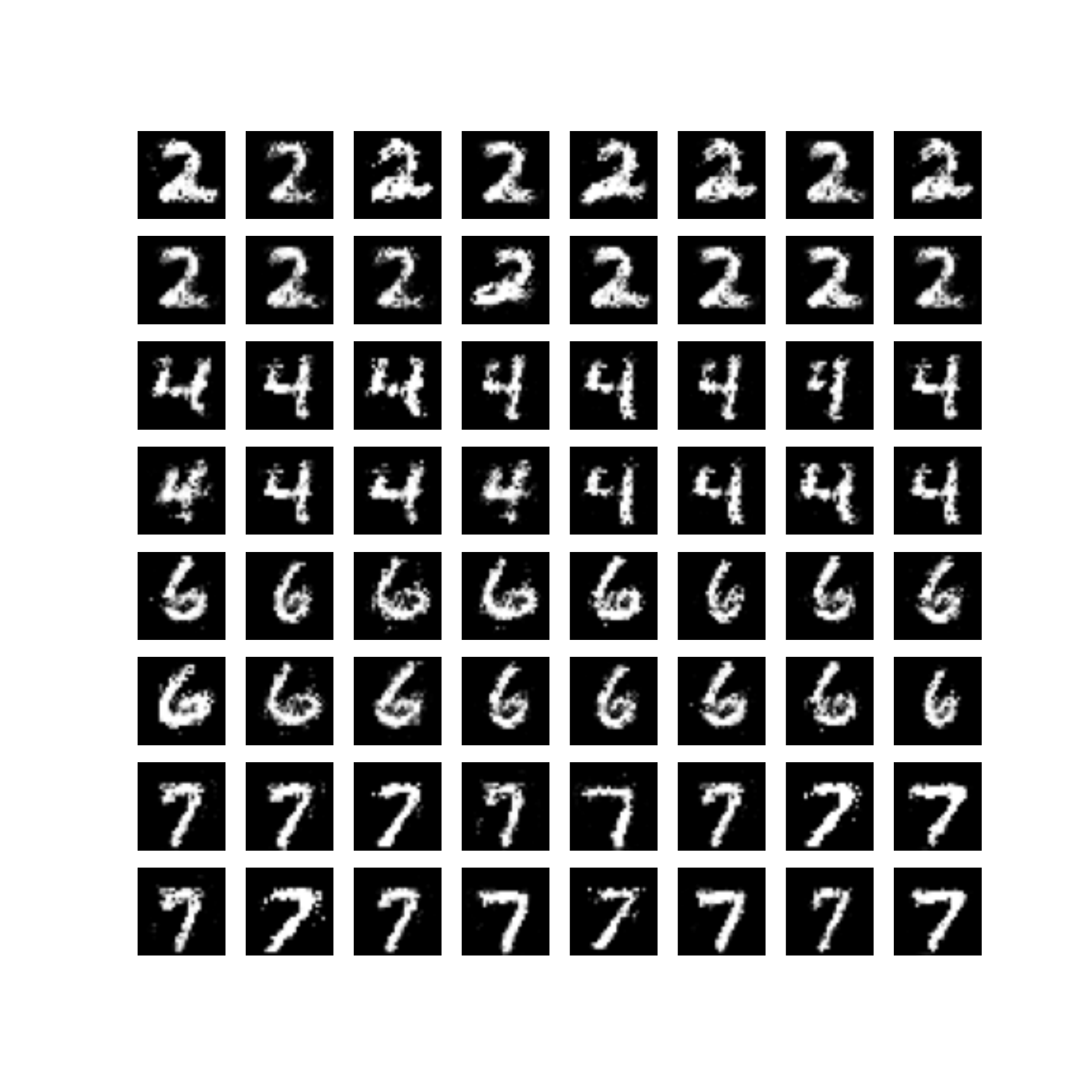}}}%
    \subfloat[Sinkhorn ($\R^4$)]{{\includegraphics[width=5cm]{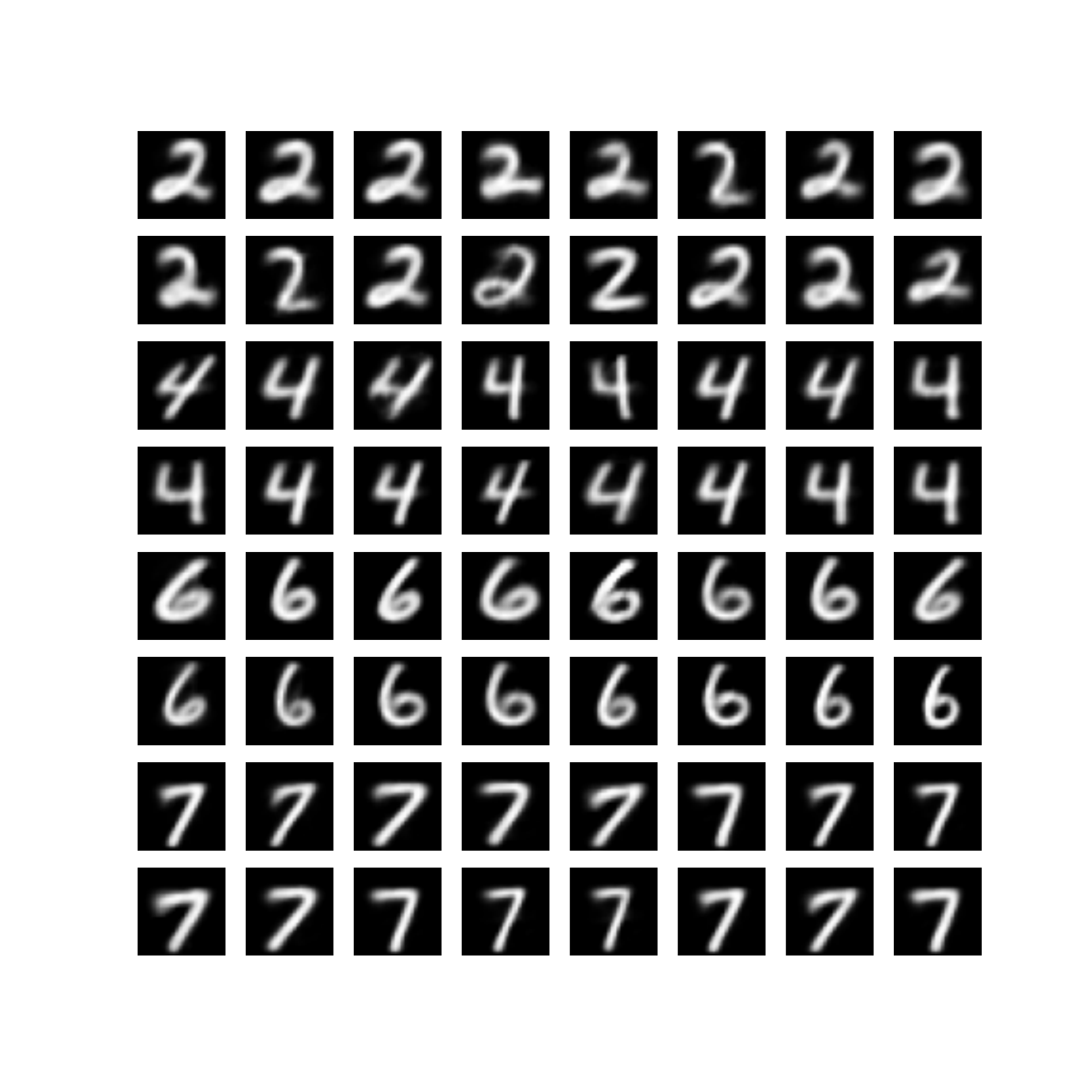}}}%
    \subfloat[Original]{{\includegraphics[width=5cm]{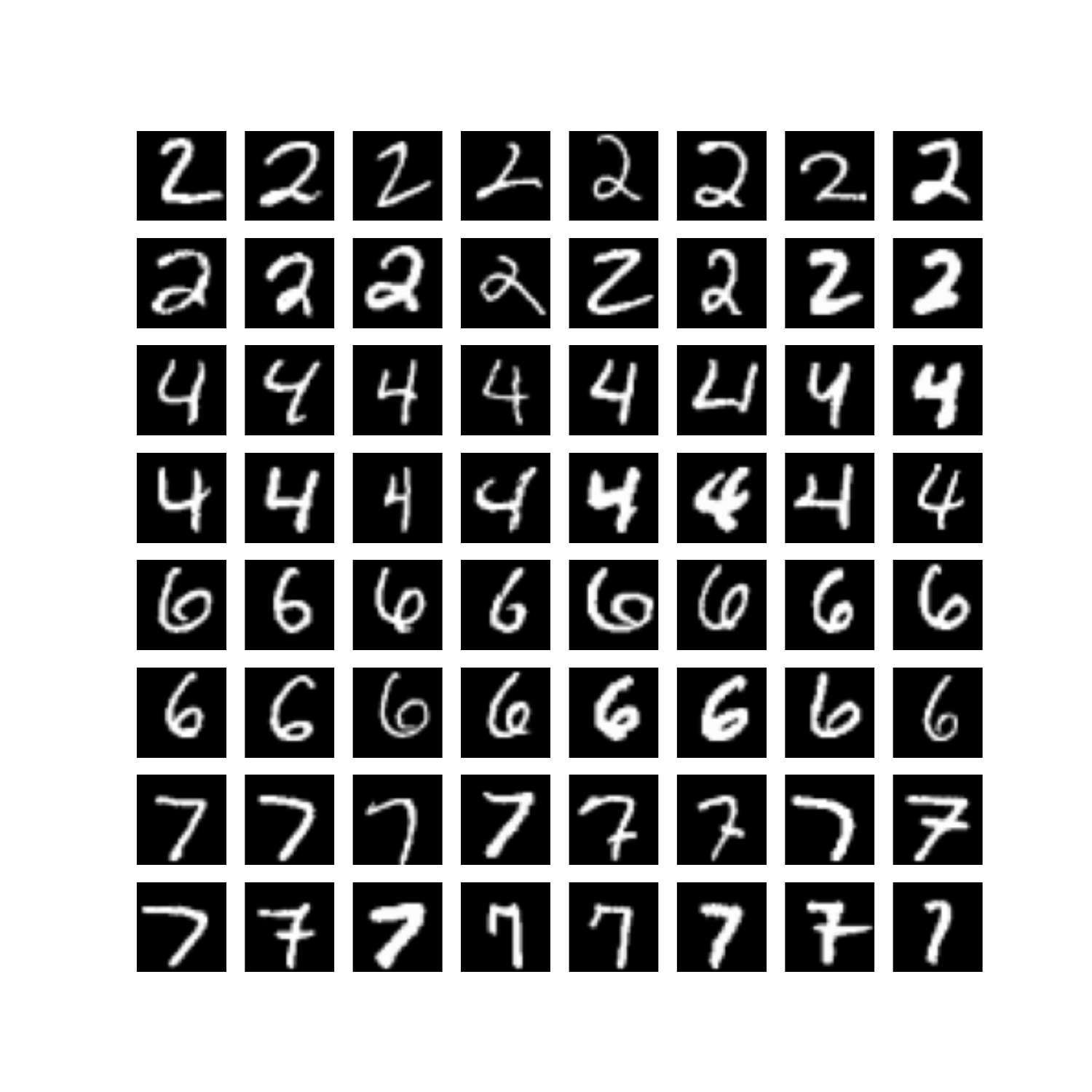}}}%
    \caption{(a) and (b) are generated by transforming new i.i.d.\ samples from $\mu$ using $\widehat{F}$: (a) from $\mu = N(0, I_2)$ with MMDs and (b) from $\mu = N(0, I_4)$ with Sinkhorn divergences. (c) shows real MNIST images.}
    \label{fig:MNIST-images}
\end{figure}
Meanwhile, the map $\widehat{B}$ shows how the MNIST images can be embedded in $\R^2$. Figure \ref{fig:MNIST-backward}(a) shows $\{\widehat{B}(\tilde{y}_j)\}_{j = 1}^{500}$, where $\{\tilde{y}_j\}_{j = 1}^{500}$ are i.i.d.\ from $\nu$ (125 $\times$ 4 digits), independent from $\{y_j\}_{j = 1}^{20000}$ used in \eqref{eqn:1}. We see that each digit forms a local cluster in $\R^2$, each of which is roughly representable according to the range of the angular coordinate. Lastly, though not perfect as in Figure \ref{fig:Gaussian}(b) (strongly isomorphic case), Figure \ref{fig:MNIST-backward}(b) shows that $\widehat{B}$ leads to a reasonable alignment of $c_\cX(\widehat{B}(y), \widehat{B}(y'))$ versus $c_\cY(y, y')$.

\begin{figure}[ht]
	\centering
    \subfloat[\centering]{{\includegraphics[height=5cm]{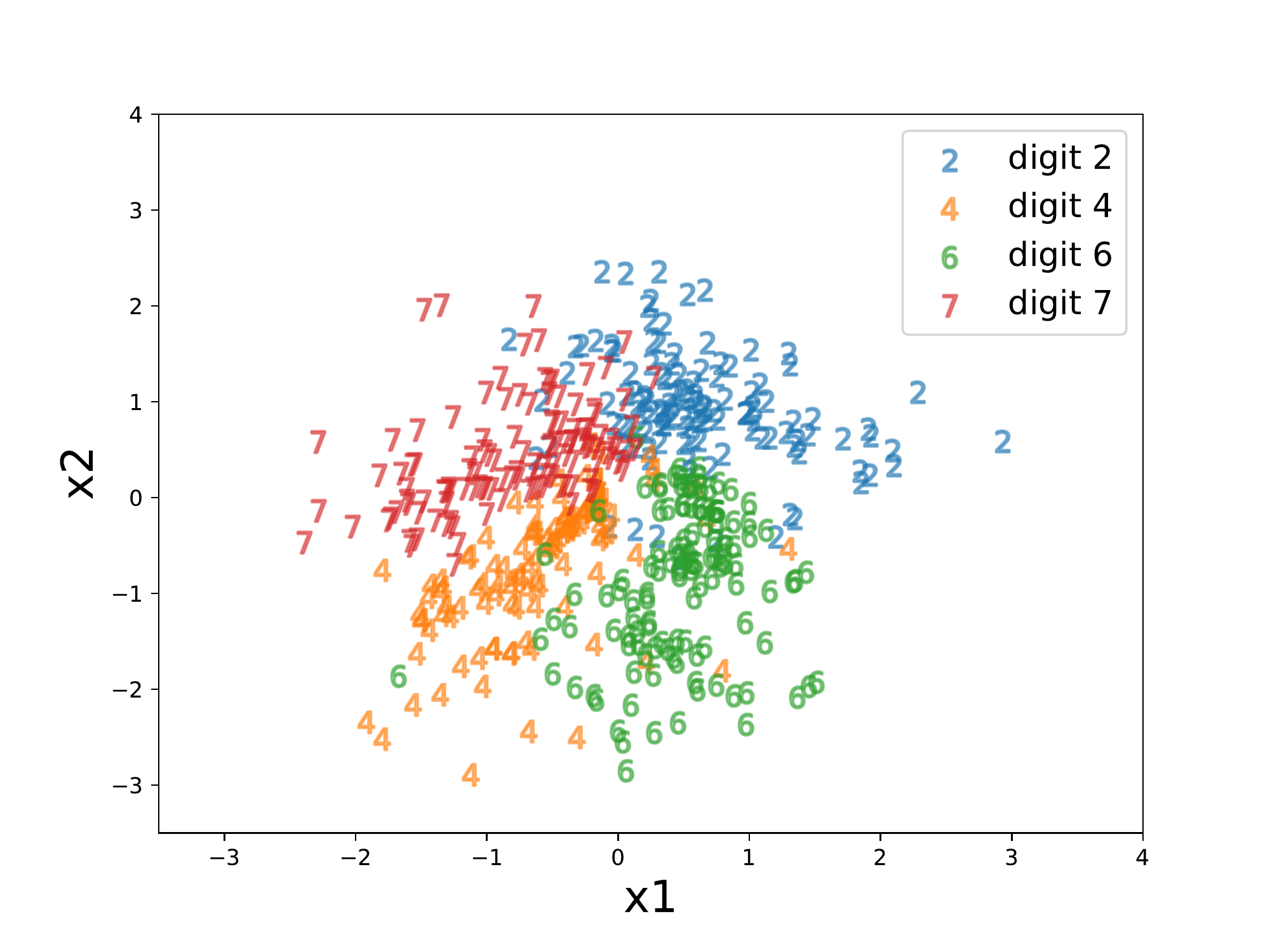}}}
    \subfloat[\centering]{{\includegraphics[height=5cm]{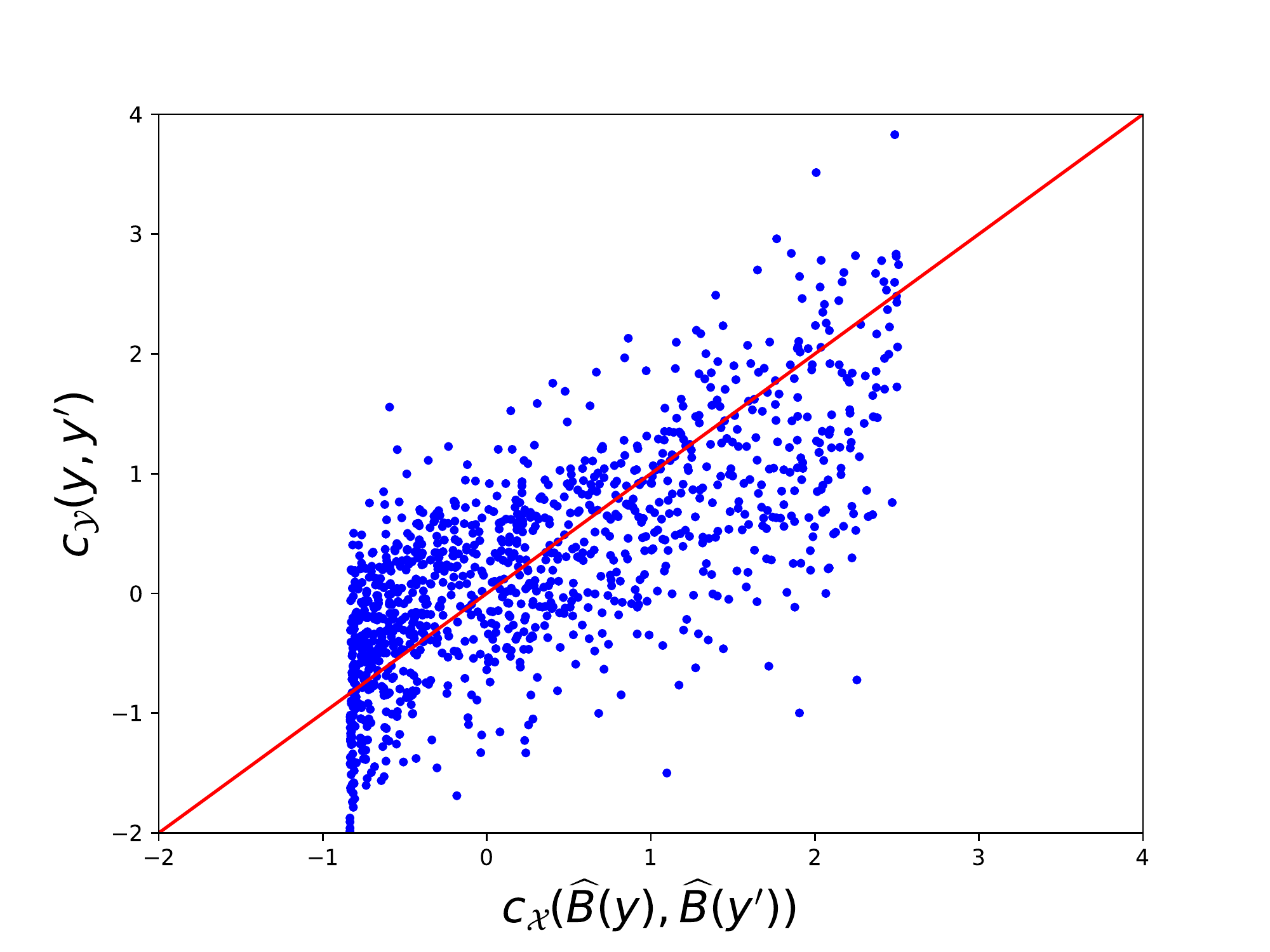}}}
    \caption{(a) is generated by applying $\widehat{B}$ to $500$ out-of-sample MNIST images, i.i.d.\ $\{\tilde{y}_j\}_{j = 1}^{500}$ from $\nu$. (b) shows the points $\{(c_\cX(\widehat{B}(\tilde{y}_{j}), \widehat{B}(\tilde{y}_{j'})), c_\cY(\tilde{y}_j, \tilde{y}_{j'}))\}_{j, j' = 1}^{50}$ and a straight line $y = x$.}
    \label{fig:MNIST-backward}
\end{figure}
We clarify that the current experiment with $\mu = N(0, I_2)$ is a proof of concept. Suppose one aims to obtain images comparable to those from dedicated MNIST generators. In that case, exhaustive tests should be done for tuning each component of the RGM sampler, which is beyond the scope of this paper. Instead, we highlight that the RGM sampler with a simple modification can indeed generate significantly improved images seen in Figure \ref{fig:MNIST-images}(b). These images are generated from the following settings that are fully introduced in Section~\ref{sec:appendix2}: $(\cX, \mu, c_{\cX}) = (\R^4, N(0, I_4), K_4)$ and MMD terms in \eqref{eqn:1} are replaced by Sinkhorn divergences \citep{genevay2018learning}. As such, the RGM sampler is amendable to other more general choices of $\cL_{\cX \times \cY}$ in its practical implementation.

\section{Discussions}
\label{sec:discussions}
In this work, we proposed the Reversible Gromov-Monge (RGM) sampler, a new variant of transform sampling based on the RGM distance operating between distributions on heterogeneous spaces. We discuss the following important aspects of the RGM sampler.

\paragraph{Inductive bias} 
Inductive bias alleviates the over-identified issue in transform sampling, as seen in Section~\ref{sec:intro}.
RGM sampler induces a bias towards finding approximate isomorphisms: namely among all possible $(F, B) \in \cI(\mu, \nu)$, RGM sampler favors the ones that $c_{\cX}(x, B(y))$ best matches $c_{\cY}(F(x), y)$ on average. In particular, if the user of the RGM sampler specifies suitable $c_\cX, c_\cY$ such that the resulting network spaces are strongly isomorphic, the RGM sampler will identify a strong isomorphism, providing a beneficial inductive bias to transform sampling. Intriguingly, borrowing insights from Brenier's polar factorization \citep{brenier1991}, we can always find $c_\cX, c_\cY$ making two complete and separable metric spaces strongly isomorphic under mild conditions. Details are provided in Section \ref{sec:brenier-polar}.

\paragraph{Optimization} We have introduced two optimization methods for computing RGM given finite samples and proved their properties. As explained in Section \ref{subsec:stat_rate}, minimization via the gradient descent algorithm is easy to implement and less complicated than the min-max optimization of GAN-type samplers. Nevertheless, it is unclear whether one can derive global convergence results from this minimization. In general, different optimization algorithms of transform samplers might output maps that favor distinctive inductive biases. Such a phenomenon is an important research direction in the optimization literature.

\section*{Acknowledgments}
Liang acknowledges the generous support from the NSF Career award (DMS-2042473), and the William Ladany faculty fellowship at the University of Chicago Booth School of Business. Liang wishes to thank Maxim Raginsky and Chris Hansen for discussions on simulation-based inference.

\bibliography{ref-minimal}
\bibliographystyle{apalike}
\newpage
\appendix

\bigskip
\begin{center}
{\large\textbf{SUPPLEMENTARY MATERIAL}}
\end{center}
This supplementary material collects details of Sections \ref{sec:summary-of-results} and \ref{sec:numerical} along with relevant discussions and technicalities. 
\begin{itemize}
	\item Section~\ref{sec:analytic-properties} studies analytic properties of the RGM distance along with the details of the inductive bias of the RGM sampler based on Brenier's polar factorization.
	\item Section~\ref{sec:statistical-theory} derives the non-asymptotic rate of convergence by analyzing the statistical properties of the RGM sampler.
	\item Section~\ref{sec:representation} discusses a further relaxation of the RGM into an infinite-dimensional convex program that relies on a new representer theorem.
	\item Section~\ref{sec:appendix} contains the proofs of the results in Section \ref{sec:statistical-theory} and auxiliary lemmas.	
	\item Section \ref{sec:comparison-with-GW} discusses computational aspects of the RGM distance.
	\item Section \ref{sec:appendix2} provides the implementation details of the experiments in Section \ref{sec:numerical}.
\end{itemize}

\section{Analytic Properties and Inductive Biases Based on Polar Factorization}
\label{sec:analytic-properties}

\subsection{Metric Properties and Some Basic Properties}
\label{sec:metric-and-basic-properties}

In this section, we derive metric properties of the proposed RGM distance. First, observe that our RGM is symmetric while the original GM is not. Next, we prove a triangle inequality using a gluing technique, as in OT.

\begin{proposition}
	\label{prop:2}
	RGM satisfies the triangle inequality, that is,
	\begin{equation*}
		\mathrm{RGM}(\mu_{\cX}, \mu_{\cZ}) \le \mathrm{RGM}(\mu_{\cX}, \mu_{\cY}) + \mathrm{RGM}(\mu_{\cY}, \mu_{\cZ})
	\end{equation*}
	holds for three network spaces $(\cX, \mu_{\cX}, c_{\cX})$, $(\cY, \mu_{\cY}, c_{\cY})$, and $(\cZ, \mu_{\cZ}, c_{\cZ})$.
\end{proposition}

\begin{proof}[Proof of Proposition~\ref{prop:2}]
	Recall that
	\begin{equation*}
		\mathrm{RGM}(\mu_{\cX}, \mu_{\cZ}) = \inf_{(F, B) \in \cI(\mu_{\cX}, \mu_{\cZ})} C_{\cX \cZ}(F, B)\;,
		\end{equation*}
	where
	\begin{equation*}
		C_{\cX \cZ}(F, B)
		= \left(\int (c_{\cX}(x, B(z)) - c_{\cZ}(F(x), z))^2 \dd{\mu_{\cX} \otimes \mu_{\cZ}}(x, z)\right)^{1/2}\;.
	\end{equation*}
	First, $(F_{\cZ} \circ F_{\cY}, B_{\cX} \circ B_{\cY}) \in \cI(\mu_{\cX}, \mu_{\cZ})$ holds for $(F_{\cY}, B_{\cX}) \in \cI(\mu_{\cX}, \mu_{\cY})$ and $(F_{\cZ}, B_{\cY}) \in \cI(\mu_{\cY}, \mu_{\cZ})$ since
	\begin{equation*}
		\begin{split}
			(\mathrm{Id}, F_{\cZ} \circ F_{\cY})_{\#}\mu_{\cX} &= (\mathrm{Id}, F_{\cZ})_{\#} (\mathrm{Id}, F_{\cY})_{\#} \mu_{\cX} \\
			&= (\mathrm{Id}, F_{\cZ})_{\#} (B_{\cX}, \mathrm{Id})_{\#} \mu_{\cY} \quad (\because (F_{\cY}, B_{\cX}) \in \cI(\mu_{\cX}, \mu_{\cY}))\\
			&= (B_{\cX}, \mathrm{Id})_{\#} (\mathrm{Id}, F_{\cZ})_{\#} \mu_{\cY} \\
			& = (B_{\cX}, \mathrm{Id})_{\#} (B_{\cY}, \mathrm{Id})_{\#} \mu_{\cZ} \quad (\because (F_{\cZ}, B_{\cY}) \in \cI(\mu_{\cY}, \mu_{\cZ})) \\
			& = (B_{\cX} \circ B_{\cY}, \mathrm{Id})_{\#} \mu_{\cZ}\;.
		\end{split}
	\end{equation*}
	\begin{equation*}
		\begin{tikzcd}[row sep = huge]
			\cX \arrow[rr, "F", shift left] \arrow[dr, "F_{\cY}", swap, shift right]
			&&
			\cZ	\arrow[ll, "B", shift left] \arrow[dl, "B_{\cY}", swap, shift right]\\
			&
			\cY \arrow[ur, "F_{\cZ}", swap, shift right] \arrow[ul, "B_{\cX}", swap, shift right]
		\end{tikzcd}
	\end{equation*}

	Moreover, in this case, we have
	\begin{equation*}
		C_{\cX \cZ}(F_{\cZ} \circ F_{\cY}, B_{\cX} \circ B_{\cY})
		\le
		C_{\cX \cY}(F_{\cY}, B_{\cX})
		+
		C_{\cY \cZ}(F_{\cZ}, B_{\cY})
	\end{equation*}
	since
	\begin{equation*}
		\begin{split}
			& C_{\cX \cZ}(F_{\cZ} \circ F_{\cY}, B_{\cX} \circ B_{\cY}) \\
			&= \left(\int \left[c_{\cX}(x, B_{\cX} \circ B_{\cY}(z)) - c_{\cZ}(F_{\cZ} \circ F_{\cY}(x), z)\right]^2 \dd{\mu_{\cX} \otimes \mu_{\cZ}}(x, z)\right)^{1/2} \\
			&\le
			\left(\int \int \left[c_{\cY}(x, B_{\cX} \circ B_{\cY}(z)) - c_{\cY}(F_{\cY}(x), B_{\cY}(z))\right]^2 \dd{\mu_{\cZ}}(z) \dd{\mu_{\cX}}(x)\right)^{1/2} \\
			&\quad  +
			\left(\int \int \left[c_{\cY}(F_{\cY}(x), B_{\cY}(z)) - c_{\cZ}(F_{\cZ} \circ F_{\cY}(x), z)\right]^2 \dd{\mu_{\cX}}(x) \dd{\mu_{\cZ}}(z)\right)^{1/2} \\
			& =
			\left(\int \int \left[c_{\cY}(x, B_{\cX}(y)) - c_{\cY}(F_{\cY}(x), y)\right]^2 \dd{\mu_{\cY}}(y) \dd{\mu_{\cX}}(x)\right)^{1/2} \, (\because (B_{\cY})_{\#} \mu_{\cZ} = \mu_{\cY})\\
			&\quad  +
			\left(\int \int \left[c_{\cY}(y, B_{\cY}(z)) - c_{\cZ}(F_{\cZ}(y), z)\right]^2 \dd{\mu_{\cY}}(y) \dd{\mu_{\cZ}}(z)\right)^{1/2} \, (\because (F_{\cY})_{\#}\mu_{\cX} = \mu_{\cY}) \\
			& =
			C_{\cX \cY}(F_{\cY}, B_{\cX})
			+
			C_{\cY \cZ}(F_{\cZ}, B_{\cY})\;.
		\end{split}
	\end{equation*}
	Hence,
	\begin{equation*}
		\begin{split}
			\mathrm{RGM}(\mu_{\cX}, \mu_{\cZ})
			&= \inf_{(F, B) \in \cI(\mu_{\cX}, \mu_{\cZ})} C_{\cX \cZ}(F, B) \\
			&\le \inf_{\substack{(F_{\cY}, B_{\cX}) \in \cI(\mu_{\cX}, \mu_{\cY}) \\ (F_{\cZ}, B_{\cY}) \in \cI(\mu_{\cY}, \mu_{\cZ})}} C_{\cX \cZ}(F_{\cZ} \circ F_{\cY}, B_{\cX} \circ B_{\cY}) \\
			&\le \inf_{(F_{\cY}, B_{\cX}) \in \cI(\mu_{\cX}, \mu_{\cY})} C_{\cX \cY}(F_{\cY}, B_{\cX}) + \inf_{(F_{\cZ}, B_{\cY}) \in \cI(\mu_{\cY}, \mu_{\cZ})} C_{\cY \cZ}(F_{\cZ}, B_{\cY}) \\
			& = \mathrm{RGM}(\mu_{\cX}, \mu_{\cY}) + \mathrm{RGM}(\mu_{\cY}, \mu_{\cZ})\;.
		\end{split}
	\end{equation*}
\end{proof}

Next, we study whether $\mathrm{RGM}(\mu, \nu) = 0$ holds if and only if $(\cX, \mu, c_{\cX}) \cong (\cY, \nu, c_{\cY})$. Here the equivalence relation induced by $\cong$ can be read from Definition~\ref{def:iso}. As in the Gromov-Wasserstein distance, in general, we can only assert the if part without further conditions. The following proposition states that $\mathrm{RGM}(\mu, \nu) = 0$ if and only if $(\cX, \mu, c_{\cX}) \cong (\cY, \nu, c_{\cY})$ under some additional conditions on $c_{\cX}$ and $c_{\cY}$, thereby implying Theorem \ref{thm:metric}.

\begin{proposition}
	\label{prop:rgm}
	Let $(\cX, \mu, c_{\cX})$ and $(\cY, \nu, c_{\cY})$ be two network spaces. If $(\cX, \mu, c_{\cX}) \cong (\cY, \nu, c_{\cY})$, then $\mathrm{RGM}(\mu, \nu) = 0$. The converse is true if there exists a continuous and strictly monotone function $h \colon \R_+ \to \R$ such that $c_{\cX} = h(d_{\cX})$ and $c_{\cY} = h(d_{\cY})$.
\end{proposition}
\begin{proof}[Proof of Proposition~\ref{prop:rgm}]
	Suppose $\mathrm{RGM}(\mu, \nu) = 0$. Due to the inequality $\mathrm{GW}(\mu, \nu) \le \mathrm{RGM}(\mu, \nu)$, we have $\mathrm{GW}(\mu, \nu) = 0$, that is, 
	\begin{equation*}
		\inf_{\gamma \in \Pi(\mu, \nu)} \left( \int_{\cX \times \cY} \int_{\cX \times \cY} (h(d_{\cX}(x, x')) - h(d_{\cY}(y, y')))^2 \dd{\gamma}(x, y) \dd{\gamma}(x', y') \right)^{1/2} = 0 \;.
	\end{equation*}
	Since there exists a coupling $\gamma^\star$ that achieves the minimum of GW due to Theorem 2.2 of \cite{chowdhury_memoli_2019}, we conclude
	\begin{equation*}
		h(d_{\cX}(x, x')) = h(d_{\cY}(y, y'))
	\end{equation*}
	holds $\gamma^\star \otimes \gamma^\star$ almost surely on $(\cX \times \cY)^2$. Since $h$ is strictly monotone, this means
	\begin{equation*}
		d_{\cX}(x, x') = d_{\cY}(y, y')
	\end{equation*}
	holds $\gamma^\star \otimes \gamma^\star$ almost surely on $(\cX \times \cY)^2$. Therefore, 
	\begin{equation*}
		\begin{split}
			& \inf_{\gamma \in \Pi(\mu, \nu)} \left( \int_{\cX \times \cY} \int_{\cX \times \cY} (d_{\cX}(x, x') - d_{\cY}(y, y'))^2 \dd{\gamma}(x, y) \dd{\gamma}(x', y') \right)^{1/2} \\
			& \ge \left( \int_{\cX \times \cY} \int_{\cX \times \cY} (d_{\cX}(x, x') - d_{\cY}(y, y'))^2 \dd{\gamma^\star}(x, y) \dd{\gamma^\star}(x', y') \right)^{1/2} \\
			& = 0 \;.
		\end{split}
	\end{equation*}
	Theorem \ref{thm:1} implies that metric measure spaces $(\cX, \mu, d_{\cX})$ and $(\cY, \nu, d_{\cY})$ are strongly isomorphic. Since $c_{\cX} = h(d_{\cX})$ and $c_{\cY} = h(d_{\cY})$, it follows easily that $(\cX, \mu, c_{\cX})$ and $(\cY, \nu, c_{\cY})$ are strongly isomorphic as well.

	To prove the if part, suppose $(\cX, \mu, c_{\cX})$ and $(\cY, \nu, c_{\cY})$ are strongly isomorphic and consider a strong isomorphism $T$. Then, $(T, T^{-1}) \in \cI(\mu, \nu)$ holds since $(\mathrm{Id}, T)_{\#} \mu = (T^{-1}, \mathrm{Id})_{\#} T_{\#} \mu = (T^{-1}, \mathrm{Id})_{\#} \nu$. Also, by definition of $T$, we have $c_{\cX}(x, T^{-1}(y)) = c_{\cY}(T(x), T \circ T^{-1}(y)) = c_{\cY}(T(x), y)$ for all $(x, y) \in \cX \times \cY$, thus
	\begin{equation*}
		\mathrm{RGM}(\mu, \nu) \le \int_{\cX \times \cY} (c_{\cX}(x, T^{-1}(y)) -  c_{\cY}(T(x), y))^2 \dd{\mu \otimes \nu} = 0\;.
	\end{equation*}
\end{proof}

We conclude this section with a few more properties and examples. We first complete the proof of Proposition~\ref{prop:1}, which characterizes the relations among three distances: GW, GM, and RGM.
\begin{proof}[Proof of Proposition~\ref{prop:1}]
	Define
	\begin{equation}\label{eq:objective}
		Q(\gamma) = \int_{\cX \times \cY} \int_{\cX \times \cY} (c_{\cX}(x, x') - c_{\cY}(y, y'))^2 \dd{\gamma}(x, y) \dd{\gamma}(x', y')
	\end{equation}
	for all $\gamma \in \Pi(\mu, \nu)$ so that
	\begin{equation*}
		\mathrm{GW}(\mu, \nu)^2 = \inf_{\gamma \in \Pi(\mu, \nu)} Q(\gamma) \;.
	\end{equation*}
	Recall that $\Pi_{\cT} = \{(\mathrm{Id}, T)_{\#} \mu : T \in \cT(\mu, \nu)\} \subset \Pi(\mu, \nu)$. Hence, as noted in Section \ref{sec:background}, 
	\begin{equation*}
		\mathrm{GM}(\mu, \nu)^2 = \inf_{\gamma \in \Pi_{\cT}} Q(\gamma) \;.
	\end{equation*}
	Define $\Pi' = \{\gamma \in \Pi(\mu, \nu) : \gamma = (\mathrm{Id}, F)_{\#} \mu = (B, \mathrm{Id})_{\#} \nu ~ \exists (F, B) \in \cI(\nu, \mu) \}$, then one can check
	\begin{equation*}
		\mathrm{RGM}(\mu, \nu)^2 = \inf_{\gamma \in \Pi'} Q(\gamma)
	\end{equation*}
	using change-of-variables. Note that $\Pi'$ may be rewritten as $\{\gamma \in \Pi_{\cT} : \gamma = (B, \mathrm{Id})_{\#} \nu ~ \exists B \in \cT(\nu, \mu) \}$. Hence, $\Pi' \subseteq \Pi_{\cT} \subseteq \Pi(\mu, \nu)$, thus we conclude $\mathrm{GW}(\mu, \nu) \le \mathrm{GM}(\mu, \nu) \le \mathrm{RGM}(\mu, \nu)$.
\end{proof}
In short, Proposition \ref{prop:1} shows that RGM, GM, and GW are minimization problems of a common objective function $Q$ over different constraint sets of couplings, namely, $\Pi' \subseteq \Pi_{\cT} \subseteq \Pi(\mu, \nu)$, respectively.

Next, we discuss the binding constraint $(\mathrm{Id}, F)_{\#} \mu = (B, \mathrm{Id})_{\#} \nu$, or equivalently, the feasible set $\cI(\mu, \nu)$. One should notice that $\cI(\mu, \nu)$ might be empty, for instance, if $\mu$ and $\nu$ are discrete and their supports have different cardinality; say, $\mu = \delta_x$ and $\nu = (\delta_{y_1} + \delta_{y_2}) / 2$, namely, Dirac measures supported on $x \in \cX$ and $y_1, y_2 \in \cY$, then even $\cT(\mu, \nu)$ is empty, meaning that the Monge problem is infeasible and so is RGM. On the flip side, the following lemma gives a sufficient condition for $(F, B) \in \cI(\mu, \nu)$ which can be useful in practice.

\begin{lemma}\label{lem:FBcondition}
	Let $(F, B) \in \cT(\mu, \nu) \times \cT(\nu, \mu)$. If $F \circ B = \mathrm{Id}$ or $B \circ F = \mathrm{Id}$ holds, then $(F, B) \in \cI(\mu, \nu)$.
\end{lemma}
\begin{proof}
	Without loss of generality, assume $B \circ F = \mathrm{Id}$. Then,
	\begin{equation*}
		(\mathrm{Id}, F)_{\#} \mu = (B \circ F, F)_{\#} \mu = (B, \mathrm{Id})_{\#} (F_{\#} \mu) = (B, \mathrm{Id})_{\#} \nu\;.
	\end{equation*}
	Hence, $(F, B) \in \cI(\mu, \nu)$.
\end{proof}

The following example illustrates that this condition can be used to find a pair $(F, B) \in \cI(\mu, \nu)$ when $\mu$ and $\nu$ are Gaussian distributions.
\begin{example}
	Given $p < q$, suppose $\mu = N(0, I_{p})$ and $\nu = N(0, \Sigma)$, where $I_{p} \in \R^{p \times p}$ is the identity matrix and $\Sigma \in \R^{q \times q}$ is of rank $p$. Then, we can find a rank-$p$ matrix $A \in \R^{q \times p}$ such that $\Sigma = A A^\top$. Let $F(x) = A x$ and $B(y) = A^\dagger y$, then one can easily check $F_{\#} \mu = \nu$, $B_{\#} \nu = \mu$, and $B \circ F = \mathrm{Id}$. Hence, $(F, B) \in \cI(\mu, \nu)$.
\end{example}

Lastly, we provide a simple example that shows that properly chosen cost functions give a strong isomorphism between two Gaussian distributions.

\begin{example}\label{ex:gauss}
	Consider two Gaussian distributions on $\R^d$, say $\mu = N(0, \Sigma_1)$ and $\nu = N(0, \Sigma_2)$. Assume $\Sigma_1$ and $\Sigma_2$ are invertible. Then two network spaces $(\R^d, \mu, c_{\cX})$ and $(\R^d, \nu, c_{\cY})$ are strongly isomorphic if $c_{\cX}$ and $c_{\cY}$ are Mahalanobis distances, that is, 
	\begin{equation*}
		c_{\cX}(x, x') = \sqrt{(x - x')^\top \Sigma_1^{- 1} (x - x')} \;, \quad c_{\cY}(y, y') = \sqrt{(y - y')^\top \Sigma_2^{- 1} (y - y')} \;.
	\end{equation*}
	To see this, let $T = \Sigma_2^{1 / 2} \Sigma_1^{- 1 / 2}$, where $\Sigma_1^{1 / 2}$ and $\Sigma_2^{1 / 2}$ are the square roots of $\Sigma_1$ and $\Sigma_2$, respectively. Obviously, a linear map $T$ satisfies $T \in \cT(\mu, \nu)$ and $c_{\cX}(x, x') = c_{\cY}(T x, T x')$ for all $x, x' \in \R^d$. According to Definition \ref{def:iso}, a linear map $T$ is a strong isomorphism. Proposition \ref{prop:rgm} implies $\mathrm{RGM}(\mu, \nu) = 0$. Notice that the same results hold for $c_{\cX}(x, x') = x^\top \Sigma_1^{-1} x'$ and $c_{\cY}(y, y') = y^\top \Sigma_2^{-1} y'$ as well.
\end{example}

\subsection{Analysis of GW = RGM in Non-atomic Cases}
\label{sec:gw=rgm}
We have seen in Proposition \ref{prop:1} that GW $\le$ GM $\le$ RGM holds in general. This section proves Theorem \ref{thm:equality}, which states that these inequalities become equalities under mild conditions.

\begin{proof}[Proof of Theorem~\ref{thm:equality}]
	First, we invoke Theorem 16 in Chapter 15 of \cite{royden1988realanalysis}, which states that any Polish probability space that has no atoms is equivalent to $([0, 1], \lambda)$, where $\lambda$ is the Lebesgue measure on $[0, 1]$. More precisely, $(\cX, \mu)$ is equivalent to $([0, 1], \lambda)$ in the following sense: there exists a bijection $\phi \colon [0, 1] \to \cX$ such that $\phi, \phi^{-1}$ are measurable, $\phi_\# \lambda = \mu$, and $(\phi^{-1})_\# \mu = \lambda$. Then, we can define the following network space $([0, 1], \lambda, \tilde{c}_{\cX})$, where 
	\begin{equation*}
		\tilde{c}_{\cX}(u, v) := c_{\cX}(\phi(u), \phi(v)) \quad \forall u, v \in [0, 1] \;.
	\end{equation*}
	By construction, $([0, 1], \lambda, \tilde{c}_{\cX})$ is strongly isomorphic to $(\cX, \mu, c_{\cX})$. Similarly, we can find a bijection $\psi \colon [0, 1] \to \cY$ such that $\psi, \psi^{-1}$ are measurable, $\psi_\# \lambda = \nu$, and $(\psi^{-1})_\# \nu = \lambda$. Then, we can also define a network space $([0, 1], \lambda, \tilde{c}_{\cY})$ that is strongly isomorphic to $(\cY, \nu, c_{\cY})$ by letting 
	\begin{equation*}
		\tilde{c}_{\cY}(u, v) := c_{\cY}(\psi(u), \psi(v)) \quad \forall u, v \in [0, 1] \;.
	\end{equation*}

	Next, we show that the RGM distance between $(\cX, \mu, c_{\cX})$ and $(\cY, \nu, c_{\cY})$ is the same as the RGM distance between $([0, 1], \lambda, \tilde{c}_{\cX})$ and $([0, 1], \lambda, \tilde{c}_{\cY})$, namely, 
	\begin{equation}
		\label{eq:RGM_equiv}
		\mathrm{RGM}(\mu, \nu) = \mathrm{RGM}((\lambda, \tilde{c}_{\cX}), (\lambda, \tilde{c}_{\cY})) \;,
	\end{equation}
	where $(\lambda, \tilde{c}_{\cX})$ and $(\lambda, \tilde{c}_{\cY})$ are placed to distinguish the two network spaces $([0, 1], \lambda, \tilde{c}_{\cX})$ and $([0, 1], \lambda, \tilde{c}_{\cY})$ defined on the same space $([0, 1], \lambda)$ with different cost functions. To see this, we use the triangle inequality of RGM (Proposition \ref{prop:2}):
	\begin{equation*}
		|\mathrm{RGM}(\mu, \nu) - \mathrm{RGM}((\lambda, \tilde{c}_{\cX}), (\lambda, \tilde{c}_{\cY}))| \le \mathrm{RGM}(\mu, (\lambda, \tilde{c}_{\cX})) + \mathrm{RGM}(\nu, (\lambda, \tilde{c}_{\cY})) \;.
	\end{equation*}
	By Proposition \ref{prop:rgm}, we have $\mathrm{RGM}(\mu, (\lambda, \tilde{c}_{\cX})) = \mathrm{RGM}(\nu, (\lambda, \tilde{c}_{\cY})) = 0$, hence we have \eqref{eq:RGM_equiv}. Similarly, one can verify that $\mathrm{GW}(\mu, \nu) = \mathrm{GW}((\lambda, \tilde{c}_{\cX}), (\lambda, \tilde{c}_{\cY}))$. 

	Now, to show $\mathrm{GW}(\mu, \nu) = \mathrm{RGM}(\mu, \nu)$, it suffices to prove 
	\begin{equation*}
		\mathrm{GW}((\lambda, \tilde{c}_{\cX}), (\lambda, \tilde{c}_{\cY})) = \mathrm{RGM}((\lambda, \tilde{c}_{\cX}), (\lambda, \tilde{c}_{\cY})) \;.
	\end{equation*}
	Equivalently, as in the proof of Proposition \ref{prop:1}, we show
	\begin{equation*}
		\inf_{\gamma \in \Pi(\lambda, \lambda)} \tilde{Q}(\gamma) = \inf_{\gamma \in \Pi'} \tilde{Q}(\gamma) \;,
	\end{equation*}
	where $\Pi' = \{\gamma \in \Pi(\lambda, \lambda) : \gamma = (\mathrm{Id}, F)_{\#} \lambda = (B, \mathrm{Id})_{\#} \lambda, ~ ~ \exists (F, B) \in \cI(\lambda, \lambda) \}$ and
	\begin{equation*}
		\tilde{Q}(\gamma) = \int_{[0, 1] \times [0, 1]} \int_{[0, 1] \times [0, 1]} (\tilde{c}_{\cX}(u, u') - \tilde{c}_{\cY}(v, v'))^2 \dd{\gamma}(u, v) \dd{\gamma}(u', v') \quad \forall \gamma \in \Pi(\lambda, \lambda) \;.
	\end{equation*}
	By Theorem 2.2 of \cite{chowdhury_memoli_2019}, there exists an optimal coupling $\gamma^\ast \in \Pi(\lambda, \lambda)$ such that $\mathrm{GW}((\lambda, \tilde{c}_{\cX}), (\lambda, \tilde{c}_{\cY})) = \tilde{Q}(\gamma^\ast)$. Next, we use the following fact from Theorem 1.1 of \cite{brenier2003}: there exists a sequence $(p_n)_{n \in \N}$ of bijective maps from $[0, 1]$ to $[0, 1]$ such that both $p_n, p_n^{-1}$ are measure-preserving, namely, $(p_n)_\# \lambda = (p_n^{-1})_\# \lambda = \lambda$ for all $n \in \N$, and $\gamma_n := (\mathrm{Id}, p_n)_{\#}\lambda$ converges weakly to $\gamma^\ast$.\footnote{In fact, Theorem 1.1 of \cite{brenier2003} is stated for the case where $\lambda$ is the Lebesgue measure restricted on $[0, 1]^d$ for $d \ge 2$. However, one can verify that the proof still applies to $d = 1$.} One can verify that 
	\begin{equation*}
		\gamma_n = (\mathrm{Id}, p_n)_\# \lambda = (\mathrm{Id}, p_n)_\# (p_n^{-1})_\#\lambda =  (p_n^{-1}, \mathrm{Id})_\# \lambda \;,
	\end{equation*}
	which shows that $\gamma_n \in \Pi'$. Accordingly, $(\gamma_n)_{n \in \N}$ is a sequence in $\Pi'$ converging weakly to $\gamma^\ast$.  As Lemma 2.3 of \cite{chowdhury_memoli_2019} shows that $\tilde{Q}$ is continuous on $\Pi(\lambda, \lambda)$ with respect to the weak topology,
	\begin{equation*}
		\mathrm{GW}((\lambda, \tilde{c}_{\cX}), (\lambda, \tilde{c}_{\cY})) = \tilde{Q}(\gamma^\ast) = \lim_{n \to \infty} \tilde{Q}(\gamma_n) \ge \inf_{\gamma \in \Pi'} \tilde{Q}(\gamma) = \mathrm{RGM}((\lambda, \tilde{c}_{\cX}), (\lambda, \tilde{c}_{\cY})) \;.
	\end{equation*}
	As $\mathrm{GW}((\lambda, \tilde{c}_{\cX}) \le \mathrm{RGM}((\lambda, \tilde{c}_{\cX}), (\lambda, \tilde{c}_{\cY}))$ by Proposition \ref{prop:1}, we have $\mathrm{GW}((\lambda, \tilde{c}_{\cX}) = \mathrm{RGM}((\lambda, \tilde{c}_{\cX}), (\lambda, \tilde{c}_{\cY}))$. Therefore, we have $\mathrm{GW}(\mu, \nu) = \mathrm{RGM}(\mu, \nu)$.
\end{proof}

\subsection{Inductive Bias of RGM and Brenier's Polar Factorization}
\label{sec:brenier-polar}

As mentioned in Section \ref{sec:discussions}, given two Polish probability spaces $(\cX, \mu)$ and $(\cY, \nu)$, we prove that there is a pair of cost functions $c_\cX$ and $c_\cY$ such that the resulting network spaces $(\cX, \mu, c_{\cX})$ and $(\cY, \nu, c_{\cY})$ are strongly isomorphic. More importantly, we discuss how these costs give rise to a specific strong isomorphism based on Brenier's polar factorization, providing a deeper insight into the inductive bias of the RGM sampler.

\paragraph{Preliminaries} Throughout this section, we focus on Borel probability measures on $\R^d$ that are absolutely continuous with respect to the $d$-dimensional Lebesgue measure and have finite second moments; $\cP_2^{ac}(\R^d)$ denotes the collection of such measures. We can always find a unique optimal transport map---given as the gradient of a convex function---between any two elements in $\cP_2^{ac}(\R^d)$; this is Brenier's theorem briefly mentioned in Section \ref{sec:intro}, which we formally state as follows.

\begin{theorem}[Brenier's Theorem]
	\label{thm:brenier}
	Given $\lambda, \rho \in \cP_2^{ac}(\R^d)$, we can find a convex function $\psi \colon \R^d \to \R$ such that $\nabla \psi$ and $(\nabla \psi)^{-1}$ are unique optimal transport maps between them:
	\begin{equation*}
		\nabla \psi = \argmin_{T_\# \lambda = \rho} \int_{\R^d} \|x - T(x)\|^2 \dd{\lambda(x)}
		\quad \text{and} \quad (\nabla \psi)^{-1} = \argmin_{T_\# \rho = \lambda} \int_{\R^d} \|x - T(x)\|^2 \dd{\rho(x)} \;.
	\end{equation*}
\end{theorem}

\begin{remark}\label{ex:brenier_gauss}
	The proof of Brenier's theorem requires intricate convex analysis techniques. It is helpful to consider a simple case---where both $\lambda$ and $\rho$ are Gaussian---to better understand the main message of Brenier's theorem. Suppose $\lambda = \cN(0, I_d)$ and $\rho = \cN(0, \Sigma)$---assuming $\Sigma$ is invertible---and focus on linear transport maps, that is, $T \in \R^{d \times d}$ such that $T T^\top = \Sigma$, for which the transport cost boils down to
	\begin{equation*}
		\E_{x \sim \cN(0, I_d)} \|(I_d - T)x\|^2 = \mathrm{tr}((I_d - T) (I_d - T)^\top) = \mathrm{tr}(I_d) - 2 \mathrm{tr}(T) + \mathrm{tr}(\Sigma) \;.
	\end{equation*}
	Therefore, the transport cost is minimized by a linear transport map $T$ that maximizes its trace under the constraint $T T^\top = \Sigma$. Using linear algebra techniques, one can verify that $\Sigma^{1/2}$, the unique square root of $\Sigma$, is the optimal linear transport map;\footnote{Though this argument is designed to show that $\Sigma^{1 / 2}$ is optimal among linear transport maps to get insights, it can be shown that $\Sigma^{1 / 2}$ is, in fact, optimal among all admissible transport maps.} we can indeed see that it is the gradient of a convex function $x \mapsto \frac{1}{2} x^\top \Sigma^{1 / 2} x$. Last but not least, notice that any admissible linear transport map $T \in \R^{d \times d}$, that satisfies $T T^\top = \Sigma$, can be decomposed as $T = \Sigma^{1 / 2} S$ for some $S \in O(d)$, where $O(d)$ is the orthogonal group in dimension $d$. The last fact is called the matrix factorization theorem, which will be elaborated in \ref{sec:polar-strong-iso} to shed light on Brenier's polar factorization, as done here to example Brenier's theorem.
\end{remark}

\subsubsection{Designing Cost Functions for Strong Isomorphism}

We have already seen in Example \ref{ex:gauss} that Mahalanobis distances make two Gaussian distributions strongly isomorphic. Though this constructive example seems to be a special case, it indicates a fundamental principle that carries over to general cases. We will elaborate on the general constructive principle in this section by referring to a common space Polish probability space $(\cZ, \lambda)$ to design cost functions. Figure \ref{fig:diagrams_cost}(a) visualizes Example \ref{ex:gauss} along with an additional base measure $\lambda := \cN(0, I_d)$;  as we have remarked after Theorem \ref{thm:brenier}, linear maps $\Sigma_1^{1 / 2}$ and $\Sigma_2^{1 / 2}$ are the optimal transport maps from $\lambda$ to $\mu$ and $\nu$, respectively.
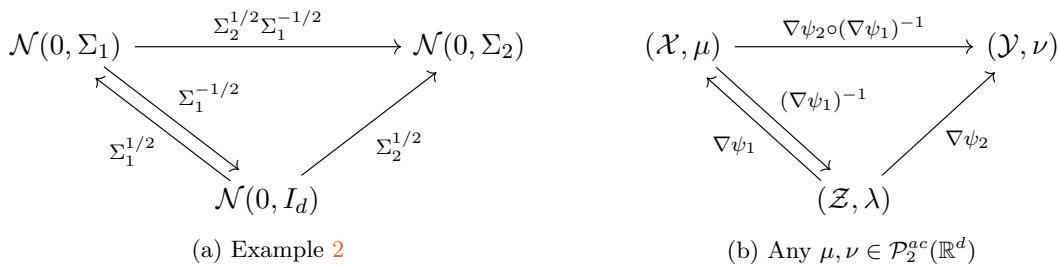
\begin{figure}[ht]
	\centering
	\tikzset{column sep=small, ampersand replacement=\&}
	\subfloat[Example \ref{ex:gauss}]{
		\begin{tikzcd}[row sep = huge]
			\cN(0, \Sigma_1) \arrow[rr, "\Sigma_2^{1 / 2} \Sigma_1^{-1 / 2}"] \arrow[dr, "\Sigma_1^{-1 / 2}", shift left]
			\&\&
			\cN(0, \Sigma_2) \\
			\&
			\cN(0, I_d) \arrow[ur, "\Sigma_2^{1 / 2}", swap] \arrow[ul, "\Sigma_1^{1 / 2}", shift left]
		\end{tikzcd}
	} \hfil
    \subfloat[Any $\mu, \nu \in \cP_2^{ac}(\R^d)$]{
		\begin{tikzcd}[row sep = huge]
			(\cX, \mu) \arrow[rr, "\nabla \psi_2 \circ (\nabla \psi_1)^{-1}"] \arrow[dr, "(\nabla \psi_1)^{-1}", shift left]
			\&\&
			(\cY, \nu) \\
			\&
			(\cZ, \lambda) \arrow[ur, "\nabla \psi_2", swap] \arrow[ul, "\nabla \psi_1", shift left]
		\end{tikzcd}
	}
	\caption{Optimal transport maps by Brenier's theorem}\label{fig:diagrams_cost}
\end{figure}

Letting $c_\cZ$ be the Euclidean distance, notice that we may rewrite the Mahalanobis distances in Example \ref{ex:gauss} as 
\begin{align*}
	c_{\cX}(x, x') = c_\cZ(\Sigma_1^{-1 / 2} x, \Sigma_1^{-1 / 2} x') \;, \\
	c_{\cY}(y, y') = c_\cZ(\Sigma_2^{-1 / 2} y, \Sigma_2^{-1 / 2} y') \;. 
\end{align*}

This shows that $c_{\cX}$ and $c_{\cY}$ are derived by properly combining the base distance $c_\cZ$ with the optimal transport maps $\Sigma_1^{-1 / 2}$ and $\Sigma_2^{-1 / 2}$, respectively. Also, in this case, $\Sigma_2^{1 / 2} \Sigma_1^{-1 / 2}$, a composition of the two optimal transport maps, gives rise to a strong isomorphism under these cost functions.

The aforementioned procedure is indeed applicable to any $\mu, \nu \in \cP_2^{ac}(\R^d)$. As visualized in Figure \ref{fig:diagrams_cost}(b), simply replace the linear maps $\Sigma_1^{1 / 2}$ and $\Sigma_2^{1 / 2}$ with the optimal transport maps $\nabla \psi_1$ and $\nabla \psi_2$ from any suitable base measure $\lambda \in \cP_2^{ac}(\R^d)$, respectively, by Brenier's theorem. Then, for any cost function $c_{\cZ}$ on $\cZ$, define
\begin{equation}
	\label{eq:cost_functions}
	\begin{aligned}
		c_{\cX}(x, x') = c_\cZ((\nabla \psi_1)^{-1}(x), (\nabla \psi_1)^{-1}(x')) \;, \\
		c_{\cY}(y, y') = c_\cZ((\nabla \psi_2)^{-1}(y), (\nabla \psi_2)^{-1}(y')) \;.
	\end{aligned}
\end{equation}
Then, the network spaces $(\cX, \mu, c_{\cX})$ and $(\cY, \nu, c_{\cY})$, where $\cX = \mathrm{supp}(\mu)$ and $\cY = \mathrm{supp}(\nu)$, are strongly isomorphic; also, $\nabla \psi_2 \circ (\nabla \psi_1)^{-1}$ is a strong isomorphism. Notice that we have derived this fundamental result by simply rethinking Example \ref{ex:gauss} via Brenier's theorem and generalizing the diagram in Figure \ref{fig:diagrams_cost}. We reiterate that the principle behind the isomorphic Gaussian example is fundamental and generalizable to generic measures in $\cP_2^{ac}(\R^d)$; it is certainly not just a toy example.

\subsubsection{Identifying Strong Isomorphism}
\label{sec:polar-strong-iso}
In the previous section, we have seen how to define suitable cost functions $c_{\cX}$ and $c_{\cY}$ that make $\mu, \nu \in \cP_2^{ac}(\R^d)$ strongly isomorphic. As pointed out in Section \ref{sec:discussions}, the RGM sampler brings in an inductive bias towards a strong isomorphism; indeed, we have shown that $\nabla \psi_2 \circ (\nabla \psi_1)^{-1}$ is a strong isomorphism under the cost functions in \eqref{eq:cost_functions}, leading to $\mathrm{RGM}(\mu, \nu) = 0$. In this section, we look at these results from a different angle using an insight from Brenier's polar factorization, highlighting unseen aspects of the inductive bias.

We start from Figure \ref{fig:diagrams_cost}(b): fix $\mu, \nu \in \cP_2^{ac}(\R^d)$ and let $\cX = \mathrm{supp}(\mu)$ and $\cY = \mathrm{supp}(\nu)$. Also, let $\lambda$ be the Lebesgue measure on $\cZ = [0, 1]^d$. Recall that the key ingredients of the diagram were optimal transport maps $\nabla \psi_1$ and $\nabla \psi_2$ from the base measure. It turns out that we may replace them with other transport maps---possibly not optimal---from the base measure. This observation comes from Brenier's polar factorization \citep{brenier1991}, which we paraphrase as follows:
\begin{theorem}[Brenier's Polar Factorization]\label{thm:brenier_polar}
	For any transport map $T_1$ from $\lambda$ to $\mu$, we can find a unique map $s_1 \colon \cZ \to \cZ$ such that $(s_1)_\# \lambda = \lambda$ and $T_1 = \nabla \psi_1 \circ s_1$.
\end{theorem}

In other words, any transport map from $\lambda$ to $\mu$ is factorized into a composition of the unique optimal transport map $\nabla \psi_1$ and a (Lebesgue) measure-preserving map $s_1 \colon \cZ \to \cZ$. Let $S(\cZ)$ be the collection of all measure-preserving maps from $\cZ$ to $\cZ$, then Brenier's polar factorization essentially shows a one-to-one correspondence between $\cT(\lambda, \mu)$ and $S(\cZ)$. Analogously, this implies a one-to-one correspondence between $\cT(\lambda, \nu)$ and $S(\cZ)$.

\begin{remark}[Matrix polar factorization]
	To get insights into this sophisticated result, let us pause and go back to the remark below Theorem \ref{thm:brenier}, where we have mentioned that in the Gaussian case, any linear transport map is decomposed as $T = \Sigma^{1 / 2} S$, the multiplication of the optimal transport map $\Sigma^{1/2}$ and an orthogonal matrix $S \in O(d)$,\footnote{Note that $S \in O(d)$ transports $\cN(0, I_d)$ to itself.} which exactly correspond to $\nabla \psi_1$ and $s_1$ in Theorem \ref{thm:brenier_polar}, respectively. In other words, Brenier's Polar Factorization is a generalization of the matrix factorization that we have discussed earlier by restricting to the Gaussian case and linear transport maps.
\end{remark}

As mentioned earlier, we now replace the two arrows---the optimal transport maps $\nabla \psi_1$ and $\nabla \psi_2$ from the base measure---in Figure \ref{fig:diagrams_cost}(b) with any transport maps, namely, $\nabla \psi_1 \circ s_1$ and $\nabla \psi_2 \circ s_2$, respectively, for any $s_1, s_2 \in S(\cZ)$. One technicality here is that we restrict our focus to bijective $s_1$ and $s_2$ to use invertibility to reverse the arrows.\footnote{Elements of $S(\cZ)$ may not be invertible in general. That said, a subset of $S(\cZ)$ consisting of bijective measure-preserving maps is dense in $S(\cZ)$ as discussed in \cite{brenier2003}.} Then, we can also define a transport map $\nabla \psi_2 \circ s_2 \circ (\nabla \psi_1 \circ s_1)^{-1}$ from $\cX$ to $\cY$ by chaining the transport maps $\mu \to \lambda$ and $\lambda \to \nu$ as in Figure \ref{fig:diagrams_cost}(b). These changes are shown in the new diagram Figure \ref{fig:diagrams}.

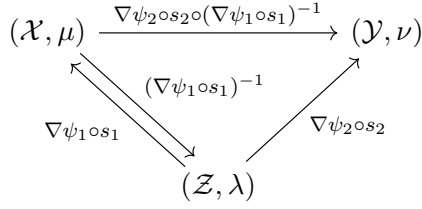
\begin{figure}[ht]
	\centering
	\tikzset{column sep=huge, ampersand replacement=\&}
	\begin{tikzcd}[row sep = huge]
		(\cX, \mu) \arrow[rr, "\nabla \psi_2 \circ s_2 \circ (\nabla \psi_1 \circ s_1)^{-1}"] \arrow[dr, "(\nabla \psi_1 \circ s_1)^{-1}", shift left]
		\&\&
		(\cY, \nu) \\
		\&
		(\cZ, \lambda) \arrow[ur, "\nabla \psi_2 \circ s_2", swap] \arrow[ul, "\nabla \psi_1 \circ s_1", shift left]
	\end{tikzcd}
	\caption{Generalization of Figure \ref{fig:diagrams_cost}(b) via Brenier's polar factorization.}\label{fig:diagrams}
\end{figure}

Now, let us go back to transform sampling. The arrow from $\mu$ to $\nu$ in Figure \ref{fig:diagrams} in fact shows that there are many transport maps from $\mu$ to $\nu$, that is,
\begin{equation*}
	\cF = \{\nabla \psi_2 \circ s_2 \circ (\nabla \psi_1 \circ s_1)^{-1} : \text{bijective} ~ s_1, s_2 \in S(\cZ)\} \subset \cT(\mu, \nu) \;.
\end{equation*}
Though $\cF$ is a collection of transport maps constructed in a certain way, that is, chaining $\mu \to \lambda$ and $\lambda \to \nu$, it still consists of infinitely many transport maps, reiterating that transform sampling is indeed over-identified. We show how the RGM sampler induces an inductive bias in this case, thereby choosing a strong isomorphism from the collection $\cF$.

Consider the cost functions $c_\cX$ and $c_\cY$ defined by \eqref{eq:cost_functions} with a cost $c_{\cZ}$. As transport maps in $\cF$ are invertible, by Lemma \ref{lem:FBcondition}, 
\begin{equation*}
	\cI' := \{(F, F^{-1}) : F \in \cF\} \subset \cI(\mu, \nu).
\end{equation*}
Recall that computing the RGM distance amounts to finding $(F, B)$ such that $c_{\cX}(x, B(y))$ best matches $c_{\cY}(F(x), y)$ on average. For $(F, B) \in \cI'$, by \eqref{eq:cost_functions}, we have
\begin{align*}
	c_\cX(x, B(y)) & = c_\cZ((\nabla \psi_1)^{-1}(x), s_1 \circ s_2^{-1} \circ (\nabla \psi_2)^{-1}(y)) \;, \\
	c_\cY(F(x), y) & = c_\cZ(s_2 \circ s_1^{-1} \circ (\nabla \psi_1)^{-1}(x), (\nabla \psi_2)^{-1}(y)) \;.
\end{align*}
Therefore, $c_\cX(x, B(y)) = c_\cY(F(x), y)$ indeed holds provided $s_1 = s_2$, which amounts to a pair $(F, B) = (\nabla \psi_2 \circ (\nabla \psi_1)^{-1}, \nabla \psi_1 \circ (\nabla \psi_2)^{-1}) \in \cI'$. In other words, among infinitely many possible elements in $\cF$ characterized by $s_1$ and $s_2$, the RGM sampler with cost functions in \eqref{eq:cost_functions} favors $\nabla \psi_2 \circ (\nabla \psi_1)^{-1}$, the strong isomorphism under these cost functions as visualized in Figure \ref{fig:diagrams_cost}(b). In summary, 
\begin{itemize}
	\item we can characterize a collection $\cF$ of transport maps by pairs of bijective measure-preserving maps $(s_1, s_2)$ via Brenier's polar factorization,
	\item the RGM sampler with cost functions in \eqref{eq:cost_functions} finds the one in $\cF$ satisfying the rearrangement correspondance $s_1 = s_2$, which is exactly the strong isomorphism (inductive bias).
\end{itemize}

\section{Statistical Theory}
\label{sec:statistical-theory}
This section serves to prove Theorem \ref{thm:stat}. Without loss of generality, we assume $\lambda_1 = \lambda_2 = \lambda_3 = 1$ in $C(\mu, \nu, F, B)$ since the proof is essentially identical with any constants $\lambda_i, 1\leq i \leq 3$. For convenience, we denote
\begin{align*}
	C_0(F, B) & = \int (c_{\cX}(x, B(y)) - c_{\cY}(F(x), y))^2 \dd{\mu \otimes \nu} \;, \\
	M(F, B) & = \mathrm{MMD}_{K_{\cY}}^2(F_{\#}\mu, \nu) + \mathrm{MMD}_{K_{\cX}}^2(\mu, B_{\#} \nu) + \mathrm{MMD}^2_{K_{\cX} \otimes K_{\cY}} ((\mathrm{Id}, F)_{\#}\mu, (B, \mathrm{Id})_{\#} \nu)
\end{align*}
and therefore $C(\mu, \nu, F, B) = C_0(F, B) + M(F, B)$. Similarly, define the empirical counterparts as
\begin{align*}
	\widehat{C}_0(F, B) & = \frac{1}{m n} \sum_{i = 1}^{m} \sum_{j = 1}^{n} (c_{\cX}(x_i, B(y_j)) - c_{\cY}(F(x_i), y_j))^2 \;, \\
	\widehat{M}(F, B) & = \mathrm{MMD}_{K_{\cY}}^2(F_{\#}\widehat{\mu}_m, \widehat{\nu}_n) + \mathrm{MMD}_{K_{\cX}}^2(\widehat{\mu}_m, B_{\#} \widehat{\nu}_n) + \mathrm{MMD}^2_{K_{\cX} \otimes K_{\cY}} ((\mathrm{Id}, F)_{\#} \widehat{\mu}_m, (B, \mathrm{Id})_{\#} \widehat{\nu}_n)
\end{align*}
and thus $C(\widehat{\mu}_m, \widehat{\nu}_n, F, B) = \widehat{C}_0(F, B) + \widehat{M}(F, B)$.

Our goal is to give an upper bound on $C(\mu, \nu, \widehat{F}, \widehat{B}) - \inf_{(F, B) \in \cF \times \cB} C(\mu, \nu, F, B) $. To this end, first recall that
\begin{equation*}
	C(\widehat{\mu}_m, \widehat{\nu}_n, \widehat{F}, \widehat{B}) \le C(\widehat{\mu}_m, \widehat{\nu}_n, F, B)
\end{equation*}
holds for any $F \in \cF$ and $B \in \cB$ by definition of $\widehat{F}$ and $\widehat{B}$ given in Theorem \ref{thm:stat}. Therefore,
\begin{equation*}
	C(\mu, \nu, \widehat{F}, \widehat{B}) - C(\mu, \nu, F, B)
	\le C(\mu, \nu, \widehat{F}, \widehat{B}) - C(\widehat{\mu}_m, \widehat{\nu}_n, \widehat{F}, \widehat{B}) + C(\widehat{\mu}_m, \widehat{\nu}_n, F, B) - C(\mu, \nu, F, B)\;.
\end{equation*}
The RHS can be decomposed as
\begin{equation*}
	C_0(\widehat{F}, \widehat{B}) - \widehat{C}_0(\widehat{F}, \widehat{B}) + M(\widehat{F}, \widehat{B}) - \widehat{M}(\widehat{F}, \widehat{B}) + \widehat{C}_0(F, B) - C_0(F, B) + \widehat{M}(F, B) - M(F, B)\;.
\end{equation*}
To further control the expression, we will first derive probabilistic bounds on $|\widehat{C}_0(F, B) - C_0(F, B)|$ and $|\widehat{M}(F, B) - M(F, B)|$ that hold for a fixed $(F, B) \in \cF \times \cB$ via standard concentration inequalities. Later, we will establish uniform probabilistic bounds on $\sup_{(F,B) \in \cF \times \cB} |\widehat{C}_0(F, B) - C_0(F, B)|$ and $\sup_{(F,B) \in \cF \times \cB} |\widehat{M}(F, B) - M(F, B)|$, using tools from empirical process theory.

\subsection{Concentration Inequalities}
We utilize the McDiarmid's inequality to derive bounds on $|\widehat{C}_0(F, B) - C_0(F, B)|$ and $|\widehat{M}(F, B) - M(F, B)|$. To give a bound on the former, we make the following boundedness assumption. 

\begin{assumption}
	\label{a:bounded1}
	$c_{\cX}(\cdot, \cdot), c_{\cY}(\cdot, \cdot)$ is uniformly bounded, that is, there exists a constant $H > 0$ such that
	\begin{equation*}
		\sup_{(x, x') \in \cX \times \cX} c_{\cX}(x, x'), \sup_{(y, y') \in \cY \times \cY} c_{\cY}(y, y') \le \sqrt{\frac{H}{4}} \;.
	\end{equation*}
\end{assumption}

\begin{proposition}
	\label{prop:4}
	Under Assumption \ref{a:bounded1}, for any pair $(F, B) \in \cF \times \cB$ and $\delta > 0$, 
	\begin{equation*}
		|\widehat{C}_0(F, B) - C_0(F, B)|
		\precsim 
		\sqrt{\frac{\log(\tfrac{m \vee n}{\delta})}{m \wedge n}}
	\end{equation*}
	holds with probability at least $1 - 4 \delta$.
\end{proposition}

To derive a similar bound on $|\widehat{M}(F, B) - M(F, B)|$, we assume that kernels are bounded.

\begin{assumption}
	\label{a:bounded_kernels}
	There exists $K > 0$ such that 
	\begin{equation*}
		\sup_{x \in \cX} |K_{\cX}(x, x)| \;, \; \sup_{y \in \cY} |K_{\cY}(y, y)| \le K \;.
	\end{equation*}
\end{assumption}

\begin{proposition}
	\label{prop:mmd_fixed}
	Under Assumption \ref{a:bounded_kernels}, for any pair $(F, B) \in \cF \times \cB$ and $\delta > 0$, 
	\begin{equation*}
		|\widehat{M}(F, B) - M(F, B)|
		\precsim \sqrt{\frac{\log(1 / \delta)}{m}} + \sqrt{\frac{\log(1 / \delta)}{n}}
	\end{equation*}
	holds with probability at least $1 - 6 \delta$.
\end{proposition}

\subsection{Uniform Deviations}
We now derive uniform deviation bounds for 
\begin{equation*}
	\sup_{(F,B) \in \cF \times \cB} |\widehat{C}_0(F, B) - C_0(F, B)| \;, \; \sup_{(F,B) \in \cF \times \cB} |\widehat{M}(F, B) - M(F, B)| \;.
\end{equation*}
For the former, we use the notion of uniform covering numbers defined below. 

\begin{definition}[Uniform Covering Number]
	Let $\cG$ be a collection of real-valued functions defined on a set $\cZ$. Given $m$ points $z_1, \ldots, z_m \in \cZ$ and any $\delta > 0$, we define $N_\infty(\delta, \cG, \{z_i\}_{i = 1}^{m})$ to be the $\delta$-covering number of $\cG$ under the pseudometric $d$ induced by points $z_1, \ldots, z_m$: 
	\begin{equation*}
		d(g, g') \coloneqq \max_{i \in [m]} |g(z_i) - g'(z_i)| \;.
	\end{equation*}
	Also, we define the uniform $\delta$-covering number of $\cG$ as follows:
	\begin{equation*}
		N_\infty(\delta, \cG, m) \coloneqq \sup\left\{N_\infty(\delta, \cG, \{z_i\}_{i = 1}^{m}) : z_1, \ldots, z_m \in \cZ\right\} \;.
	\end{equation*}
	Here, the supremum is taken over all possible combinations of $m$ points in $\cZ$.
\end{definition}

Also, we make the following assumptions.
\begin{assumption}
	\label{a:uniform_boundedness}
	$\cF_k$ and $\cB_\ell$ (see Section \ref{subsec:stat_rate}) consist of uniformly bounded functions, that is, there exists a constant $b > 0$ such that 
	\begin{equation*}
		\max_{k \in [\mathrm{dim}(\cY)]} \sup_{F_k \in \cF_k} \|F_k\|_\infty \;, \; \max_{\ell \in [\mathrm{dim}(\cX)]} \sup_{B_\ell \in \cB_\ell} \|B_\ell\|_\infty \le b \;.
	\end{equation*}
\end{assumption}

\begin{assumption}
	\label{a:lip_of_C}
	There exists a constant $L > 0$ such that 
	\begin{equation*}
		|c_{\cX}(x, x_1) - c_{\cX}(x, x_2)| \le L \|x_1 - x_2\| \;, \; |c_{\cY}(y_1, y) - c_{\cY}(y_2, y)| \le L \|y_1 - y_2\| \;.
	\end{equation*}
\end{assumption}
This Lipschitzness assumption ensures the smoothness of a map $(F, B) \mapsto |\widehat{C}_0(F, B) - C_0(F, B)|$ over $\cF \times \cB$, which allows us to utilize the uniform covering numbers.

\begin{proposition}
	\label{prop:union-bound-GW-term}
	Under Assumptions \ref{a:bounded1}, \ref{a:uniform_boundedness}, \ref{a:lip_of_C}, for any $\epsilon > 0$ and $\delta > 0$, 
	\begin{equation*}
		\begin{split}
			& \sup_{(F, B) \in \cF \times \cB} |\widehat{C}_0(F, B) - C_0(F, B)| \\
			\precsim
			& \sqrt{\frac{\log(\tfrac{m \vee n}{\delta})}{m \wedge n}} + \epsilon + \sqrt{\frac{\sum_{k = 1}^{\mathrm{dim}(\cY)} \log N_{\infty}(\epsilon, \cF_k, m) + \sum_{\ell = 1}^{\mathrm{dim}(\cX)} \log N_{\infty}(\epsilon, \cB_\ell, n)}{m \wedge n}}
		\end{split}
	\end{equation*}
	holds with probability at least $1 - 2 \delta$.
\end{proposition}

Now, the remaining task is to choose $\epsilon$ carefully in Proposition \ref{prop:union-bound-GW-term} for a concrete upper bound. To this end, we utilize the pseudo-dimension defined below.

\begin{definition}[Pseudo-Dimension]
	\label{def:pdim}
	Let $\cG$ be a collection of real-valued functions defined on a set $\cZ$. Given a subset $S\coloneqq\{z_1, \ldots, z_m\} \subset \cZ$, we say $S$ is pseudo-shattered by $\cG$ if there are $r_1, \ldots, r_m \in \R$ such that for each $b \in \{0, 1\}^m$ we can find $g_b \in \cG$ satisfying $\mathrm{sign}(g_b(z_i) - r_i) = b_i$ for all $i \in [m]$. We define the pseudo-dimension of $\cG$, denoted as $\mathrm{Pdim}(\cG)$, as the maximum cardinality of a subset $S \subset \cZ$ that is pseudo-shattered by $\cG$.
\end{definition}

Using a well-established relation of the uniform covering number and the pseudo-dimension (Lemma \ref{lem:uniform_covering_to_pdim}), we can simplify Proposition \ref{prop:union-bound-GW-term} as follows.

\begin{corollary}
	\label{cor:union-bound-GW-term-pdim}
	Under Assumptions \ref{a:bounded1}, \ref{a:uniform_boundedness}, \ref{a:lip_of_C}, for any $\delta > 0$, 
	\begin{equation*}
		\begin{split}
			& \sup_{(F, B) \in \cF \times \cB} |\widehat{C}_0(F, B) - C_0(F, B)| \\
			\precsim
			& \sqrt{\frac{\log(\tfrac{m \vee n}{\delta})}{m \wedge n}} + \sqrt{\frac{\log(m \vee n)}{m\wedge n} \left( \sum_{k =1}^{{\rm dim}(\cY)} {\rm Pdim}(\cF_{k}) + \sum_{\ell = 1}^{{\rm dim}(\cX)} {\rm Pdim}(\cB_{\ell}) \right)}
		\end{split}
	\end{equation*}
	holds with probability at least $1 - 2 \delta$.
\end{corollary}

To derive an upper bound on $\sup_{(F,B) \in \cF \times \cB} |\widehat{M}(F, B) - M(F, B)|$, we first introduce Rademacher complexities defined below.

\begin{definition}[Rademacher Complexity]
	Let $(\cZ, \rho)$ be a probability space and $\cG$ be a collection of measurable functions defined on $\cZ$. We define the Rademacher complexity of $\cG$ with respect to $m$ samples from $\rho$ as follows:
		\begin{equation*}
			R_m(\cG, \rho) =  \E_{z_i \overset{\mathrm{iid}}{\sim} \rho} \E_{\epsilon_i} \sup_{g \in \cG} \left|\frac{1}{m} \sum_{i=1}^{m} \epsilon_i g(z_i)\right|\;,
		\end{equation*}
		Here, $z_1, \ldots, z_m$ are i.i.d.\ samples from $\rho$ and $\epsilon_1, \ldots, \epsilon_m$ are i.i.d.\ Rademacher random variables such that $(z_1, \ldots, z_m)$ and $(\epsilon_1, \ldots, \epsilon_m)$ are independent.
\end{definition}

\begin{proposition}
	\label{prop:uniform-deviation-mmd-term}
	Denote a closed unit ball of any RKHS $\cH$ as $\cH(1)$. Also, let $(\mathrm{Id}, \cF) \coloneqq \{(\mathrm{Id}, F) : F \in \cF\}$ and $(\cB, \mathrm{Id}) \coloneqq \{(B, \mathrm{Id}) : B \in \cB\}$; hence, they are classes of maps from $\cX$ to $\cX \times \cY$ and from $\cY$ to $\cX \times \cY$, respectively. Under Assumption \ref{a:bounded_kernels}, for any $\delta > 0$,
	\begin{equation*}
		\begin{split}
			\sup_{(F,B) \in \cF \times \cB} |\widehat{M}(F, B) - M(F, B)|
			& \precsim \sqrt{\frac{\log(1 / \delta)}{m}} + \sqrt{\frac{\log(1 / \delta)}{n}} + R_m(\cH_{\cY}(1) \circ \cF, \mu) + R_n(\cH_{\cX}(1) \circ \cB, \nu) \\
			& \quad + R_m(\cH_{\cX \times \cY}(1) \circ (\mathrm{Id}, \cF), \mu) + R_n(\cH_{\cX \times \cY}(1) \circ (\cB, \mathrm{Id}), \nu)
		\end{split}
	\end{equation*}
	holds with probability at least $1 - 6 \delta$. Here, $\cF \circ \cG = \{f \circ g : f \in \cF, g \in \cG\}$ for any function classes $\cF$ and $\cG$ with matching input and output space.
\end{proposition}

Now, the only remaining task is to bound four Rademacher complexities. We will derive upper bounds using the chaining technique. To illustrate the main idea, let us consider $\cH_{\cY}(1) \circ \cF$. Recall that
\begin{equation*}
	R_m(\cH_{\cY}(1) \circ \cF, \mu) = \E_{x_i \overset{\mathrm{iid}}{\sim} \mu} R_m(\cH_{\cY}(1) \circ \cF, \{x_i\}_{i=1}^{m})\;,
\end{equation*}
where $R_m(\cH_{\cY}(1) \circ \cF, \{x_i\}_{i=1}^{m})$ is the empirical Rademacher complexity of $\cH_{\cY}(1) \circ F$ associated with $\{x_i\}_{i=1}^{m}$:
\begin{equation*}
	R_m(\cH_{\cY}(1) \circ \cF, \{x_i\}_{i=1}^{m}) = \E_{\epsilon_i} \sup_{h \in \cH_{\cY}(1), F \in \cF} \left|\frac{1}{m} \sum_{i=1}^{m} \epsilon_i h(F(x_i))\right| = \E_{\epsilon_i} \sup_{h \in \cH_{\cY}(1), F \in \cF} \frac{1}{m} \sum_{i=1}^{m} \epsilon_i h(F(x_i))\;.
\end{equation*}
Notice that we may remove the absolute value since $\cH_{\cY}(1) = - \cH_{\cY}(1)$. Now, considering $\{x_i\}_{i=1}^{m}$ as fixed, we will first bound the empirical Rademacher complexity by replacing the Rademacher random variables with Gaussian random variables. Let $g_i$ be i.i.d.\ standard Gaussian random variables, then it is well known that
\begin{equation*}
	R_m(\cH_{\cY}(1)\circ \cF, \{x_i\}_{i=1}^m)
	\le \sqrt{\frac{\pi}{2}} \E_{g_i} \sup_{h \in \cH_{\cY}(1), F \in \cF} \frac{1}{m} \sum_{i=1}^m  g_i h(F(x_i)) =: \sqrt{\frac{\pi}{2}} \cG_m(\cH_{\cY}(1)\circ \cF, \{x_i\}_{i=1}^m)\;.
\end{equation*}
Also, under the assumption that $K_{\cY}$ is bounded by $K$, the reproducing property and the Cauchy-Schwarz inequality imply
\begin{equation*}
	\begin{split}
		& \sup_{h \in \cH_{\cY}(1), F \in \cF} \sum_{i=1}^{m} g_i h(F(x_i)) \\
		& = \sup_{h \in \cH_{\cY}(1), F \in \cF} \left\langle h, \sum_{i=1}^{m} g_i K_{\cY}(\cdot, F(x_i)) \right\rangle_{\cH_{\cY}} \\
		& \le \sup_{h \in \cH_{\cY}(1), F \in \cF} \|h\|_{\cH_{\cY}} \left[\sum_{i = 1}^{m} g_i^2 K_{\cY}( F(x_i), F(x_i)) + \sum_{i\neq j} g_i g_j K_{\cY}(F(x_i), F(x_j)) \right]^{1/2} \\
		& \le \sup_{F \in \cF} \left[\sum_{i = 1}^{m} g_i^2 K + \sum_{i\neq j} g_i g_j K_{\cY}(F(x_i), F(x_j)) \right]^{1/2} \\
		& \le \left[\sum_{i = 1}^{m} g_i^2 K + \sup_{F \in \cF} \sum_{i\neq j} g_i g_j K_{\cY}(F(x_i), F(x_j)) \right]^{1/2} \;.
	\end{split}
\end{equation*}
Here, $\langle \cdot, \cdot \rangle_{\cH_{\cY}}$ denotes the inner product on $\cH_{\cY}$. Hence, 
\begin{align*}
	\cG_m(\cH_{\cY}(1)\circ \cF, \{x_i\}_{i=1}^m)
	& \le \frac{1}{m} \E_{g_i} \left[\sum_{i = 1}^{m} g_i^2 K + \sup_{F \in \cF} \sum_{i\neq j} g_i g_j K_{\cY}(F(x_i), F(x_j)) \right]^{1/2} \\
	& \le \frac{1}{m} \left[ m K + \E_{g_i} \sup_{F \in \cF} \sum_{i\neq j} g_i g_j K_{\cY}(F(x_i), F(x_j)) \right]^{1/2} \;,
\end{align*}
where the second inequality follows from the Jensen's inequality and $\E g_i^2 = 1$.

For any $F \colon \cX \to \cY$, let $A_F \in \R^{m \times m}$ be a matrix whose diagonal elements are zero and $(i, j)$-th element is $K_{\cY}(F(x_i), F(x_j))$ for $i \neq j$. Then, the last term amounts to the supremum of a quadratic process
\begin{align*}
	\E_{g} \sup_{F \in \cF} g^\top A_F g \;,
\end{align*}
where $g \coloneqq [g_1, \ldots, g_m]^\top \sim N(0, I_m)$.

We rely on the following chaining bound for the quadratic processes, derived in Section \ref{sec:appendix}.
\begin{lemma}[Chaining Bound]
	\label{lem:chaining}
	Let $\mathbb{S}_0^{m \times m}$ be the collection of all symmetric matrices $A$ whose diagonal elements are zero. Endow $\mathbb{S}_0^{m\times m}$ with a metric $d$ given by $d(A, A') \coloneqq \| A - A'\|$. Given $\cT \subset \mathbb{S}_0^{m \times m}$ and a fixed $A_0 \in \cT$, define $\Delta = \sup_{A \in \cT} d(A, A_0)$. Let $N(\delta, \cT)$ be the covering number of $\cT$ under the metric $d(\cdot, \cdot)$, then
	\begin{equation}
		\label{eqn:chain}
		\E_{g} \sup_{A \in \cT} g^\top A g 
		\leq \inf_{J \in \N} \left\{m \delta_J +  12 \int_{\delta_J / 2}^{\Delta/2} \sqrt{2 \log N(\delta, \cT)} \dd{\delta} + 24 \int_{\delta_J / 2}^{\Delta/2} \log N(\delta, \cT) \dd{\delta} \right\} \;,
	\end{equation}
	where for any integer $J \geq 0$, we define $\delta_J = 2^{-J} \Delta$.	
\end{lemma}

With the above chaining bound, we can directly upper bound the Rademacher complexities of the compositional classes such as $R_m(\cH_{\cY}(1)\circ \cF, \mu)$ and $R_m(\cH_{\cX \times \cY}(1) \circ (\mathrm{Id}, \cF), \mu)$. More specifically, for the former class, we will apply this chaining bound to $\cT \coloneqq \{ A_F : F \in \cF \}$. Then, to further bound the RHS of \eqref{eqn:chain}, we make the following assumptions.

\begin{assumption}
	\label{a:lip_kernels}
	Suppose $K_{\cX}$ and $K_{\cY}$ are Lipschitz: there exists $L > 0$ such that
	\begin{align*}
		|K_{\cX}(x_1, x') - K_{\cX}(x_2, x')| \le L \|x_1 - x_2\| \;, 
		\quad |K_{\cY}(y_1, y') - K_{\cY}(y_2, y')| \le L \|y_1 - y_2\| \;.
	\end{align*}
\end{assumption}
This plays a similar role as Assumption \ref{a:lip_of_C}: we can derive an upper bound on $d(A_F, A_{F'})$ via closeness of $F$ and $F'$ in $\cF$. As a result, we will see that the covering number $N(\delta, \cT)$ can be bounded by the complexity of $\cF$.

\begin{assumption}
	\label{a:separation}
	There exist $y_0$ and $y_0'$ in $\cY$ with $K_{\cY}(y_0, y_0') \neq K_{\cY}(y_0, y_0)$ such that 
	\begin{itemize}
		\item $\cF$ contains a constant map $F$ satisfying $F(x) = y_0$ for all $x \in \cX$, 
		\item whenever we have $x \neq x' \in \cX$, we can find a non-constant map $F \in \cF$ such that $F(x) = y_0$ and $F(x') = y_0'$. 
	\end{itemize}
	Similarly, there exist $x_0$ and $x_0'$ in $\cX$ with $K_{\cX}(x_0, x_0') \neq K_{\cX}(x_0, x_0)$ such that 
	\begin{itemize}
		\item $\cB$ contains a constant map $B$ such that $B(y) = x_0$ for all $y \in \cY$, 
		\item whenever we have $y \neq y' \in \cY$, we can find a non-constant map $B \in \cB$ such that $B(y) = x_0$ and $B(y') = x_0'$. 
	\end{itemize}
\end{assumption}
The main purpose of this assumption is to exclude overly restrictive $\cF$ and $\cB$, and is minimal: $\cF$ and $\cB$ should contain constant maps, as well as non-constant maps. With these assumptions, we can derive the following result.

\begin{proposition}\label{prop:Rademacher-upper-bound-by-chaining}
	Under Assumptions \ref{a:bounded_kernels}, \ref{a:uniform_boundedness}, \ref{a:lip_kernels}, \ref{a:separation},
	\begin{align*}
		R_m(\cH_{\cY}(1)\circ \cF, \mu)\;, \; R_m(\cH_{\cX \times \cY}(1) \circ (\mathrm{Id}, \cF), \mu) \precsim \sqrt{\frac{\log m}{m} \sum_{k =1}^{{\rm dim}(\cY)} {\rm Pdim}(\cF_{k})}\;, \\
		R_n(\cH_{\cX}(1)\circ \cB, \mu)\;, \; R_n(\cH_{\cX \times \cY}(1) \circ (\cB, \mathrm{Id}), \nu) \precsim \sqrt{\frac{\log n}{n} \sum_{k =1}^{{\rm dim}(\cX)} {\rm Pdim}(\cB_{k})}\;.
	\end{align*}
\end{proposition}

In summary, Propositions \ref{prop:4}, \ref{prop:mmd_fixed}, \ref{prop:uniform-deviation-mmd-term}, \ref{prop:Rademacher-upper-bound-by-chaining} and Corollary \ref{cor:union-bound-GW-term-pdim} directly imply Theorem \ref{thm:stat}.

\section{Representer Theorem and Convex Formulation}
\label{sec:representation}

This section provides details of the results presented in Section \ref{subsec:cvx-representer}. Again, without loss of generality, we only consider $\lambda_1 = \lambda_2 = \lambda_3 = 1$ in \eqref{eqn:1}.

First, we clarify how measurable maps correspond to bounded linear operators between $L^2$ spaces.
\begin{proposition}
	Let $F \colon \cX \to \cY$ be a measurable map such that $\|\dd{F_\# \pi_\cX} / \dd{\pi_\cY}\|_\infty < \infty$. If we define
	\begin{equation*}
		\bF(g) = g \circ F
	\end{equation*}
	for all $g \in L^2_{\cY}$, then $\bF \colon L^2_{\cY} \to L^2_{\cX}$ is a bounded linear operator. Similarly, a measurable map $B \colon \cY \to \cX$ satisfying $\|\dd{B_\# \pi_\cY} / \dd{\pi_\cX}\|_\infty < \infty$ induces a bounded linear operator $\bB \colon L^2_{\cX} \to L^2_{\cY}$ such that $\bB(g) = g \circ B$ for all $g \in L^2_{\cX}$.
\end{proposition}
\begin{proof}[Proof]
	Linearity of $\bF$ is obvious. Since
	\begin{equation*}
		\int_{\cX} g(F(x))^2 \dd{\pi_{\cX}} = \int_{\cY} g(y)^2 \dd{F_{\#} \pi_{\cX}}(y) = \int_{\cY} g(y)^2 \frac{\dd{F_{\#} \pi_\cX}}{\dd{\pi_\cY}}(y) \dd{\pi_{\cY}}(y) \le \|g\|_{L^2(\pi_{\cY})}^2 \left\|\frac{\dd{F_{\#} \pi_\cX}}{\dd{\pi_\cY}} \right\|_\infty\;,
	\end{equation*}
	we can see $\bF(g) = g \circ F \in L^2_{\cX}$ and thus $\bF \colon L^2_{\cY} \to L^2_{\cX}$. From this inequality, the operator norm of $\bF$ is bounded by $\|\dd{F_\# \pi_\cX} / \dd{\pi_\cY}\|_\infty^{1 / 2}$; hence, boundedness of $\bF$ follows. The same argument applies to $\bB$.
\end{proof}

Next, we prove that \eqref{eqn:1} can be written in terms of $\bF$ and $\bB$ if $K_{\cX}$ and $K_{\cY}$ are given by the Mercer's representation:
\begin{align}
	\label{eqn:k_x}
	K_{\cX}(x, x') = \sum_{k = 1}^{\infty} \lambda_k \phi_k(x) \phi_k(x')\;, \\
	\label{eqn:k_y}
	K_{\cY}(y, y') = \sum_{\ell = 1}^{\infty} \gamma_\ell \psi_\ell(y) \psi_\ell(y')\;.
\end{align}
Let $\Phi_x = [\cdots, \phi_k(x), \cdots ]^\top \in \R^\infty$ and $\Psi_y = [\cdots, \psi_\ell(y), \cdots]^\top \in \R^\infty$. Then, $K_{\cX}(x, x') = \Phi_{x}^\top \Lambda \Phi_{x'}$ and $K_{\cY}(y, y') = \Psi_{y}^\top \Gamma \Psi_{y'}$. Also, 
\begin{align*}
	K_{\cX}(x, B(y)) &= \sum_{k} \lambda_k \phi_k(x) [\phi_k \circ B] (y) \\
	& =\sum_{k} \lambda_k \phi_k(x) \bB[\phi_k] (y) \\
	& = \sum_{k, \ell} \lambda_k \phi_k(x) \bB_{\ell k} \psi_\ell(y) \\
	& = \Psi_{y}^\top \bB \Lambda \Phi_{x} \;.
\end{align*}
Analogously, we can obtain 
\begin{align*}
	& K_{\cY}(F(x), y) = \Phi_{x}^\top \bF \Gamma \Psi_{y} \;, \\
	& K_{\cX}(B(y), B(y')) = \Psi_{y}^\top \bB \Lambda \bB^\top \Psi_{y'} \;, \\
	& K_{\cY}(F(x), F(x')) = \Phi_{x}^\top \bF \Gamma \bF^\top \Phi_{x'} \;. \\
\end{align*}

Using this, we have
\begin{equation}
	\label{eqn:rgm}
	\frac{1}{m n} \sum_{i = 1}^{m} \sum_{j = 1}^{n} (K_{\cX}(x_i, B(y_j)) - K_{\cY}(F(x_i), y_j))^2 = \frac{1}{mn} \sum_{i, j} (\Psi_{y_j}^\top \bB \Lambda \Phi_{x_i} - \Phi_{x_i}^\top \bF \Gamma \Psi_{y_j})^2\;.
\end{equation}
Also, 
\begin{equation} 
	\label{eqn:mmd-x}
	\begin{split}
		\mathrm{MMD}_{K_{\cX}}^2(\widehat{\mu}_m, B_{\#} \widehat{\nu}_n) 
		& = 
		\frac{1}{m^2} \sum_{i, i'} K_{\cX}(x_i, x_{i'})	+ \frac{1}{n^2} \sum_{j, j'} K_{\cX}(B(y_j), B(y_{j'}))	- \frac{2}{mn} \sum_{i, j} K_{\cX}(x_i, B(y_j)) \\
		& = \frac{1}{m^2} \sum_{i, i'} \Phi_{x_i}^\top \Lambda \Phi_{x_i'} + \frac{1}{n^2} \sum_{j, j'} \Psi_{y_j}^\top \bB \Lambda \bB^\top \Psi_{y_{j'}} - \frac{2}{m n} \sum_{i, j} \Psi_{y_j}^\top \bB \Lambda \Phi_{x_i} \;.
	\end{split}
\end{equation}
Similarly, we have
\begin{equation}
	\label{eqn:mmd-y}
	\mathrm{MMD}_{K_{\cY}}^2(F_{\#} \widehat{\mu}_m, \widehat{\nu}_n) 
	= 
	\frac{1}{m^2} \sum_{i, i'} \Phi_{x_i}^\top \bF \Gamma \bF^\top \Phi_{x_{i'}} + \frac{1}{n^2} \sum_{j, j'} \Psi_{y_j}^\top \Gamma \Psi_{y_{j'}} - \frac{2}{m n} \sum_{i, j} \Phi_{x_i}^\top \bF \Gamma \Psi_{y_j} \;
\end{equation}
and 
\begin{equation}
	\label{eqn:mmd-xy}
	\begin{split}
		& \mathrm{MMD}_{K_{\cX} \otimes K_{\cY}}^2((\mathrm{Id}, F)_{\#} \widehat{\mu}_m, (B, \mathrm{Id})_{\#} \widehat{\nu}_n) \\
		& = \frac{1}{m^2} \sum_{i, i'} \Phi_{x_i}^\top \Lambda \Phi_{x_i'} \Phi_{x_i}^\top \bF \Gamma \bF^\top \Phi_{x_{i'}} + \frac{1}{n^2} \sum_{j, j'} \Psi_{y_j}^\top \Gamma \Psi_{y_{j'}} \Psi_{y_j}^\top \bB \Lambda \bB^\top \Psi_{y_{j'}} \\
		& - \frac{2}{m n} \sum_{i, j} \Psi_{y_j}^\top \bB \Lambda \Phi_{x_i} \Phi_{x_i}^\top \bF \Gamma \Psi_{y_j} \;.
	\end{split}
\end{equation}

The following proposition summarizes the discussion so far. 
\begin{proposition}
	\label{prop:representer}
	Given Borel measures $\pi_{\cX}$ and $\pi_{\cY}$ over $\cX$ and $\cY$, respectively, suppose their corresponding $L^2$ spaces $L^2_{\cX}$ and $L^2_{\cY}$ have countable orthonormal bases: $\{\phi_k\}_{k \in \N}$ and $\{\psi_\ell\}_{\ell \in \N}$. Also, assume $K_{\cX}$ and $K_{\cY}$ are given by the Mercer's representation \eqref{eqn:k_x} and \eqref{eqn:k_y}. Let $\cF_o$ and $\cB_o$ be collections of all $F \colon \cX \to \cY$ and $B \colon \cY \to \cX$ such that $\|\dd{F_\# \pi_\cX} / \dd{\pi_\cY}\|_\infty < \infty$ and $\|\dd{B_\# \pi_\cY} / \dd{\pi_\cX}\|_\infty < \infty$, respectively. Then, solving \eqref{eqn:1} over $\cF_o \times \cB_o$ is equivalent to \eqref{eqn:op}, where $\cC$ denotes the collection of all pairs of matrices $(\bF, \bB)$ that correspond to a pair of bounded linear operators induced by $(F, B) \in \cF_o \times \cB_o$. Also, $\Omega$ is defined as 	
	\begin{equation*}
		\begin{split}
			\Omega(\bF, \bB)
			&\coloneqq
			\frac{1}{mn} \sum_{i, j} (\Psi_{y_j}^\top \bB \Lambda \Phi_{x_i} - \Phi_{x_i}^\top \bF \Gamma \Psi_{y_j})^2 \\
			&+ 
			\frac{1}{m^2} \sum_{i, i'} \Phi_{x_i}^\top \Lambda \Phi_{x_i'} + \frac{1}{n^2} \sum_{j, j'} \Psi_{y_j}^\top \bB \Lambda \bB^\top \Psi_{y_{j'}} - \frac{2}{m n} \sum_{i, j} \Psi_{y_j}^\top \bB \Lambda \Phi_{x_i} \\
			&+ 
			\frac{1}{m^2} \sum_{i, i'} \Phi_{x_i}^\top \bF \Gamma \bF^\top \Phi_{x_{i'}} + \frac{1}{n^2} \sum_{j, j'} \Psi_{y_j}^\top \Gamma \Psi_{y_{j'}} - \frac{2}{m n} \sum_{i, j} \Phi_{x_i}^\top \bF \Gamma \Psi_{y_j} \\
			&+ 
			\frac{1}{m^2} \sum_{i, i'} \Phi_{x_i}^\top \Lambda \Phi_{x_i'} \Phi_{x_i}^\top \bF \Gamma \bF^\top \Phi_{x_{i'}} + \frac{1}{n^2} \sum_{j, j'} \Psi_{y_j}^\top \Gamma \Psi_{y_{j'}} \Psi_{y_j}^\top \bB \Lambda \bB^\top \Psi_{y_{j'}} \\
			& - \frac{2}{m n} \sum_{i, j} \Psi_{y_j}^\top \bB \Lambda \Phi_{x_i} \Phi_{x_i}^\top \bF \Gamma \Psi_{y_j} \;.
		\end{split}
	\end{equation*}
\end{proposition}
Based on this, we now prove Theorem \ref{thm:representer}.

\begin{proof}[Proof of Theorem \ref{thm:representer}]
	Let $\bx_i = \Lambda^{1 / 2} \Phi_{x_i}$ and $\by_j = \Gamma^{1 / 2} \Psi_{y_j}$. Notice that we can view them as elements of a Hilbert space $\ell^2_{\N}$, that is, the space of square-summable sequences: 
	\begin{equation*}
		\ell^2_{\N} = \{ (a_k)_{k \in \N} : \sum_{k} a_k^2 < \infty \} \;.
	\end{equation*}
	Also, define $\bar{\bF} = \Lambda^{- 1 / 2} \bF \Gamma^{1 / 2} $ and $\bar{\bB} = \Gamma^{- 1 / 2} \bB \Lambda^{1 / 2}$ where $\Lambda, \Gamma \succ 0$. By rewriting \eqref{eqn:rgm}-\eqref{eqn:mmd-xy} using $\bx_i$, $\by_j$, $\bar{\bF}$, and $\bar{\bB}$, we have
	\begin{align*}
		& \Omega(\bF, \bB) \\
		& = \frac{1}{mn} \sum_{i, j} (\by_j^\top \bar{\bB} \bx_i - \bx_i^\top \bar{\bF} \by_j)^2 \tag{i} \\
		& + \frac{1}{m^2} \sum_{i, i'} \bx_i^\top \bx_{i'} + \frac{1}{n^2} \sum_{j, j'} \by_{j}^\top \bar{\bB} \bar{\bB}^\top \by_{j'} - \frac{2}{m n} \sum_{i, j} \by_j^\top \bar{\bB} \bx_i \tag{ii} \\
		& + \frac{1}{m^2} \sum_{i, i'} \bx_i^\top \bar{\bF} \bar{\bF}^\top \bx_{i'} + \frac{1}{n^2} \sum_{j, j'} \by_j^\top \by_{j'} - \frac{2}{m n} \sum_{i, j} \bx_i^\top \bar{\bF} \by_j \tag{iii} \\
		& + \frac{1}{m^2} \sum_{i, i'} (\bx_i^\top \bx_{i'}) (\bx_i^\top \bar{\bF} \bar{\bF}^\top \bx_{i'}) + \frac{1}{n^2} \sum_{j, j'} (\by_j^\top \by_{j'}) (\by_{j}^\top \bar{\bB} \bar{\bB}^\top \by_{j'}) - \frac{2}{m n} \sum_{i, j} (\by_j^\top \bar{\bB} \bx_i) (\bx_i^\top \bar{\bF} \by_j) \tag{iv} \\
		& =: \bar{\Omega}(\bar{\bF}, \bar{\bB}) \;.
	\end{align*}
	As a result, \eqref{eqn:relaxed} reduces to $\min_{\bar{\bF}, \bar{\bB} \in \R^{\infty \times \infty}} \bar{\Omega}(\bar{\bF}, \bar{\bB})$.	Now, we define two finite-dimensional subspaces of $\ell^2_{\N}$ spanned by $(\bx_1, \ldots, \bx_m)$ and $(\by_1, \ldots, \by_n)$, respectively: 
	\begin{equation*}
		U_m \coloneqq \mathrm{span}\{\bx_1, \ldots, \bx_m\} \;, \quad 
		V_n \coloneqq \mathrm{span}\{\by_1, \ldots, \by_n\} \;.
	\end{equation*}
	Also, we define $P_{U_m}$ and $P_{V_n}$ to be matrices that correspond to the orthogonal projection operators from $\ell^2_{\N}$ to $U_m$ and to $V_n$, respectively. Recall that $P_{U_m}$ and $P_{V_n}$ are symmetric and idempotent by definition.

	Our goal is to prove
	\begin{equation*}
		\bar{\Omega}(\bar{\bF}, \bar{\bB})
		\ge
		\bar{\Omega}(P_{U_m} \bar{\bF} P_{V_n}, P_{V_n} \bar{\bB} P_{U_m})\;.
	\end{equation*}
	More precisely, we show that four terms (i)-(iv) decrease if we replace $\bar{\bF}$ and $\bar{\bB}$ with $P_{U_m} \bar{\bF} P_{V_n}$ and $P_{V_n} \bar{\bB} P_{U_m}$, respectively. First, observe that (i) remains the same. By definition, $P_{U_m} \bx_i = \bx_i$ and $P_{V_n} \by_j = \by_j$, thus $\by_j^\top \bar{\bB} \bx_i = \by_j^\top P_{V_n} \bar{\bB} P_{U_m} \bx_i$ and $\bx_i^\top \bar{\bF} \by_j = \bx_i^\top P_{U_m} \bar{\bF} P_{V_n} \by_j$. Hence, (i) does not change.

	To prove that (ii) decrease, it suffices to prove
	\begin{equation*}
		\sum_{j, j'} \by_{j}^\top \bar{\bB} \bar{\bB}^\top \by_{j'} 
		\ge
		\sum_{j, j'} \by_{j}^\top (P_{V_n} \bar{\bB} P_{U_m}) (P_{V_n} \bar{\bB} P_{U_m})^\top \by_{j'}
		=
		\sum_{j, j'} \by_{j}^\top \bar{\bB} P_{U_m} \bar{\bB}^\top \by_{j'} \;.
	\end{equation*}
	To this end, define $P_{U_m}^\perp$ to be a matrix that corresponds to the orthogonal projection from $\ell^2_{\N}$ to $U_m^\perp$, the orthogonal complement of $U_m$. By definition, $P_{U_m} + P_{U_m}^\perp$ is the identity matrix and $P_{U_m} P_{U_m}^\perp = 0$. Hence,
	\begin{equation*}
		\sum_{j, j'} \by_{j}^\top \bar{\bB} \bar{\bB}^\top \by_{j'} 
		=
		\left\|\bar{\bB}^\top \sum_{j} \by_{j}\right\|^2
		=
		\left\|P_{U_m} \bar{\bB}^\top \sum_{j} \by_{j}\right\|^2 + \left\|P_{U_m}^\perp \bar{\bB}^\top \sum_{j} \by_{j}\right\|^2
		\ge
		\sum_{j, j'} \by_{j}^\top \bar{\bB} P_{U_m} \bar{\bB}^\top \by_{j'}
	\end{equation*}
	Here, the second equality is the Pythagorean theorem. Therefore, we can see (ii) decreases if we replace $\bar{\bB}$ with $P_{V_n} \bar{\bB} P_{U_m}$. Similarly, (iii) decreases.

	For (iv), it suffices to prove
	\begin{equation*}
		\begin{split}
			\sum_{j, j'} (\by_j^\top \by_{j'}) (\by_{j}^\top \bar{\bB} \bar{\bB}^\top \by_{j'}) 
			&\ge
			\sum_{j, j'} (\by_j^\top \by_{j'}) \left(\by_{j}^\top (P_{V_n} \bar{\bB} P_{U_m}) (P_{V_n} \bar{\bB} P_{U_m})^\top \by_{j'}\right) \\
			&=
			\sum_{j, j'} (\by_j^\top \by_{j'}) (\by_{j}^\top \bar{\bB} P_{U_m} \bar{\bB}^\top \by_{j'}) \;.
		\end{split}
	\end{equation*}
	To see this, 
	\begin{equation*}
		\begin{split}
			\sum_{j, j'} (\by_j^\top \by_{j'}) (\by_{j}^\top \bar{\bB} \bar{\bB}^\top \by_{j'}) 
			& =
			\sum_{j, j'} (\by_j^\top \by_{j'}) (\by_{j}^\top \bar{\bB} P_{U_m} \bar{\bB}^\top \by_{j'}) + 
			\sum_{j, j'} (\by_j^\top \by_{j'}) (\by_{j}^\top \bar{\bB} P_{U_m}^\perp \bar{\bB}^\top \by_{j'}) \\
			& \ge \sum_{j, j'} (\by_j^\top \by_{j'}) (\by_{j}^\top \bar{\bB} P_{U_m} \bar{\bB}^\top \by_{j'}) \;,
		\end{split}
	\end{equation*}
	where the inequality holds since 
	\begin{equation*}
		\sum_{j, j'} (\by_j^\top \by_{j'}) (\by_{j}^\top \bar{\bB} P_{U_m}^\perp \bar{\bB}^\top \by_{j'}) = \mathrm{Tr} \left[\left(\sum_{j} P_{U_m}^\perp \bar{\bB}^\top \by_j \by_j^\top\right) \left(\sum_{j} P_{U_m}^\perp \bar{\bB}^\top \by_j \by_j^\top\right)^\top \right] \ge 0 \;.
	\end{equation*}
	Similarly, we can obtain
	\begin{equation*}
		\sum_{i, i'} (\bx_i^\top \bx_{i'}) (\bx_i^\top \bar{\bF} \bar{\bF}^\top \bx_{i'})
		\ge 
		\sum_{i, i'} (\bx_i^\top \bx_{i'}) (\bx_i^\top \bar{\bF} P_{V_n} \bar{\bF}^\top \bx_{i'}).
	\end{equation*}
	Hence, (iv) decreases.	
	
	Consequently, we have
	\begin{equation*}
		\eqref{eqn:relaxed}
		=
		\min_{\bar{\bF}, \bar{\bB} \in \R^{\infty \times \infty}} \bar{\Omega}(\bar{\bF}, \bar{\bB})
		=
		\min_{\bar{\bF}, \bar{\bB} \in \R^{\infty \times \infty}} \bar{\Omega}(P_{U_m} \bar{\bF} P_{V_n}, P_{V_n} \bar{\bB} P_{U_m}) \;.
	\end{equation*}
	By definition of a projection operator, we can find $\mathsf{U}_m \in \R^{m \times \infty}$ and $\mathsf{V}_n \in \R^{n \times \infty}$ such that 
	\begin{equation*}
		P_{U_m} = \Lambda^{1 / 2} \Phi_m \mathsf{U}_m \;, \quad 
		P_{V_n} = \Gamma^{1 / 2} \Psi_n \mathsf{V}_n.
	\end{equation*}
	By letting $\mathsf{U}_m \bar{\bF} \mathsf{V}_n^\top = \mathsf{F}_{m, n} \in \R^{m \times n}$ and $\mathsf{V}_n \bar{\bB} \mathsf{U}_m^\top = \mathsf{B}_{n, m} \in \R^{n \times m}$, we have
	\begin{align*}
		P_{U_m} \bar{\bF} P_{V_n} = \Lambda^{1/2} \Phi_m \mathsf{F}_{m,n} \Psi_n^\top \Gamma^{1/2} \;, \\
		P_{V_n} \bar{\bB} P_{U_m} = \Gamma^{1/2} \Psi_n \mathsf{B}_{n,m} \Phi_m^\top \Lambda^{1/2} \;.
	\end{align*}
	Hence, 
	\begin{equation*}
		\min_{\bar{\bF}, \bar{\bB} \in \R^{\infty \times \infty}} \bar{\Omega}(P_{U_m} \bar{\bF} P_{V_n}, P_{V_n} \bar{\bB} P_{U_m}) 
		=
		\min_{(\mathsf{F}_{m, n}, \mathsf{B}_{n, m}) \in \mathsf{C}} \omega(\mathsf{F}_{m, n}, \mathsf{B}_{n, m}) \;,
	\end{equation*}
	where	
	\begin{equation*}
		\omega(\mathsf{F}_{m, n}, \mathsf{B}_{n, m})
		\coloneqq \bar{\Omega}(\Lambda^{1 / 2} \Phi_m \mathsf{F}_{m, n} \Psi_n^\top \Gamma^{1 / 2}, \Gamma^{1 / 2} \Psi_n \mathsf{B}_{n, m} \Phi_m^\top \Lambda^{1 / 2}) \;.
	\end{equation*} 
	Here, $\mathsf{C}$ is a constraint set implying that $\mathsf{F}_{m, n}$ and $\mathsf{B}_{n, m}$ are associated with $\bar{\bF}$ and $\bar{\bB}$, respectively, namely,
	\begin{equation*}
		\mathsf{C} = \{(\mathsf{U}_m \bar{\bF} \mathsf{V}_n^\top,  \mathsf{V}_n \bar{\bB} \mathsf{U}_m^\top) : \bar{\bF}, \bar{\bB} \in \R^{\infty \times \infty}\} \subset \R^{m \times n} \times \R^{n \times m} \;.
	\end{equation*}
	Therefore,
	\begin{equation*}
		\eqref{eqn:relaxed}
		=
		\min_{(\mathsf{F}_{m, n}, \mathsf{B}_{n, m}) \in \mathsf{C}} \omega(\mathsf{F}_{m, n}, \mathsf{B}_{n, m})
		\ge
		\min_{\substack{\mathsf{F}_{m, n} \in \R^{m \times n} \\ \mathsf{B}_{n, m} \in \R^{n \times m}}} \omega(\mathsf{F}_{m, n}, \mathsf{B}_{n, m}) \;.
	\end{equation*}
	Finally, note that $\mathsf{C} = \R^{m \times n} \times \R^{n \times m}$ if $\mathsf{U}_m$ and $\mathsf{V}_n$ are full rank, that is, row spaces of $\mathsf{U}_m$ and $\mathsf{V}_n$ are rank-$m$ and rank-$n$, respectively. This is true if kernel matrices 
	\begin{equation*}
		\bK_{\cX} = (\Lambda^{1 / 2} \Phi_m)^\top (\Lambda^{1 / 2} \Phi_m) \;, \; \bK_{\cY} = (\Gamma^{1 / 2} \Psi_n)^\top (\Gamma^{1 / 2} \Psi_n)
	\end{equation*}
	are invertible. This is equivalent to say that they are positive definite. In this case, 
	\begin{equation*}
		\mathsf{U}_m = \bK_{\cX}^{- 1} (\Lambda^{1 / 2} \Phi_m)^\top \;, \;
		\mathsf{V}_n = \bK_{\cY}^{- 1} (\Gamma^{1 / 2} \Psi_n)^\top\;,
	\end{equation*}
	which are indeed full rank. Accordingly, we have 
	\begin{equation*}
		\eqref{eqn:relaxed}
		=
		\min_{(\mathsf{F}_{m, n}, \mathsf{B}_{n, m}) \in \mathsf{C}} \omega(\mathsf{F}_{m, n}, \mathsf{B}_{n, m})
		=
		\min_{\substack{\mathsf{F}_{m, n} \in \R^{m \times n} \\ \mathsf{B}_{n, m} \in \R^{n \times m}}} \omega(\mathsf{F}_{m, n}, \mathsf{B}_{n, m}) \;.
	\end{equation*}
	Finally, we prove $\omega$ is convex. To see this, verify
	\begin{align*}
		\omega(\mathsf{F}_{m, n}, \mathsf{B}_{n, m})
		& = \frac{1}{mn} \| \bK_{\cY} \mathsf{B}_{n, m} \bK_{\cX} - \bK_{\cY} \mathsf{F}_{m, n}^\top \bK_{\cX} \|^2 \\
		& + \left\|\bK_{\cX}^{1/2} \cdot \left(\frac{1}{m} \mathbf{1}_m - \mathsf{B}_{n, m}^\top \bK_{\cY} \frac{1}{n} \mathbf{1}_n \right) \right\|^2
		+ \left\|\bK_{\cY}^{1/2} \cdot \left(\frac{1}{n} \mathbf{1}_n - \mathsf{F}_{m, n}^\top \bK_{\cX} \frac{1}{m} \mathbf{1}_m \right) \right\|^2 \\
		& + \left\| \frac{1}{m} \bK_{\cX}^{3/2} \mathsf{F}_{m,n} \bK_{\cY}^{1/2} - \frac{1}{n} \bK_{\cX}^{1/2} \mathsf{B}_{n,m}^\top \bK_{\cY}^{3/2} \right\|^2 \;,
	\end{align*}
	where $\bK_{\cX}^{1 / 2}$ and $\bK_{\cY}^{1 / 2}$ are the square root matrices of $\bK_{\cX}$ and $\bK_{\cY}$, respectively, and $\mathbf{1}_m \in \R^m$ and $\mathbf{1}_n \in \R^n$ are all-ones vectors. 
\end{proof}

\section{Supporting Proofs of Section \ref{sec:statistical-theory}}
\label{sec:appendix}

\subsection{Proofs in Section~\ref{sec:statistical-theory}}

\begin{proof}[Proof of Proposition~\ref{prop:4}]
	Let $h_{F, B}(x, y) \coloneqq (c_{\cX}(x, B(y)) - c_{\cY}(F(x), y))^2$, then
	\begin{equation*}
		\begin{split}
			\widehat{C}_0(F, B) - C_0(F, B)
			& = \frac{1}{m n} \sum_{i=1}^{m} \sum_{j=1}^{n} h_{F, B}(x_i, y_j) - \E_{(x, y) \sim \mu \otimes \nu} h_{F, B}(x, y) \\
			& = \frac{1}{m} \sum_{i=1}^{m} \left(\frac{1}{n} \sum_{j=1}^{n} h_{F, B}(x_i, y_j) - \E_{y \sim \nu} h_{F, B}(x_i, y)\right) \\
			& \quad + \frac{1}{m} \sum_{i=1}^{m} \E_{y \sim \nu} h_{F, B}(x_i, y) - \E_{(x, y) \sim \mu \otimes \nu} h_{F, B}(x, y)\;.
		\end{split}
	\end{equation*}
	Assumption \ref{a:bounded1} implies that a function $x \mapsto \E_{y \sim \nu} h_{F, B}(x, y)$ is bounded in $[0, H]$. Thus, by the McDiarmid's inequality,
	\begin{equation*}
		\frac{1}{m} \sum_{i=1}^{m} \E_{y \sim \nu} h_{F, B}(x_i, y) - \E_{(x, y) \sim \mu \otimes \nu} h_{F, B}(x, y) \le \sqrt{\frac{H^2 \log(1 / \delta)}{2 m}}
	\end{equation*}
	holds with probability at least $1 - \delta$. By the same logic, for fixed $x_i$,
	\begin{equation*}
		\frac{1}{n} \sum_{j=1}^{n} h_{F, B}(x_i, y_j) - \E_{y \sim \nu} h_{F, B}(x_i, y) \le \sqrt{\frac{H^2 \log(1 / \delta)}{2 n}}
	\end{equation*}
	holds with probability at least $1 - \delta$, where the probability is the conditional probability of $y_1, \ldots, y_n$ given $x_1, \ldots, x_m$. Since this is true for all $x_i$, the union bound implies
	\begin{equation*}
		\frac{1}{m} \sum_{i=1}^{m} \left(\frac{1}{n} \sum_{j=1}^{n} h_{F, B}(x_i, y_j) - \E_{y \sim \nu} h_{F, B}(x_i, y)\right) \le \sqrt{\frac{H^2 \log(m / \delta)}{2 n}}
	\end{equation*}
	holds with probability at least $1 - \delta$. Hence,
	\begin{equation*}			
		\widehat{C}_0(F, B) - C_0(F, B)
		\precsim \sqrt{\frac{\log(m / \delta)}{m}} \le \sqrt{\frac{\log(\tfrac{m \vee n}{\delta})}{m \wedge n}}
	\end{equation*}
	holds with probability at least $1 - 2 \delta$. The same result holds for $C_0(F, B) - \widehat{C}_0(F, B)$, hence we complete the proof.
\end{proof}

\begin{proof}[Proof of Proposition~\ref{prop:mmd_fixed}]
	By the triangle inequality, $|\widehat{M}(F, B) - M(F, B)|$ is bounded above by the sum of the following three terms:
	\begin{align*}
		& |\mathrm{MMD}^2_{K_{\cY}}(F_{\#}\widehat{\mu}_m, \widehat{\nu}_n) - \mathrm{MMD}^2_{K_{\cY}}(F_{\#}\mu, \nu)| \;, \\
		& |\mathrm{MMD}^2_{K_{\cX}}(\widehat{\mu}_m, B_{\#}\widehat{\nu}_n) - \mathrm{MMD}^2_{K_{\cX}}(\mu, B_{\#} \nu)| \;, \\
		& |\mathrm{MMD}^2_{K_{\cX} \otimes K_{\cY}} ((\mathrm{Id}, F)_{\#}\widehat{\mu}_m, (B, \mathrm{Id})_{\#} \widehat{\nu}_n) - \mathrm{MMD}^2_{K_{\cX} \otimes K_{\cY}} ((\mathrm{Id}, F)_{\#}\mu, (B, \mathrm{Id})_{\#} \nu)| \;.
	\end{align*}
	First, we give an upper bound on the first term. Boundedness of kernels (Assumption \ref{a:bounded_kernels}) implies
	\begin{equation*}
		\mathrm{MMD}_{K_{\cY}}(F_{\#}\widehat{\mu}_m, \widehat{\nu}_n) \;, \; \mathrm{MMD}_{K_{\cY}}(F_{\#}\mu, \nu) \le 2 \sqrt{K} \;.
	\end{equation*}
	Hence,
	\begin{equation*}
		|\mathrm{MMD}^2_{K_{\cY}}(F_{\#}\widehat{\mu}_m, \widehat{\nu}_n) - \mathrm{MMD}^2_{K_{\cY}}(F_{\#}\mu, \nu)| \le 4 \sqrt{K} |\mathrm{MMD}_{K_{\cY}}(F_{\#}\widehat{\mu}_m, \widehat{\nu}_n) - \mathrm{MMD}_{K_{\cY}}(F_{\#}\mu, \nu)|\;.
	\end{equation*}
	Due to the triangle inequality of MMD, we have
	\begin{equation*}
		|\mathrm{MMD}_{K_{\cY}}(F_{\#}\widehat{\mu}_m, \widehat{\nu}_n) - \mathrm{MMD}_{K_{\cY}}(F_{\#}\mu, \nu)|
		\le \mathrm{MMD}_{K_{\cY}}(F_{\#}\widehat{\mu}_{m}, F_{\#}\mu) + \mathrm{MMD}_{K_{\cY}}(\widehat{\nu}_n, \nu)\;.
	\end{equation*}
	By Theorem 3.4 of \cite{muandet_fukumizu_sriperumbudur_scholkopf_2017},
	\begin{equation*}
		\mathrm{MMD}_{K_{\cY}}(\widehat{\nu}_n, \nu) \le \sqrt{\frac{K}{n}} + \sqrt{\frac{2 K \log(1/\delta)}{n}}
	\end{equation*}
	holds with probability at least $1 - \delta$. Next, note that $F_{\#}\widehat{\mu}_m = \frac{1}{m} \sum_{i} \delta_{F(x_i)}$ is the empirical measure constructed from $\{F(x_i)\}_{i = 1}^{m}$. Since they are $m$ many i.i.d.\ samples from $F_{\#}\mu$, by the same theorem, 
	\begin{equation*}
		\mathrm{MMD}_{K_{\cY}}(F_{\#}\widehat{\mu}_m, F_{\#}\mu) \le \sqrt{\frac{K}{m}} + \sqrt{\frac{2 K \log(1/\delta)}{m}}
	\end{equation*}
	holds with probability at least $1 - \delta$. Hence,
	\begin{equation*}
		|\mathrm{MMD}^2_{K_{\cY}}(F_{\#}\widehat{\mu}_m, \widehat{\nu}_n) - \mathrm{MMD}^2_{K_{\cY}}(F_{\#}\mu, \nu)| \precsim \sqrt{\frac{\log(1 / \delta)}{m}} + \sqrt{\frac{\log(1 / \delta)}{n}}
	\end{equation*}
	holds with probability at least $1 - 2 \delta$. Similarly, we have
	\begin{align*}
		& |\mathrm{MMD}^2_{K_{\cX}}(\widehat{\mu}_m, B_{\#}\widehat{\nu}_n) - \mathrm{MMD}^2_{K_{\cX}}(\mu, B_{\#} \nu)| \precsim \sqrt{\frac{\log(1 / \delta)}{m}} + \sqrt{\frac{\log(1 / \delta)}{n}}\;, \\
		& |\mathrm{MMD}^2_{K_{\cX} \otimes K_{\cY}} ((\mathrm{Id}, F)_{\#}\widehat{\mu}_m, (B, \mathrm{Id})_{\#} \widehat{\nu}_n) - \mathrm{MMD}^2_{K_{\cX} \otimes K_{\cY}} ((\mathrm{Id}, F)_{\#}\mu, (B, \mathrm{Id})_{\#} \nu)| \\ 
		& \precsim \sqrt{\frac{\log(1 / \delta)}{m}} + \sqrt{\frac{\log(1 / \delta)}{n}}\;,
	\end{align*}
	each of which holds with probability at least $1 - 2 \delta$. Combining these three probabilistic bounds, we obtain a bound for $|\widehat{M}(F, B) - M(F, B)|$.
\end{proof}

\begin{proof}[Proof of Proposition~\ref{prop:union-bound-GW-term}]
	Without loss of generality, assume $n \ge m$. From the proof of Proposition \ref{prop:4}, 
	\begin{equation*}
		\begin{split}
			\sup_{(F, B) \in \cF \times \cB} |\widehat{C}_0(F, B) - C_0(F, B)|
			& \le \frac{1}{m} \sum_{i=1}^{m} 
			\sup_{(F, B) \in \cF \times \cB} \left|\frac{1}{n} \sum_{j=1}^{n} h_{F, B}(x_i, y_j) - \E_{y \sim \nu} h_{F, B}(x_i, y)\right| \\
			& \quad + 
			\sup_{(F, B) \in \cF \times \cB} \left|\frac{1}{m} \sum_{i=1}^{m} \E_{y \sim \nu} h_{F, B}(x_i, y) - \E_{x \sim \mu} \E_{y \sim \nu} h_{F, B}(x, y)\right| \;.
		\end{split}
	\end{equation*}
	Since $x \mapsto \E_{y \sim \nu} h_{F, B}(x, y)$ is bounded in $[0, H]$, Lemma \ref{lem:bound} implies
	\begin{equation*}
		\begin{split}
			& \sup_{(F, B) \in \cF \times \cB} \left|\frac{1}{m} \sum_{i=1}^{m} \E_{y \sim \nu} h_{F, B}(x_i, y) - \E_{x \sim \mu} \E_{y \sim \nu} h_{F, B}(x, y)\right| \\
			& \precsim
			\sqrt{\frac{\log(1 / \delta)}{m}} + 
			\E_{x_i} \E_{\epsilon_i} \sup_{(F, B) \in \cF \times \cB} \left|\frac{1}{m} \sum_{i=1}^{m} \epsilon_i \E_{y \sim \nu} h_{F, B}(x_i, y) \right|
		\end{split}
	\end{equation*}
	holds with probability at least $1 - \delta$. Since $n \ge m$, 
	\begin{equation*}
		\begin{split}
			\E_{x_i} \E_{\epsilon_i} \sup_{(F, B) \in \cF \times \cB} \left|\frac{1}{m} \sum_{i=1}^{m} \epsilon_i \E_{y \sim \nu} h_{F, B}(x_i, y) \right| 
			& = 
			\E_{x_i} \E_{\epsilon_i} \sup_{(F, B) \in \cF \times \cB} \left|\E_{y_1, \ldots, y_m} \frac{1}{m} \sum_{i=1}^{m} \epsilon_i h_{F, B}(x_i, y_i) \right| \\
			& \le 
			\E_{x_i} \E_{\epsilon_i} \E_{y_1, \ldots, y_m} \sup_{(F, B) \in \cF \times \cB} \left|\frac{1}{m} \sum_{i=1}^{m} \epsilon_i h_{F, B}(x_i, y_i) \right| \\
			& =
			\E_{x_i} \E_{y_i} \E_{\epsilon_i} \sup_{(F, B) \in \cF \times \cB} \left|\frac{1}{m} \sum_{i=1}^{m} \epsilon_i h_{F, B}(x_i, y_i) \right| \;.
		\end{split}
	\end{equation*}
	We first give an upper bound on 
	\begin{equation*}
		\E_{\epsilon_i} \sup_{(F, B) \in \cF \times \cB} \underbrace{\left|\frac{1}{m} \sum_{i=1}^{m} \epsilon_i h_{F, B}(x_i, y_i) \right|}_{=: X_{F, B}} \;.
	\end{equation*}
	First, observe that Assumption \ref{a:bounded1} and Assumption \ref{a:lip_of_C} imply 
	\begin{equation*}
		\begin{split}
			|h_{F, B}(x, y) - h_{F', B'}(x, y)|
			& \le
			\left| \sqrt{h_{F, B}(x, y)} + \sqrt{h_{F', B'}(x, y)} \right| \left| \sqrt{h_{F, B}(x, y)} - \sqrt{h_{F', B'}(x, y)} \right| \\
			& \le 
			2 \sqrt{H} \left( |c_{\cX}(x, B(y)) - c_{\cX}(x, B'(y))| + |c_{\cY}(F(x), y) - c_{\cY}(F'(x), y)| \right) \\
			& \le
			2 \sqrt{H} L \left( \|F(x) - F'(x)\| + \|B(y) - B'(y)\| \right) \\
			\\
			& =
			2 \sqrt{H} L \left[ \sqrt{\sum_{k = 1}^{\mathrm{dim}(\cY)} |F_k(x) - F_k'(x)|^2} + \sqrt{\sum_{\ell = 1}^{\mathrm{dim}(\cX)} |B_\ell(y) - B_\ell'(y)|^2} \right]\\
			& \le 
			2 \sqrt{H} L \left(\sum_{k = 1}^{\mathrm{dim}(\cY)} |F_k(x) - F_k'(x)| + \sum_{\ell = 1}^{\mathrm{dim}(\cX)} |B_\ell(y) - B_\ell'(y)|\right) \;. 
		\end{split}
	\end{equation*}
	Therefore, 
	\begin{equation*}
		\begin{split}
			|X_{F, B} - X_{F', B'}| 
			& \le \frac{1}{m} \sum_{i = 1}^{m} |h_{F, B}(x_i, y_i) - h_{F', B'}(x_i, y_i)| \\
			& \le 
			2 \sqrt{H} L \left(\sum_{k = 1}^{\mathrm{dim}(\cY)} \frac{1}{m} \sum_{i = 1}^{m} |F_k(x_i) - F_k'(x_i)| + \sum_{\ell = 1}^{\mathrm{dim}(\cX)} \frac{1}{m} \sum_{i = 1}^{m} |B_\ell(y_i) - B_\ell'(y_i)|\right) \\
			& \le
			2 \sqrt{H} L \left(\sum_{k = 1}^{\mathrm{dim}(\cY)} \max_{i \in [m]} |F_k(x_i) - F_k'(x_i)| + \sum_{\ell = 1}^{\mathrm{dim}(\cX)} \max_{i \in [m]} |B_\ell(y_i) - B_\ell'(y_i)|\right) \\
			& =: \rho((F, B), (F', B')) \;.
		\end{split}
	\end{equation*}
	For $\epsilon > 0$, let $\cN_\infty(\epsilon, \cF_k, \{x_i\}_{i=1}^{m})$ be the minimal $\epsilon$-covering net of $\cF_k$ under the pseudometric $d$ induced by $x_1, \ldots, x_m$:
	\begin{equation*}
		d(F_k, F_k') \coloneqq \max_{i \in [m]} | F_k(x_i) - F_k(x_i') |\;.
	\end{equation*}
	In other words, for any $F_k \in \cF_k$, we can find $F_k' \in \cN_\infty(\epsilon, \cF_k, \{x_i\}_{i=1}^{m})$ such that $d(F_k, F_k') \le \epsilon$. Also, $|\cN_\infty(\epsilon, \cF_k, \{x_i\}_{i=1}^{m})| = N_\infty(\epsilon, \cF_k, \{x_i\}_{i=1}^{m})$. We define $\cN_\infty(\epsilon, \cB_\ell, \{y_i\}_{i=1}^{m})$ in a similar fashion.
	
	Given $\epsilon > 0$, let $T_{\epsilon} = \otimes_{k = 1}^{\mathrm{dim}(\cY)} \cN_\infty(\epsilon, \cF_k, \{x_i\}_{i=1}^{m}) \times \otimes_{\ell = 1}^{\mathrm{dim}(\cX)} \cN_\infty(\epsilon, \cB_\ell, \{y_i\}_{i=1}^{m})$. Then, for any $(F, B) \in \cF \times \cB$, we can find $(F', B') \in T_{\epsilon}$ such that 
	\begin{equation*}
		\rho((F, B), (F', B')) \le \eta \epsilon \;,
	\end{equation*}
	where $\eta = 2 \sqrt{H} L (\mathrm{dim}(\cX) + \mathrm{dim}(\cY))$. As a result, one can easily check
	\begin{equation*}
		\begin{split}
			\sup_{(F, B) \in \cF \times \cB} X_{F, B}
			& \le \sup_{\rho((F, B), (F', B')) \le \eta \epsilon} |X_{F, B} - X_{F', B'}| + \sup_{(F, B) \in T_\epsilon} X_{F, B} \\
			& \le \eta \epsilon + \sup_{(F, B) \in T_\epsilon} X_{F, B}\;.
		\end{split}
	\end{equation*}
	Note that $X_{F, B}$ is the absolute value of a sub-Gaussian random variable with parameter $H^2 / (4 m)$. Hence, the maximal inequality yields
	\begin{equation*}
		\begin{split}
			\E_{\epsilon_i} \sup_{(F, B) \in \cF \times \cB} X_{F, B}
			& \le \eta \epsilon + \E_{\epsilon_i} \sup_{(F, B) \in T_\epsilon} X_{F, B}
			\le \eta \epsilon + \sqrt{\frac{H^2 \log(|T_\epsilon|)}{m}}
			\;.
		\end{split}
	\end{equation*}
	Using $|T_\epsilon| = \prod_{k = 1}^{\mathrm{dim}(\cY)} N_{\infty}(\epsilon, \cF_k, \{x_i\}_{i=1}^{m}) \times \prod_{\ell = 1}^{\mathrm{dim}(\cX)} N_\infty(\epsilon, \cB_\ell, \{y_i\}_{i=1}^{m})$, we have
	\begin{equation*}
		\begin{split}
			& \E_{\epsilon_i} \sup_{(F, B) \in \cF \times \cB} \left|\frac{1}{m} \sum_{i=1}^{m} \epsilon_i h_{F, B}(x_i, y_i) \right| \\
			& \le 
			\eta \epsilon +  H \sqrt{\frac{\sum_{k = 1}^{\mathrm{dim}(\cY)} \log N_{\infty}(\epsilon, \cF_k, \{x_i\}_{i=1}^{m}) + \sum_{\ell = 1}^{\mathrm{dim}(\cX)} \log N_{\infty}(\epsilon, \cB_\ell, \{y_i\}_{i=1}^{m})}{m}} \\
			& \le \eta \epsilon +  H \sqrt{\frac{\sum_{k = 1}^{\mathrm{dim}(\cY)} \log N_{\infty}(\epsilon, \cF_k, m) + \sum_{\ell = 1}^{\mathrm{dim}(\cX)} \log N_{\infty}(\epsilon, \cB_\ell, m)}{m}} \;.
		\end{split}
	\end{equation*}
	The second inequality is obvious from the definition of the uniform covering number. Since the last equation is independent of $x_i$ and $y_i$, we have 
	\begin{equation*}
		\begin{split}
			& \E_{x_i} \E_{y_i} \E_{\epsilon_i} \sup_{(F, B) \in \cF \times \cB} \left|\frac{1}{m} \sum_{i=1}^{m} \epsilon_i h_{F, B}(x_i, y_i) \right| \\
			& \le 
			\eta \epsilon +  H \sqrt{\frac{\sum_{k = 1}^{\mathrm{dim}(\cY)} \log N_{\infty}(\epsilon, \cF_k, m) + \sum_{\ell = 1}^{\mathrm{dim}(\cX)} \log N_{\infty}(\epsilon, \cB_\ell, m)}{m}} \;.
		\end{split}
	\end{equation*}
	As a result,
	\begin{equation}
		\label{eqn:b1}
		\begin{split}
			& \sup_{(F, B) \in \cF \times \cB} \left|\frac{1}{m} \sum_{i=1}^{m} \E_{y \sim \nu} h_{F, B}(x_i, y) - \E_{x \sim \mu} \E_{y \sim \nu} h_{F, B}(x, y)\right| \\
			& \precsim
			\sqrt{\frac{\log(1 / \delta)}{m}} + \epsilon +  \sqrt{\frac{\sum_{k = 1}^{\mathrm{dim}(\cY)} \log N_{\infty}(\epsilon, \cF_k, m) + \sum_{\ell = 1}^{\mathrm{dim}(\cX)} \log N_{\infty}(\epsilon, \cB_\ell, m)}{m}} \\
			& \le
			\sqrt{\frac{\log(1 / \delta)}{m}} + \epsilon +  \sqrt{\frac{\sum_{k = 1}^{\mathrm{dim}(\cY)} \log N_{\infty}(\epsilon, \cF_k, m) + \sum_{\ell = 1}^{\mathrm{dim}(\cX)} \log N_{\infty}(\epsilon, \cB_\ell, n)}{m}}
		\end{split}
	\end{equation}
	holds with probability at least $1 - \delta$. Here, $\log N_{\infty}(\epsilon, \cB_\ell, m) \le \log N_{\infty}(\epsilon, \cB_\ell, n)$ holds since $n \ge m$, which is obvious from the definition of the uniform covering number. 

	Next, we give a bound on 
	\begin{equation*}
		\frac{1}{m} \sum_{i=1}^{m} \sup_{(F, B) \in \cF \times \cB} \left|\frac{1}{n} \sum_{j=1}^{n} h_{F, B}(x_i, y_j) - \E_{y \sim \nu} h_{F, B}(x_i, y)\right| \;.
	\end{equation*}
	Considering $x_1, \ldots, x_m$ are fixed, Lemma \ref{lem:bound} implies
	\begin{equation*}
		\begin{split}
			& \sup_{(F, B) \in \cF \times \cB} \left|\frac{1}{n} \sum_{j=1}^{n} h_{F, B}(x_i, y_j) - \E_{y \sim \nu} h_{F, B}(x_i, y)\right| \\
			& \precsim \sqrt{\frac{\log(1 / \delta)}{n}} + \E_{y_j} \E_{\epsilon_j} \sup_{(F, B) \in \cF \times \cB} \underbrace{\left|\frac{1}{n} \sum_{j=1}^{n} \epsilon_j h_{F, B}(x_i, y_j) \right|}_{Y_{F, B}} \\
		\end{split}
	\end{equation*}
	holds with probability at least $1 - \delta$. Here, the probability should be understood as a conditional probability of $y_1, \ldots, y_n$ given $x_1, \ldots, x_m$. Again, we have
	\begin{equation*}
		\begin{split}
			|Y_{F, B} - Y_{F', B'}| 
			& \le
			2 \sqrt{H} L \left(\sum_{k = 1}^{\mathrm{dim}(\cY)} |F_k(x_i) - F_k'(x_i)| + \sum_{\ell = 1}^{\mathrm{dim}(\cX)} \frac{1}{n} \sum_{j = 1}^{n} |B_\ell(y_j) - B_\ell'(y_j)|\right) \\
			& \le
			2 \sqrt{H} L \left(\sum_{k = 1}^{\mathrm{dim}(\cY)} \max_{i \in [m]} |F_k(x_i) - F_k'(x_i)| + \sum_{\ell = 1}^{\mathrm{dim}(\cX)} \max_{j \in [n]} |B_\ell(y_j) - B_\ell'(y_j)|\right) \;.
		\end{split}
	\end{equation*}
	Also, $Y_{F, B}$ is the absolute value of a sub-Gaussian random variable with parameter $H^2 / (4 n)$. By the same argument as before, 
	\begin{equation*}
		\begin{split}
			& \E_{y_j} \E_{\epsilon_j} \sup_{(F, B) \in \cF \times \cB} \left|\frac{1}{n} \sum_{j=1}^{n} \epsilon_j h_{F, B}(x_i, y_j) \right| \\
			& \le \eta \epsilon +  H \sqrt{\frac{\sum_{k = 1}^{\mathrm{dim}(\cY)} \log N_{\infty}(\epsilon, \cF_k, m) + \sum_{\ell = 1}^{\mathrm{dim}(\cX)} \log N_{\infty}(\epsilon, \cB_\ell, n)}{n}} \;.
		\end{split}
	\end{equation*}
	Hence, 
	\begin{equation*}
		\begin{split}
			& \sup_{(F, B) \in \cF \times \cB} \left|\frac{1}{n} \sum_{j=1}^{n} h_{F, B}(x_i, y_j) - \E_{y \sim \nu} h_{F, B}(x_i, y)\right| \\
			& \precsim \sqrt{\frac{\log(1 / \delta)}{n}} + \epsilon + \sqrt{\frac{\sum_{k = 1}^{\mathrm{dim}(\cY)} \log N_{\infty}(\epsilon, \cF_k, m) + \sum_{\ell = 1}^{\mathrm{dim}(\cX)} \log N_{\infty}(\epsilon, \cB_\ell, n)}{n}} 
		\end{split}
	\end{equation*}
	holds with probability (conditional probability as explained earlier) at least $1 - \delta$. Since this holds for all $x_i$, the union bound implies
	\begin{equation}
		\label{eqn:b2}
		\begin{split}
			& \frac{1}{m} \sum_{i=1}^{m} \sup_{(F, B) \in \cF \times \cB} \left|\frac{1}{n} \sum_{j=1}^{n} h_{F, B}(x_i, y_j) - \E_{y \sim \nu} h_{F, B}(x_i, y)\right| \\
			& \precsim 
			\sqrt{\frac{\log(m / \delta)}{n}} + \epsilon + \sqrt{\frac{\sum_{k = 1}^{\mathrm{dim}(\cY)} \log N_{\infty}(\epsilon, \cF_k, m) + \sum_{\ell = 1}^{\mathrm{dim}(\cX)} \log N_{\infty}(\epsilon, \cB_\ell, n)}{n}} \\
			& \le \sqrt{\frac{\log(m / \delta)}{m}} + \epsilon + \sqrt{\frac{\sum_{k = 1}^{\mathrm{dim}(\cY)} \log N_{\infty}(\epsilon, \cF_k, m) + \sum_{\ell = 1}^{\mathrm{dim}(\cX)} \log N_{\infty}(\epsilon, \cB_\ell, n)}{m}} 
		\end{split}
	\end{equation}
	holds with probability at least $1 - \delta$. Combining \eqref{eqn:b1} and \eqref{eqn:b2}, for any $\epsilon > 0$, we have
	\begin{equation*}
		\begin{split}
			& \sup_{(F, B) \in \cF \times \cB} |\widehat{C}_0(F, B) - C_0(F, B)| \\
			\precsim
			& \sqrt{\frac{\log(\tfrac{m \vee n}{\delta})}{m \wedge n}} + \epsilon + \sqrt{\frac{\sum_{k = 1}^{\mathrm{dim}(\cY)} \log N_{\infty}(\epsilon, \cF_k, m) + \sum_{\ell = 1}^{\mathrm{dim}(\cX)} \log N_{\infty}(\epsilon, \cB_\ell, n)}{m \wedge n}}
		\end{split}
	\end{equation*}
	holds with probability at least $1 - 2 \delta$.
\end{proof}
\begin{proof}[Proof of Corollary~\ref{cor:union-bound-GW-term-pdim}]
	Combining Assumption \ref{a:uniform_boundedness} and Lemma \ref{lem:uniform_covering_to_pdim}, we have
	\begin{equation*}
		N_\infty(\epsilon, \cF_k, m) \le \left( \frac{2 e m b}{\epsilon \cdot {\rm Pdim}(\cF_{k})} \right)^{{\rm Pdim}(\cF_{k})} \;.
	\end{equation*}
	Hence,
	\begin{equation*}
		\begin{split}
			& \sup_{(F, B) \in \cF \times \cB} |\widehat{C}_0(F, B) - C_0(F, B)| \\
			\precsim
			& \sqrt{\frac{\log(\tfrac{m \vee n}{\delta})}{m \wedge n}} + \epsilon + \sqrt{\frac{\sum_{k = 1}^{\mathrm{dim}(\cY)} \log N_{\infty}(\epsilon, \cF_k, m) + \sum_{\ell = 1}^{\mathrm{dim}(\cX)} \log N_{\infty}(\epsilon, \cB_\ell, n)}{m \wedge n}} \\
			& \le \sqrt{\frac{\log(\tfrac{m \vee n}{\delta})}{m \wedge n}} + \epsilon + \sqrt{\frac{\log\left(\tfrac{2e b (m \vee n)}{\epsilon}\right)}{m \wedge n} \left( \sum_{k =1}^{{\rm dim}(\cY)} {\rm Pdim}(\cF_{k}) + \sum_{\ell = 1}^{{\rm dim}(\cX)} {\rm Pdim}(\cB_{\ell})\right)} \\
			& \precsim \sqrt{\frac{\log(\tfrac{m \vee n}{\delta})}{m \wedge n}} + \sqrt{\frac{\log(m \vee n)}{m\wedge n} \left( \sum_{k =1}^{{\rm dim}(\cY)} {\rm Pdim}(\cF_{k}) + \sum_{\ell = 1}^{{\rm dim}(\cX)} {\rm Pdim}(\cB_{\ell}) \right)} 
		\end{split}
	\end{equation*}
	holds with probability at least $1 - 2 \delta$, where the last bound comes from choosing $\epsilon = (m \wedge n)^{- 1 / 2}$.
\end{proof}

\begin{proof}[Proof of Proposition~\ref{prop:uniform-deviation-mmd-term}]
	Using the triangle inequality, we bound $\sup_{(F,B) \in \cF \times \cB} |\widehat{M}(F, B) - M(F, B)|$ by the sum of the following three terms:
	\begin{align*}
		& \sup_{F \in \cF} |\mathrm{MMD}^2_{K_{\cY}}(F_{\#}\widehat{\mu}_m, \widehat{\nu}_n) - \mathrm{MMD}^2_{K_{\cY}}(F_{\#}\mu, \nu)| \;, \\
		& \sup_{B \in \cB} |\mathrm{MMD}^2_{K_{\cX}}(\widehat{\mu}_m, B_{\#}\widehat{\nu}_n) - \mathrm{MMD}^2_{K_{\cX}}(\mu, B_{\#} \nu)| \;, \\
		& \sup_{(F,B) \in \cF \times \cB} |\mathrm{MMD}^2_{K_{\cX} \otimes K_{\cY}} ((\mathrm{Id}, F)_{\#}\widehat{\mu}_m, (B, \mathrm{Id})_{\#} \widehat{\nu}_n) - \mathrm{MMD}^2_{K_{\cX} \otimes K_{\cY}} ((\mathrm{Id}, F)_{\#}\mu, (B, \mathrm{Id})_{\#} \nu)| \;.
	\end{align*}
	As in the proof of Proposition \ref{prop:mmd_fixed}, we have
	\begin{equation*}
		\begin{split}
			& \sup_{F \in \cF} |\mathrm{MMD}^2_{K_{\cY}}(F_{\#}\widehat{\mu}_m, \widehat{\nu}_n) - \mathrm{MMD}^2_{K_{\cY}}(F_{\#}\mu, \nu)| \\
			& \le 
			4 \sqrt{K} \left[\sup_{F \in \cF} \mathrm{MMD}_{K_{\cY}}(F_{\#}\widehat{\mu}_{m}, F_{\#}\mu)  + \mathrm{MMD}_{K_{\cY}}(\widehat{\nu}_n, \nu) \right] \;.
		\end{split}
	\end{equation*}
	$\mathrm{MMD}_{K_{\cY}}(\widehat{\nu}_n, \nu)$ has already been bounded in Proposition \ref{prop:mmd_fixed}. For the first term on the RHS, observe that
	\begin{equation*}
		\begin{split}
			\sup_{F \in \cF} \mathrm{MMD}_{K_{\cY}}(F_{\#}\widehat{\mu}_{m}, F_{\#}\mu)
			& = \sup_{F \in \cF} \sup_{f \in \cH_{\cY}(1)} \left|\int f \dd{F_{\#} \widehat{\mu}_{m}} - \int f \dd{F_{\#} \mu} \right| \\
			& = \sup_{F \in \cF} \sup_{f \in \cH_{\cY}(1)} \left|\int f \circ F \dd{\widehat{\mu}_{m}} - \int f \circ F \dd{\mu} \right| \\
			& = \sup_{f \in \cH_{\cY}(1) \circ \cF} \left|\int f \dd{\widehat{\mu}_{m}} - \int f \dd{\mu} \right|\;,
		\end{split}
	\end{equation*}
	where the second equality follows from change-of-variables.

	First, we show $\cH_{\cY}(1)$ consists of $\sqrt{K}$-uniformly bounded functions. Let $\|\cdot\|_{\cH_{\cY}}$ be the norm of $\cH_{\cY}$ so that $f \in \cH_{\cY}(1)$ is equivalent to $\|f\|_{\cH_{\cY}} \le 1$. Then, the reproducing property implies
	\begin{equation*}
		|f(y)| \le \|f\|_{\cH_{\cY}} \sqrt{K_{\cY}(y, y)} \le \sqrt{K} 
	\end{equation*}
	for any $f \in \cH_{\cY}(1)$. Accordingly, $\cH_{\cY}(1) \circ \cF$ also consists of $\sqrt{K}$-uniformly bounded functions. Hence, Lemma \ref{lem:bound} implies that
	\begin{equation*}
		\begin{split}
			\sup_{F \in \cF} \mathrm{MMD}_{K_{\cY}}(F_{\#}\widehat{\mu}_{m}, F_{\#}\mu)
			&= \sup_{f \in \cH_{\cY}(1) \circ \cF} \left|\int f \dd{\widehat{\mu}_{m}} - \int f \dd{\mu} \right| \\
			&\le 2 R_m(\cH_{\cY}(1) \circ \cF, \mu) + \sqrt{\frac{2 K \log(1/\delta)}{m}}
		\end{split}
	\end{equation*}
	holds with probability at least $1 - \delta$. Therefore, combining this with the upper bound on $\mathrm{MMD}_{K_{\cY}}(\widehat{\nu}_n, \nu)$ derived in Proposition \ref{prop:mmd_fixed},
	\begin{equation}
		\label{eqn:m1}
		\sup_{F \in \cF} |\mathrm{MMD}^2_{K_{\cY}}(F_{\#}\widehat{\mu}_m, \widehat{\nu}_n) - \mathrm{MMD}^2_{K_{\cY}}(F_{\#}\mu, \nu)|
		\precsim 
		R_m(\cH_{\cY}(1) \circ \cF, \mu) + \sqrt{\frac{\log(1/\delta)}{m}} + \sqrt{\frac{\log(1/\delta)}{n}}
	\end{equation}
	holds with probability at least $1 - 2 \delta$. Similarly, we can prove that
	\begin{equation}
		\label{eqn:m2} 
		\sup_{B \in \cB} |\mathrm{MMD}^2_{K_{\cX}}(\widehat{\mu}_m, B_{\#}\widehat{\nu}_n) - \mathrm{MMD}^2_{K_{\cX}}(\mu, B_{\#} \nu)| 
		\precsim 
		R_n(\cH_{\cX}(1) \circ \cB, \nu) + \sqrt{\frac{\log(1/\delta)}{m}} + \sqrt{\frac{\log(1/\delta)}{n}}
	\end{equation}
	holds with probability at least $1 - 2 \delta$. 
	
	Lastly, since $K_{\cX} \otimes K_{\cY}$ is bounded by $K^2$, that is, 
	\begin{equation*}
		\sup_{(x, y), (x', y') \in \cX \times \cY} K_{\cX} \otimes K_{\cY}((x, y), (x', y')) \le K^2,
	\end{equation*}
	by the same argument, we have
	\begin{equation*}
		\begin{split}
			& \sup_{(F,B) \in \cF \times \cB} |\mathrm{MMD}^2_{K_{\cX} \otimes K_{\cY}} ((\mathrm{Id}, F)_{\#}\widehat{\mu}_m, (B, \mathrm{Id})_{\#} \widehat{\nu}_n) - \mathrm{MMD}^2_{K_{\cX} \otimes K_{\cY}} ((\mathrm{Id}, F)_{\#}\mu, (B, \mathrm{Id})_{\#} \nu)| \\
			&\quad \le 4 K \left[\sup_{F \in \cF} \mathrm{MMD}_{K_{\cX} \otimes K_{\cY}} ((\mathrm{Id}, F)_{\#}\widehat{\mu}_m, (\mathrm{Id}, F)_{\#}\mu) + \sup_{B \in \cB} \mathrm{MMD}_{K_{\cX} \otimes K_{\cY}} ((B, \mathrm{Id})_{\#} \widehat{\nu}_n, (B, \mathrm{Id})_{\#} \nu)\right] \;.
		\end{split}
	\end{equation*}
	Analogously, $\cH_{\cX \times \cY}(1)$ consists of $K$-uniformly bounded functions, hence
	\begin{align*}
		\sup_{F \in \cF} \mathrm{MMD}_{K_{\cX} \otimes K_{\cY}} ((\mathrm{Id}, F)_{\#}\widehat{\mu}_m, (\mathrm{Id}, F)_{\#}\mu)
		& = \sup_{f \in \cH_{\cX \times \cY}(1) \circ (\mathrm{Id}, \cF)} \left|\int f \dd{\widehat{\mu}_{m}} - \int f \dd{\mu} \right| \\
		& \le 2 R_m(\cH_{\cX \times \cY}(1) \circ (\mathrm{Id}, \cF), \mu) + \sqrt{\frac{2 K^2 \log(1/\delta)}{m}} \;, \\
		\sup_{B \in \cB} \mathrm{MMD}_{K_{\cX} \otimes K_{\cY}} ((B, \mathrm{Id})_{\#} \widehat{\nu}_n, (B, \mathrm{Id})_{\#} \nu)
		& = \sup_{f \in \cH_{\cX \times \cY}(1) \circ (\cB, \mathrm{Id})} \left|\int f \dd{\widehat{\nu}_{n}} - \int f \dd{\nu} \right| \\
		& \le 2 R_n(\cH_{\cX \times \cY}(1) \circ (\cB, \mathrm{Id}), \nu) + \sqrt{\frac{2 K^2 \log(1/\delta)}{n}} \;,
	\end{align*}	
	each of which holds with probability at least $1 - \delta$. Therefore,
	\begin{equation}
		\begin{split}		
			\label{eqn:m3}
			& \sup_{(F,B) \in \cF \times \cB} |\mathrm{MMD}^2_{K_{\cX} \otimes K_{\cY}} ((\mathrm{Id}, F)_{\#}\widehat{\mu}_m, (B, \mathrm{Id})_{\#} \widehat{\nu}_n) - \mathrm{MMD}^2_{K_{\cX} \otimes K_{\cY}} ((\mathrm{Id}, F)_{\#}\mu, (B, \mathrm{Id})_{\#} \nu)| \\
			& \precsim R_m(\cH_{\cX \times \cY}(1) \circ (\mathrm{Id}, \cF), \mu) + R_n(\cH_{\cX \times \cY}(1) \circ (\cB, \mathrm{Id}), \nu) + \sqrt{\frac{\log(1/\delta)}{m}} + \sqrt{\frac{\log(1/\delta)}{n}}
		\end{split}
	\end{equation}
	holds with probability at least $1 - 2 \delta$. We complete the proof by combining \eqref{eqn:m1}, \eqref{eqn:m2}, \eqref{eqn:m3}.
\end{proof}

\begin{proof}[Proof of Lemma~\ref{lem:chaining}]
	For any integer $j \in \N \cup \{0\}$, define $\delta_j = 2^{-j} \Delta$, and let $\cT_{j} \subset \mathbb{S}_0^{m \times m}$ be a minimal $\delta_j$-covering net of $\cT$; clearly, $|\cT_j| = N(\delta_j, \cT)$. For each $j$, the covering set induces a mapping $\Pi_j : \cT \rightarrow \cT_j$ such that
	\begin{align*}
		\sup_{A \in \cT} d(A, \Pi_j(A)) \leq \delta_j \;.
	\end{align*}
	By definition of $\Delta$, we may assume $\cT_0 = \{A_0\}$ so that $\Pi_0(A) = A_0$ for all $A \in \cT$.

	Note that $\E g^\top A_0 g = 0$ by definition. Using this, we write $\E \sup_{A \in \cT} g^\top A g$ as a chaining sum:
	\begin{align*}
		\E \sup_{A \in \cT} g^\top A g
		& = \E_{g} \sup_{A \in \cT} g^\top A g - \E g^\top A_0 g \\ 
		& = \E \sup_{A \in \cT} \left(g^\top A g  - g^\top A_0 g  \right) \\
		& = \E \sup_{A \in \cT} \left( g^\top A g  - g^\top \Pi_J(A) g + \sum_{j=0}^{J-1} g^\top \Pi_{j+1}(A) g -  g^\top \Pi_{j}(A) g \right) \\
		& \leq \E \sup_{A \in \cT}  \left( g^\top A g  - g^\top \Pi_J(A) g \right) + \sum_{j=0}^{J-1} \E \sup_{A \in \cT} \left( g^\top \Pi_{j+1}(A) g -  g^\top \Pi_{j}(A) g \right)\;.
	\end{align*}
	For the first term on RHS, using the Cauchy-Schwarz inequality and Jensen's inequality, we have
	\begin{align*}
		\E \sup_{A \in \cT} \left( g^\top A g  - g^\top \Pi_J(A) g \right) \leq \E \left[  \big(\sum_{i\neq j} g_i^2 g_j^2 \big)^{1/2} \cdot \delta_J \right]  \leq m \delta_J \;.
	\end{align*}
	For each summand in the second term on RHS, use Lemma~\ref{lem:maximal-ineq}. Note that for any $j$, the maximal cardinality of
	\begin{equation*}
		| \{ (\Pi_{j+1}(A), \Pi_{j}(A)) ~:~ A \in \cT \} | \leq N(\delta_{j+1}, \cT) \times N(\delta_{j}, \cT) \leq N(\delta_{j+1}, \cT)^2
	\end{equation*}
	and that
	\begin{equation*}
		d(\Pi_{j+1}(A), \Pi_{j}(A)) \leq  d(\Pi_{j+1}(A), A) + d(A, \Pi_{j}(A)) \leq 3 \delta_{j+1} \;.
	\end{equation*}
	Since $\Pi_{j + 1}(A) - \Pi_j(A) \in \mathbb{S}_0^{m \times m}$ and $\|\Pi_{j + 1}(A) - \Pi_j(A)\| \le 3 \delta_{j + 1}$, Lemma~\ref{lem:maximal-ineq} asserts that for any $j$ 
	\begin{align*}
		\E \sup_{A \in \cT}  \left( g^\top \Pi_{j+1}(A) g -  g^\top \Pi_{j}(A) g \right) \leq 6\delta_{j+1} \sqrt{2\log N(\delta_{j+1}, \cT)} + 12 \delta_{j+1} \log N(\delta_{j+1}, \cT) \;.
	\end{align*}
	Summing over $j$, we have for any $J$, the following inequality,
	\begin{align*}
		\E \sup_{A \in \cT} \left( g^\top A g  - g^\top A_0 g  \right) 
		\leq m \delta_J +  12\int_{\delta_J / 2}^{\Delta/2} \sqrt{2 \log N(\delta, \cT)} \dd{\delta} + 24 \int_{\delta_J / 2}^{\Delta/2} \log N(\delta, \cT) \dd{\delta} \;.
	\end{align*}
\end{proof}

\begin{proof}[Proof of Proposition~\ref{prop:Rademacher-upper-bound-by-chaining}]
	To make use of the chaining inequality, we are only left to bound the covering number $N(\delta, \cT)$ with
	\begin{equation*}
		\cT \coloneqq \{ A_F : F \in \cF \} \subset \mathbb{S}_0^{m \times m}\;.
	\end{equation*}
	Lipschitzness of $K_{\cY}$ (Assumption \ref{a:lip_kernels}) implies 
	\begin{equation*}
		\left|K_{\cY}(F(x_i), F(x_j)) - K_{\cY}(F'(x_i), F'(x_j))\right| \le \frac{L}{2} \|F(x_i) - F'(x_i)\| + \frac{L}{2} \|F(x_j) - F'(x_j)\|\;.
	\end{equation*}
	Hence, 
	\begin{equation*}
		\begin{split}
			d(A_F, A_{F'})
			& \leq m \max_{i\neq j} \left|K_{\cY}(F(x_i), F(x_j)) - K_{\cY}(F'(x_i), F'(x_j))\right| \\
			& \le \frac{m L}{2} \left( \max_{i \in [m]} \|F(x_i) - F'(x_i)\| + \max_{j \in [m]} \|F(x_j) - F'(x_j)\| \right)\\
			& = m L \max_{i \in [m]} \| F(x_i) - F'(x_i) \| \\
			& \le m L \max_{i \in [m]} \left(\sum_{k = 1}^{\mathrm{dim}(\cY)}  |F_k(x_i) - F_k'(x_i)|^2 \right)^{1/2} \\
			& \le m L \left[\sum_{k = 1}^{\mathrm{dim}(\cY)} \left(\max_{i \in [m]} |F_k(x_i) - F_k'(x_i)|\right)^2 \right]^{1/2} \;.	
		\end{split}
	\end{equation*}
	As in the proof of Proposition \ref{prop:union-bound-GW-term}, for $\epsilon > 0$, let $\cN_\infty(\epsilon, \cF_k, \{x_i\}_{i=1}^{m})$ be a minimal $\epsilon$-covering net of $\cF_k$. Then, one can easily see that 
	\begin{equation*}
		\left\{ A_F : F \in \otimes_{k = 1}^{\mathrm{dim}(\cY)}\cN_\infty \left( \tfrac{\delta}{m L \sqrt{\mathrm{dim}(\cY)}}, \cF_{k}, \{x_i\}_{i = 1}^{m}) \right) \right\}
	\end{equation*}
	is a $\delta$-covering of $\cT$. Therefore, we conclude
	\begin{equation*}
		N(\delta, \cT) 
		\le
		\prod_{k=1}^{{\rm dim}(\cY)} N_\infty \left(\tfrac{\delta}{m L \sqrt{\mathrm{dim}(\cY)}}, \cF_{k}, \{x_i\}_{i = 1}^{m} \right) 
		\le \prod_{k=1}^{{\rm dim}(\cY)} N_\infty \left(\tfrac{\delta}{m L \sqrt{\mathrm{dim}(\cY)}}, \cF_{k}, m \right) \;.
	\end{equation*}
	Lastly, we bound $N_\infty(\delta, \cF_k, m)$ by the pseudo-dimension of $\cF_k$ via Lemma \ref{lem:uniform_covering_to_pdim}. Combining Assumption \ref{a:uniform_boundedness} and Lemma \ref{lem:uniform_covering_to_pdim}, we have 
	\begin{align*}
		N(\delta, \cT) 
		\leq \prod_{k=1}^{{\rm dim}(\cY)} \left( \frac{2 e m b}{ \tfrac{\delta}{m L \sqrt{\mathrm{dim}(\cY)}} \cdot {\rm Pdim}(\cF_{k})} \right)^{{\rm Pdim}(\cF_{k})} \;.
	\end{align*}

	Now, to apply Lemma~\ref{lem:chaining}, fix $F_0 \in \cF$ and let $A_0 = A_{F_0}$. Then, 
	\begin{equation*}
		\int_{\delta_J / 2}^{\Delta/2} \log N(\delta, \cT) \dd{\delta} 
		\leq \Delta / 2 \cdot \sum_{k=1}^{{\rm dim}(\cY)} {\rm Pdim}(\cF_{k}) \cdot \log \left( \frac{2 e m b}{ \tfrac{1}{m L \sqrt{\mathrm{dim}(\cY)}} \Delta 2^{-J} / 2 \cdot {\rm Pdim}(\cF_{k})} \right) \;.
	\end{equation*}
	First, we obtain an upper bound on $\Delta$:
	\begin{equation*}
		\Delta = \sup_{F \in \cF} \|A_F - A_0\| \le m \max_{i\neq j} \left|K_{\cY}(F(x_i), F(x_j)) - K_{\cY}(F_0(x_i), F_0(x_j))\right| \le 2 m K \;.
	\end{equation*}
	Next, we claim that $\Delta$ is bounded below by a universal constant; this is to upper bound $\Delta$ in the denominator. Consider $y_0$ and $y_0'$ given in Assumption \ref{a:separation}. We may assume that $F_0$ is the constant map explained in Assumption \ref{a:separation}: $F_0(x) = y_0$ for all $x \in \cX$. Without loss of generality, we assume $x_1 \neq x_2$. Then, we can find $F \in \cF$ such that $F(x_1) = y_0$ and $F(x_2) = y_0'$ according to Assumption \ref{a:separation}. Hence, 
	\begin{equation*}
		\Delta \ge |K_{\cY}(F(x_1), F(x_2)) - K_{\cY}(F_0(x_1), F_0(x_2))| = |K_{\cY}(y_0, y_0') - K_{\cY}(y_0, y_0)| > 0 \;. 
	\end{equation*}
	Therefore, with the choice of $J$ such that $m 2^{-J} \asymp \sum_k {\rm Pdim}(\cF_{k})$, 
	\begin{equation*}
		\begin{split}
			\int_{\delta_J / 2}^{\Delta/2} \log N(\delta, \cT) \dd{\delta} 
			& \precsim m \left[\sum_{k=1}^{{\rm dim}(\cY)} {\rm Pdim}(\cF_{k}) \right] \cdot \log \left(\frac{m}{\min_{k \in [\mathrm{dim}(\cY)]} {\rm Pdim}(\cF_{k}) } \right) \\
			& \le m \log(m) \left[\sum_{k=1}^{{\rm dim}(\cY)} {\rm Pdim}(\cF_{k}) \right] \;.
		\end{split} 
	\end{equation*}
	Analogously, 
	\begin{equation*}
		\int_{\delta_J / 2}^{\Delta/2} \sqrt{2 \log N(\delta, \cT)} \dd{\delta} 
		\precsim m \log(m) \left[\sum_{k=1}^{{\rm dim}(\cY)} {\rm Pdim}(\cF_{k}) \right] \;.
	\end{equation*}
	Thus,
	\begin{align*}
		R_m(\cH_y(1)\circ \cF, \{x_i\}_{i=1}^{m}) \precsim \frac{1}{m} \left[ m K + \E_{g} \sup_{F \in \cF}  g^\top A_F g \right]^{1/2}
		\precsim \sqrt{\frac{\log m}{m} \sum_{k =1}^{{\rm dim}(\cY)} {\rm Pdim}(\cF_{k}) }
	\end{align*}
	
	The same argument can be applied to the other three Rademacher complexities. Hence, we have proved the proposition.
\end{proof}

\subsection{Auxiliary Lemmas}
\begin{lemma}[Theorem 4.10 of \cite{wainwright_2019}]
	\label{lem:bound}
	Let $(\cZ, \rho)$ be a probability space and $\cG$ be a class of $b$-uniformly bounded measurable functions defined on $\cZ$, that is, $\sup_{g \in \cG} \|g\|_\infty \le b$. Let $z_1, \ldots, z_m$ are i.i.d.\ samples from $\rho$ and let $\widehat{\rho}_m$ be the empirical measure constructed from them. Then, for any $\delta > 0$,
	\begin{equation*}
		\sup_{g \in \cG} \left|\int g \dd{\widehat{\rho}_m} - \int g \dd{\rho} \right| \le 2 R_m(\cG, \rho) + \sqrt{\frac{2 b^2 \log(1 / \delta)}{m}}
	\end{equation*}
	holds with probability at least $1 - \delta$.
\end{lemma}

\begin{lemma}[Theorem 12.2 of \cite{anthony_bartlett_1999}]
	\label{lem:uniform_covering_to_pdim}
	Let $\cG$ be a collection of real-valued functions defined on a set $\cZ$. Suppose $\sup_{g \in \cG} ||g||_\infty = b < \infty$. For $\epsilon > 0$ and $m \ge \mathrm{Pdim}(\cG)$,  
	\begin{equation*}
		N_\infty(\epsilon, \cG, m) \le \left( \frac{2 e m b}{\epsilon \cdot {\rm Pdim}(\cG)} \right)^{{\rm Pdim}(\cG)} \;.
	\end{equation*}
\end{lemma}

\begin{lemma}[Example 2.12 of \cite{boucheron_lugosi_massart_2013}]
	 For any $A \in \mathbb{S}_0^{m \times m}$ and $0 \leq \lambda < 1/(2\|A\|_{\mathrm{op}})$,
	\begin{align}
		\log \E_g e^{\lambda g^\top A g} \leq \frac{\lambda^2 \|A \|}{1- 2\lambda \| A \|_{\mathrm{op}}} \;,
	\end{align}
	where $g \sim N(0, I_m)$. Here, $\| \cdot \|_{\mathrm{op}}$ denotes the operator norm of $A$.
\end{lemma}
This lemma tells that $g^\top A g$ is a sub-Gamma random variable with variance factor $2 \|A\|^2$ and scale parameter $2 \|A\|_{\mathrm{op}}$ (see Chapter 2.4 of \cite{boucheron_lugosi_massart_2013} for the definition). Using Corollary 2.6 of the same text, we can derive the following maximal inequality.
\begin{lemma}[Maximal inequality]
	\label{lem:maximal-ineq}
	For $A_1, \ldots, A_N \in \mathbb{S}_0^{m \times m}$, suppose $\max_{i = 1, \ldots, N} \| A_i \| \leq \delta$. Then,
	\begin{align}
		\E_{g} \max_{i = 1, \ldots, N} g^\top A_i g \leq 2 \delta \left( \sqrt{\log N} + \log N \right) \;,
	\end{align}
	where $g \sim N(0, I_m)$.
\end{lemma}

\section{Computational Aspects of the RGM Distance} 
\label{sec:comparison-with-GW}

Recall from Section \ref{sec:RGM-sampler} that we have utilized the Lagrangian form instead of the constrained form \eqref{eq:RGM} of the RGM distance. Such a practical computation allows us to easily implement the RGM sampler and we have observed its good empirical performance in Section \ref{sec:numerical}. This section shifts our focus to the exact computation of the RGM distance; we discuss conditions under which minimizing the Lagrangian form leads to a close approximation to the RGM distance. Another important computational aspect is the comparison of the RGM distance and the GW distance. $\mathrm{GW}(\mu, \nu) \le \mathrm{RGM}(\mu, \nu)$ holds in theory by Proposition \ref{prop:1}; using a concrete example, we approximate both quantities numerically and see how large the gap between them is. Lastly, we examine the numerical performance of the convex formulation discussed in \ref{subsec:cvx-representer} and compare with the results from the Lagrangian form.

\subsection{Approximation with the Lagrangian Form}
\label{subsec:approx-math}
First, we derive connections between the RGM distance and its Lagrangian formulation. As in Section \ref{sec:statistical-theory}, let 
\begin{equation*}
	C_0(F, B) = \int (c_{\cX}(x, B(y)) - c_{\cY}(F(x), y))^2 \dd{\mu \otimes \nu} \;
\end{equation*}
for any $F \colon \cX \to \cY$ and $B \colon \cY \to \cX$ so that $\mathrm{RGM}(\mu, \nu)^2 = \inf_{(F, B) \in \cI(\mu, \nu)} C_0(F, B)$. Define the Lagrangian form as 
\begin{equation*}
	L_{\lambda_1, \lambda_2, \lambda_3}(F, B) = C_0(F, B) + \lambda_1 \cdot \underbrace{\cL_{\cX \times \cY}((\mathrm{Id}, F)_{\#} \mu, (B, \mathrm{Id})_{\#} \nu)}_{\ell_1(F, B)} + \lambda_2 \cdot \underbrace{\cL_{\cX}(\mu, B_{\#} \nu)}_{\ell_2(F, B)} + \lambda_3 \cdot \underbrace{\cL_{\cY}(F_{\#} \mu, \nu)}_{\ell_3(F, B)} \;,
\end{equation*}
where $\cL_{\cX \times \cY}$, $\cL_{\cX}$, and $\cL_{\cY}$ are suitable nonnegative discrepancy measures as in Section \ref{sec:RGM-sampler}; in particular, we assume $\ell_1(F, B) = \ell_2(F, B) = \ell_3(F, B) = 0$ for $(F, B) \in \cI(\mu, \nu)$.\footnote{Though we may define the Lagrangian form without $\cL_{\cX}$ and $\cL_{\cY}$, we include them for a seamless connection with the experiment results in Section \ref{subsec:comparion-experiments}.} 

Suppose $\lambda_1, \lambda_2, \lambda_3 \ge 0$, then
\begin{equation*}
	\inf_{\substack{F \colon \cX \to \cY \\ B \colon \cY \to \cX}} L_{\lambda_1, \lambda_2, \lambda_3}(F, B) 
	\leq \inf_{(F, B) \in \cI(\mu, \nu)} L_{\lambda_1, \lambda_2, \lambda_3}(F, B) 
	= \mathrm{RGM}(\mu, \nu)^2
	= \inf_{(F, B) \in \cI(\mu, \nu)} C_0(F, B) \;.
\end{equation*}
We seek a minimizer of $L_{\lambda_1, \lambda_2, \lambda_3}$ over $\cF \times \cB$, namely, the product of suitable function classes as discussed in Section \ref{subsec:stat_rate}. Let 
\begin{equation}
	\label{eq:minimizer-lag}
	(F^\star, B^\star) \in \argmin_{(F, B) \in \cF \times \cB} L_{\lambda_1, \lambda_2, \lambda_3}(F, B) \;.
\end{equation}
If $\cI(\mu, \nu) \subseteq \cF \times \cB$, 
\begin{align*}
	C_0(F^\star, B^\star) 
	& \le L_{\lambda_1, \lambda_2, \lambda_3}(F^\star, B^\star) \quad (\because \lambda_1, \lambda_2, \lambda_3 \ge 0) \\
	& = \inf_{(F, B) \in \cF \times \cB} L_{\lambda_1, \lambda_2, \lambda_3}(F, B) \quad (\because \eqref{eq:minimizer-lag}) \\
	& \leq \inf_{(F, B) \in \cI(\mu, \nu)} L_{\lambda_1, \lambda_2, \lambda_3}(F, B) \quad (\because \cI(\mu, \nu) \subseteq \cF \times \cB) \\
	& = \mathrm{RGM}(\mu, \nu)^2 \;.
\end{align*}
Roughly speaking, if the function classes are rich enough to ensure $\cI(\mu, \nu) \subseteq \cF \times \cB$, the minimizer $(F^\star, B^\star)$ produces a lower bound $C_0(F^\star, B^\star)$ on $\mathrm{RGM}(\mu, \nu)^2$. On the other hand, if the minimizer satisfies the constraint $(F^\star, B^\star) \in \cI(\mu, \nu)$, then $C_0(F^\star, B^\star)$ is an upper bound on $\mathrm{RGM}(\mu, \nu)^2$ by definition.

Therefore, a sufficient condition for $C_0(F^\star, B^\star) = \mathrm{RGM}(\mu, \nu)^2$ is that the following two hold: $\cI(\mu, \nu) \subseteq \cF \times \cB$ and $(F^\star, B^\star) \in \cI(\mu, \nu)$.

\subsection{Numerical Experiments}
\label{subsec:comparion-experiments}
Using a concrete example, we compute the aforementioned quantities related to the RGM distance, approximate the GW distance, and compare them; we will also discuss the results from the convex formulation in Section~\ref{subsec:cvx-representer}. Throughout, we consider two point clouds on $\R^2$ as in Figure \ref{fig:Comparison}(a), that is, $\mu$ and $\nu$ are uniform distributions supported on 30 grid points of a segment and a circle, respectively. We fix the cost functions: $c_{\cX} = c_{\cY}$ is the RBF kernel that maps $(x, y)$ to $\exp(-\|x - y\|^2)$.
\begin{figure}[ht]
	\centering\includegraphics[width=8cm]{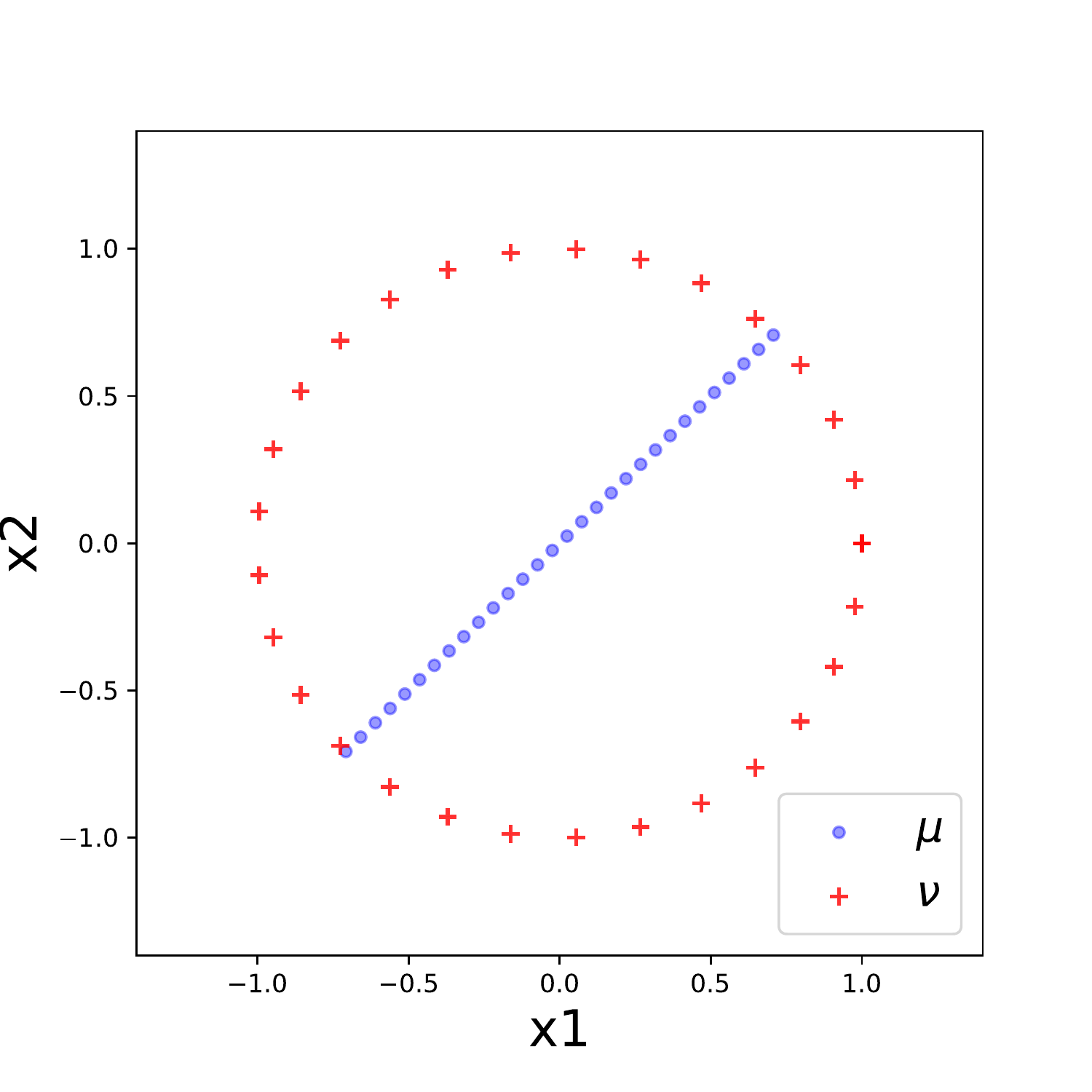}
    \caption{The supports of $\mu$ and $\nu$ are the grid points of a segment connecting $(-1, -1)$ and $(1, 1)$ and a circle $x^2 + y^2 = 1$, respectively.}
    \label{fig:Comparison}
\end{figure}

First, we aim to compute the quantities discussed in Section \ref{subsec:approx-math}. To this end, we specify the discrepancy measures and the function classes as follows.
\begin{itemize}
	\item $\cL_{\cX \times \cY} = \mathrm{MMD}_{K_{\cX} \otimes K_{\cY}}^2$, $\cL_{\cX} = \mathrm{MMD}_{K_{\cX}}^2$, $\cL_{\cY} = \mathrm{MMD}_{K_{\cY}}^2$, where $K_{\cX} = K_{\cY} = c_{\cX}$.
	\item $\cF = \cB$ is the class of neural networks with two hidden layers as follows:
	\begin{equation*}
		\{x \mapsto \tanh(W_2 \tanh(W_1 x + b_1)) + b_2) : W_1 \in \R^{2 \times 30}, b_1 \in \R^{30}, W_2 \in \R^{30 \times 2}, b_2 \in \R^{2}\} \;,
	\end{equation*}
	where tanh is the tangent hyperbolic function applied elementwise, that is, $\tanh(x) = (\tanh(x_1), \ldots, \tanh(x_k)) \in \R^k$ for $x = (x_1, \ldots, x_k) \in \R^k$. 
\end{itemize}
\begin{figure}[ht]
	\centering
    \subfloat[$\lambda_1 = \lambda_2 = \lambda_3 = 1$]{{\includegraphics[width=8cm]{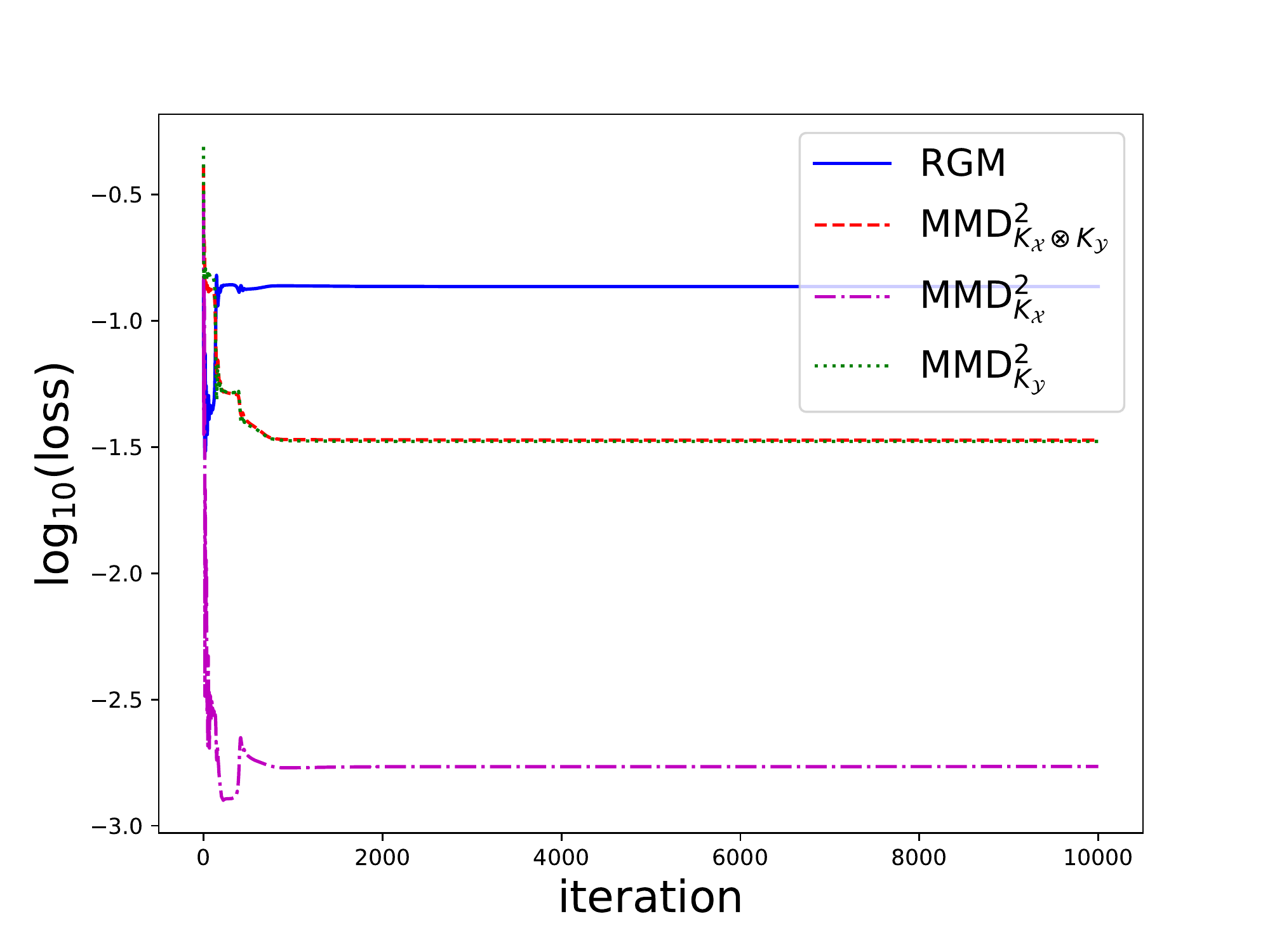}}}    
	\subfloat[$\lambda_1 = \lambda_2 = \lambda_3 = 10^2$]{{\includegraphics[width=8cm]{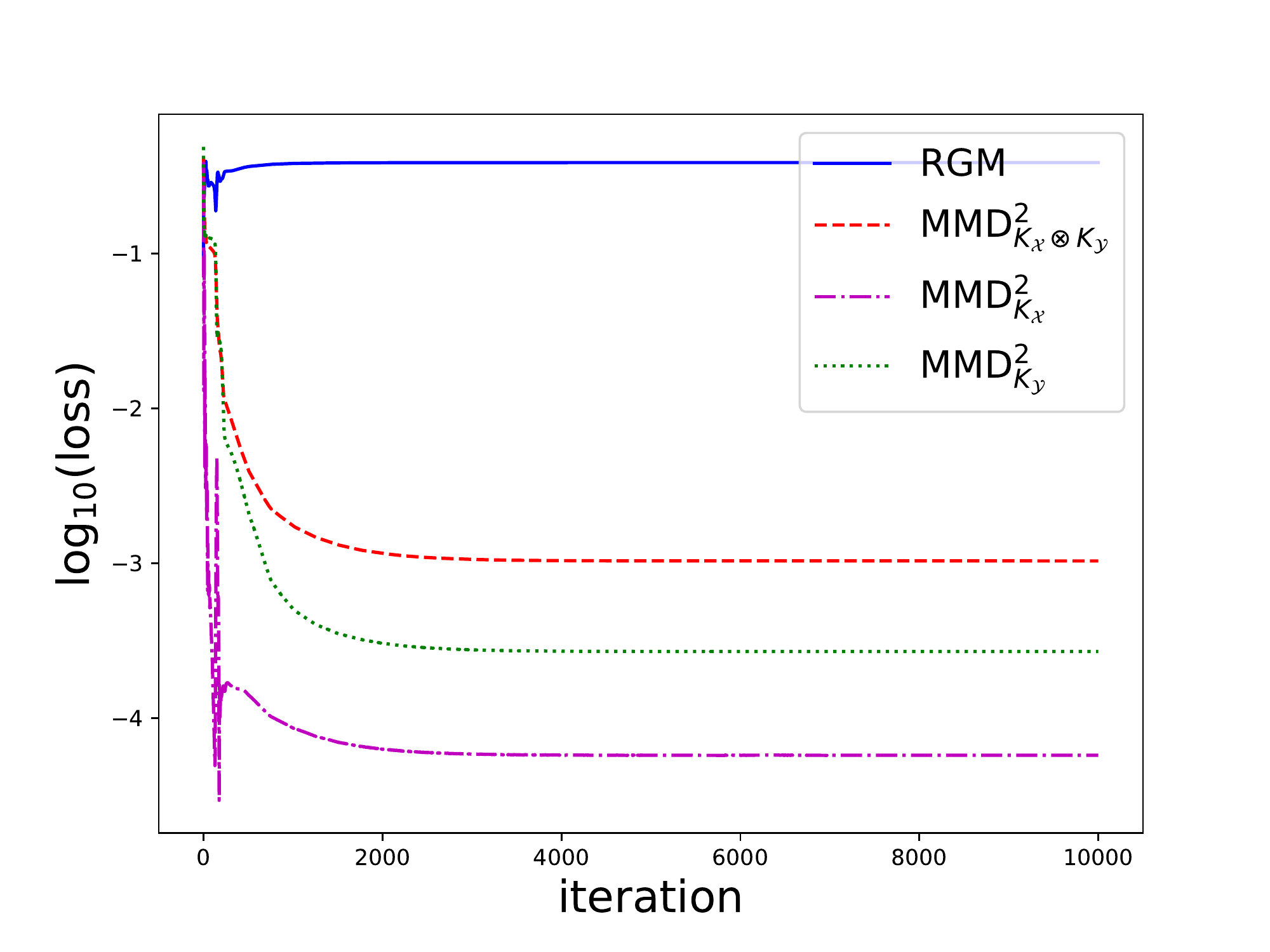}}} \\
	\subfloat[$\lambda_1 = \lambda_2 = \lambda_3 = 10^4$]{{\includegraphics[width=8cm]{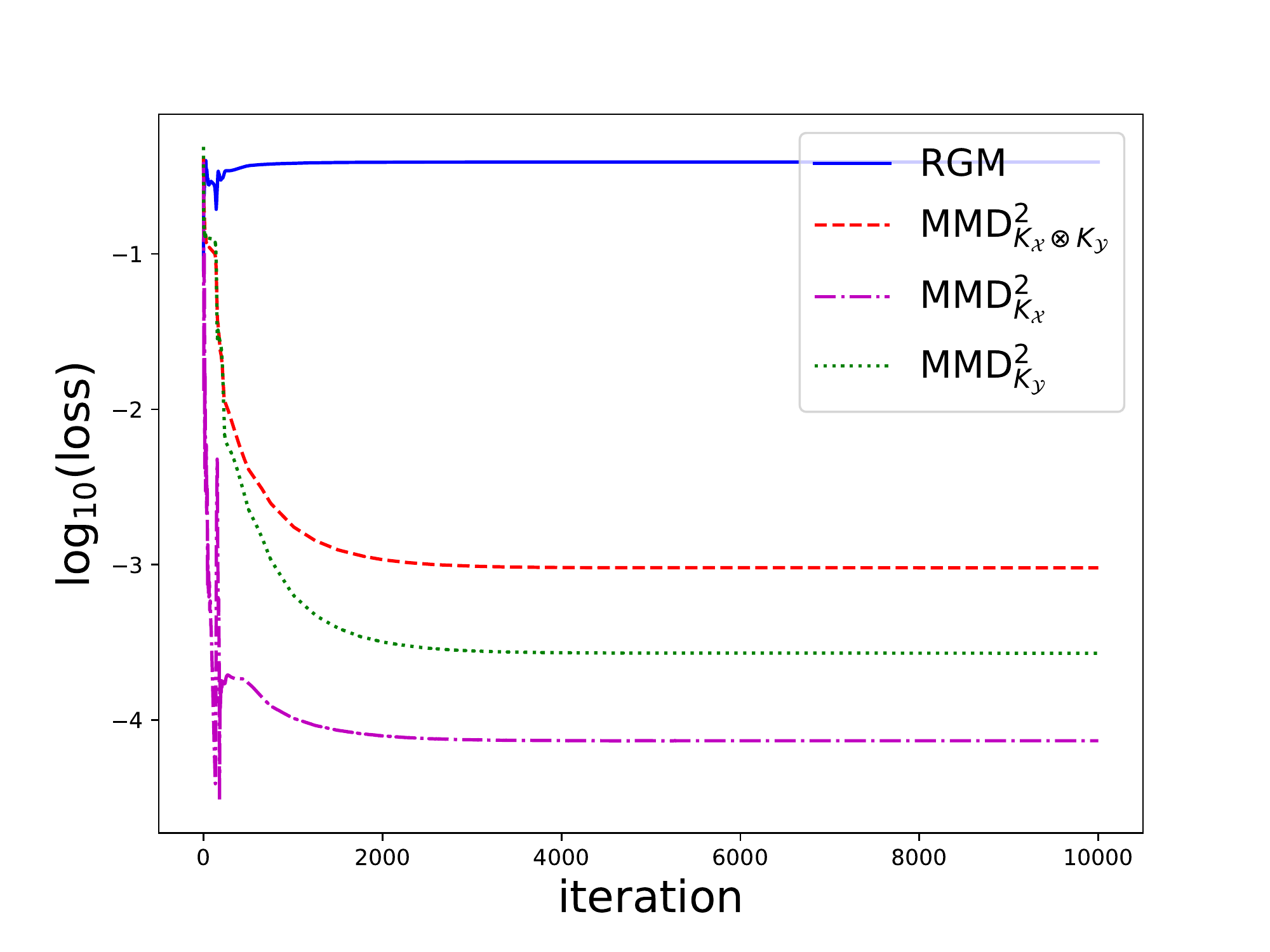}}}
    \caption{Training curves (10000 iterations). }
    \label{fig:training-nn}
\end{figure}
We use Adam \citep{kingma2017}, a variant of stochastic gradient descent, to find a minimizer of $L_{\lambda_1, \lambda_2, \lambda_3}$ over $\cF \times \cB$ as in \eqref{eq:minimizer-lag}. After 10000 iterations, we can see that the loss converges as in Figure \ref{fig:training-nn}, indicating we have a local minimizer $(\widehat{F}^\star, \widehat{B}^\star) \in \cF \times \cB$. Since this optimization problem may be nonconvex, there is no guarantee that this is a global minimizer, hence 
\begin{equation*}
	C_0(F^\star, B^\star) \le C_0(\widehat{F}^\star, \widehat{B}^\star) \;.
\end{equation*}
To examine the experiment results in light of Section \ref{subsec:approx-math}, we will assume
\begin{itemize}
	\item $\cF$ and $\cB$ are rich enough to ensure $\cI(\mu, \nu) \subset \cF \times \cB$,
	\item $(F, B) \in \cI(\mu, \nu)$ if and only if $\ell_1(F, B) = \ell_2(F, B) = \ell_3(F, B) = 0$,\footnote{This is true since we are using the RBF kernel as mentioned in Section \ref{sec:numerical}.}
	\item $(\widehat{F}^\star, \widehat{B}^\star)$ is indeed a global minimizer: $(\widehat{F}^\star, \widehat{B}^\star) \in \argmin_{(F, B) \in \cF \times \cB} L_{\lambda_1, \lambda_2, \lambda_3}(F, B)$.
\end{itemize}
Under these assumptions, $(\widehat{F}^\star, \widehat{B}^\star) \in \cI(\mu, \nu)$ implies $C_0(\widehat{F}^\star, \widehat{B}^\star) = \mathrm{RGM}(\mu, \nu)^2$ as discussed in Section \ref{subsec:approx-math}. To verify $(\widehat{F}^\star, \widehat{B}^\star) \in \cI(\mu, \nu)$, we check the values of $\ell_1(F, B)$, $\ell_2(F, B)$, and $\ell_3(F, B)$ in Table \ref{tab:computation-results}. We observe that they get smaller as we increase the values of the Lagrangian multipliers. For the cases where $\lambda_1 = \lambda_2 = \lambda_3 = 10^2$ or $\lambda_1 = \lambda_2 = \lambda_3 = 10^4$, the values $\ell_1, \ell_2, \ell_3$ are sufficiently small to conclude $(\widehat{F}^\star, \widehat{B}^\star) \in \cI(\mu, \nu)$; hence, we can roughly estimate $\mathrm{RGM}(\mu, \nu)^2 \approx C_0(\widehat{F}^\star, \widehat{B}^\star) \approx 0.39$.

\begin{table}[ht]
	\centering
	\begin{tabular}{| c | c | c | c |} 
	\hline
	& $\lambda_1 = \lambda_2 = \lambda_3 = 1$ & $\lambda_1 = \lambda_2 = \lambda_3 = 10^2$ & $\lambda_1 = \lambda_2 = \lambda_3 = 10^4$ \\
	\hline
	$C_0(\widehat{F}^\star, \widehat{B}^\star)$ & 0.136 & 0.386 & 0.390 \\
	$\ell_1(\widehat{F}^\star, \widehat{B}^\star)$ & $3.366 \times 10^{-2}$ & $1.034 \times 10^{-3}$ & $9.550 \times 10^{-4}$ \\
	$\ell_2(\widehat{F}^\star, \widehat{B}^\star)$ & $1.716 \times 10^{-3}$ & $5.758 \times 10^{-5}$ & $7.379 \times 10^{-5}$ \\
	$\ell_3(\widehat{F}^\star, \widehat{B}^\star)$ & $3.327 \times 10^{-2}$ & $2.689 \times 10^{-4}$ & $2.698 \times 10^{-4}$ \\
	\hline
	$L_{\lambda_1, \lambda_2, \lambda_3}(\widehat{F}^\star, \widehat{B}^\star)$ & 0.205 & 0.522 & 13.377 \\
	\hline
	\end{tabular}
	\caption{Minimum values of the Lagrangian form.}
	\label{tab:computation-results}
\end{table}

\paragraph{Comparison with GW}
As discussed earlier, exact computation of the GW distance (Definition \ref{def:mms}) is impossible in general. Here, we estimate it using an off-the-shelf computational tool called Python Optimal Transport (POT) \citep{flamary2021pot} widely used in literature, which yields $\mathrm{GW}^2(\mu, \nu) \approx 0.171$.\footnote{Technically, this should be an upper bound on the exact value of $\mathrm{GW}^2(\mu, \nu)$ because the result of POT ought to be a local minimizer.} Combined with the previous computation, we can say
\begin{equation*}
	1 \le \frac{\mathrm{RGM}(\mu, \nu)}{\mathrm{GW}(\mu, \nu)} \approx \sqrt{\frac{0.390}{0.171}} = 1.515 \;,
\end{equation*}
indicating that the RGM distance is approximately the GW distance times 1.5. This rough computation is based on the aforementioned assumptions regarding the RGM computation and the accuracy of POT in computing the GW distance. 

Instead, we may give an upper bound on the ratio of the two distances using the well-known lower bounds on the GW distance: the First Lower Bound (FLB) and the Second Lower Bound (SLB) on GW \citep{memoli_2011}. Letting $\mu = \frac{1}{m} \sum_{i = 1}^{m} \delta_{x_i}$ and $\nu = \frac{1}{n} \sum_{j = 1}^{n} \delta_{y_j}$ with $m = n = 30$, these bounds are computed as follows:
\begin{align*}
	\mathrm{FLB}^2(\mu, \nu) & = W_2^2\left(\frac{1}{m} \sum_{i = 1}^{m} \delta_{e_{\cX}(x_i)}, \frac{1}{n} \sum_{j = 1}^{n} \delta_{e_{\cY}(y_j)}\right) \approx 0.061 \;, \\
	\mathrm{SLB}^2(\mu, \nu) & = W_2^2\left(\frac{1}{m^2} \sum_{i, i' = 1}^{m} \delta_{c_{\cX}(x_i, x_{i'})}, \frac{1}{n^2} \sum_{j, j' = 1}^n\delta_{c_{\cY}(y_j, y_{j'})}\right) \approx 0.135 \;,
\end{align*}
where $e_{\cX}(x_i) = \sqrt{\frac{1}{m} \sum_{i' = 1}^{m} c_{\cX}^2(x_i, x_{i'})}$ and $e_{\cY}(y_j) = \sqrt{\frac{1}{n} \sum_{j' = 1}^{n} c_{\cY}^2(y_j, y_{j'})}$ are called the eccentricity; see \citep{memoli_2011} for details. These quantities, computed by using POT as well, are known to be lower bounds on $\mathrm{GW}(\mu, \nu)^2$, hence 
\begin{equation*}
	1 \le \frac{\mathrm{RGM}(\mu, \nu)}{\mathrm{GW}(\mu, \nu)} \le \sqrt{\frac{0.390}{0.135}} = 1.700 \;.
\end{equation*}
Therefore, we can conclude that the ratio of the two distances is bounded by 1.7.

\paragraph{Convex formulation}
Next, we estimate the RGM distance based on the convex formulation; as in Theorem \ref{thm:representer}, we solve the convex optimization problem:
\begin{equation}
	\label{eq:convex-relax}
	\min_{\substack{\mathsf{F}_{m, n} \in \R^{m \times n} \\ \mathsf{B}_{n, m} \in \R^{n \times m}}} \omega(\mathsf{F}_{m, n}, \mathsf{B}_{n, m}) \; ,
\end{equation}
where
\begin{align*}
	\omega(\mathsf{F}_{m, n}, \mathsf{B}_{n, m})
	& = \overbrace{\frac{1}{mn} \| \bK_{\cY} \mathsf{B}_{n, m} \bK_{\cX} - \bK_{\cY} \mathsf{F}_{m, n}^\top \bK_{\cX} \|^2}^{c_0(\mathsf{F}_{m, n}, \mathsf{B}_{n, m})} + \lambda_1 \cdot \overbrace{\left\| \frac{1}{m} \bK_{\cX}^{3/2} \mathsf{F}_{m,n} \bK_{\cY}^{1/2} - \frac{1}{n} \bK_{\cX}^{1/2} \mathsf{B}_{n,m}^\top \bK_{\cY}^{3/2} \right\|^2}^{m_1(\mathsf{F}_{m, n}, \mathsf{B}_{n, m})} \\
	& + \lambda_2 \cdot \underbrace{\left\|\bK_{\cX}^{1/2} \cdot \left(\frac{1}{m} \mathbf{1}_m - \mathsf{B}_{n, m}^\top \bK_{\cY} \frac{1}{n} \mathbf{1}_n \right) \right\|^2}_{m_2(\mathsf{F}_{m, n}, \mathsf{B}_{n, m})} + \lambda_3 \cdot \underbrace{\left\|\bK_{\cY}^{1/2} \cdot \left(\frac{1}{n} \mathbf{1}_n - \mathsf{F}_{m, n}^\top \bK_{\cX} \frac{1}{m} \mathbf{1}_m \right) \right\|^2}_{m_3(\mathsf{F}_{m, n}, \mathsf{B}_{n, m})}
\end{align*}
as derived in Section \ref{sec:representation}. It should be noted that the minimum of \eqref{eq:convex-relax} is a lower bound on the minimum of the Lagrangian, that is, 
\begin{equation*}
	\min_{\substack{\mathsf{F}_{m, n} \in \R^{m \times n} \\ \mathsf{B}_{n, m} \in \R^{n \times m}}} \omega(\mathsf{F}_{m, n}, \mathsf{B}_{n, m}) 
	\le 
	\min_{(F, B) \in \cF \times \cB} L_{\lambda_1, \lambda_2, \lambda_3}(F, B) \;.
\end{equation*}
To see this, first note that the RHS is exactly \eqref{eqn:1}.\footnote{Since $\mu$ and $\nu$ are discrete, the RHS is the same as its empirical estimate \eqref{eqn:1}.} Then, recall from Section \ref{subsec:cvx-representer} that \eqref{eqn:1} = \eqref{eqn:op} is relaxed to \eqref{eqn:relaxed} and is further relaxed to the convex problem \eqref{eq:convex-relax} due to Theorem \ref{thm:representer}. 
\begin{table}[ht]
	\centering
	\begin{tabular}{| c | c | c | c |} 
	\hline
	& $\lambda_1 = \lambda_2 = \lambda_3 = 1$ & $\lambda_1 = \lambda_2 = \lambda_3 = 10^2$ & $\lambda_1 = \lambda_2 = \lambda_3 = 10^4$ \\
	\hline
	$c_0(\mathsf{F}_{m, n}^\star, \mathsf{B}_{n, m}^\star)$ & 0.044 & 0.106 & 0.108 \\
	$m_1(\mathsf{F}_{m, n}^\star, \mathsf{B}_{n, m}^\star)$ & 0.011 & $4.521 \times 10^{-6}$ & $6.369 \times 10^{-10}$ \\
	$m_2(\mathsf{F}_{m, n}^\star, \mathsf{B}_{n, m}^\star)$ & 0.012 & $2.592 \times 10^{-6}$ & $2.626 \times 10^{-12}$ \\
	$m_3(\mathsf{F}_{m, n}^\star, \mathsf{B}_{n, m}^\star)$ & 0.001 & $1.050 \times 10^{-7}$ & $2.461 \times 10^{-10}$ \\
	\hline
	$\omega(\mathsf{F}_{m, n}^\star, \mathsf{B}_{n, m}^\star)$ & 0.068 & 0.107 & 0.108 \\
	\hline
	\end{tabular}
	\caption{Mimimum values of the convex problem \eqref{eq:convex-relax} obtained by CVXPY.}
	\label{tab:convex-results}
\end{table}

Table \ref{tab:convex-results} shows the results obtained by a convex optimization tool called CVXPY \cite{diamond2016cvxpy}. The mimimum $\omega(\mathsf{F}_{m, n}^\star, \mathsf{B}_{n, m}^\star)$ of \eqref{eq:convex-relax} is always smaller than $\mathrm{GW}^2(\mu, \nu) \approx 0.171$ and is between the two lower bounds: $\mathrm{FLB}^2(\mu, \nu) \approx 0.061$ and $\mathrm{SLB}^2(\mu, \nu) \approx 0.135$. Also, the MMD terms vanish if we use the large Lagrangian multipliers, indicating that the constraints (represented via the MMD terms) are met. That said, we can see that the gap between the minimum and the RGM distance can be large. Therefore, finding conditions under which this gap vanishes would be interesting future work.

\section{Details of the Experiments in Section \ref{sec:numerical}}
\label{sec:appendix2}
Here, we provide implementation details of the experiments in Section \ref{sec:numerical}.

\subsection{Gaussian}
In the Gaussian experiment in Section \ref{sec:numerical}, we minimize \eqref{eqn:1} using Adam for $3000$ iterations. The learning rate at the initial iteration is $0.1$ and we halve it after every $500$ iterations. Figure \ref{fig:training-curves}(a) shows the training curve.

\begin{figure}[ht]
    \centering    
	\subfloat[Gaussian]{{\includegraphics[width=8cm]{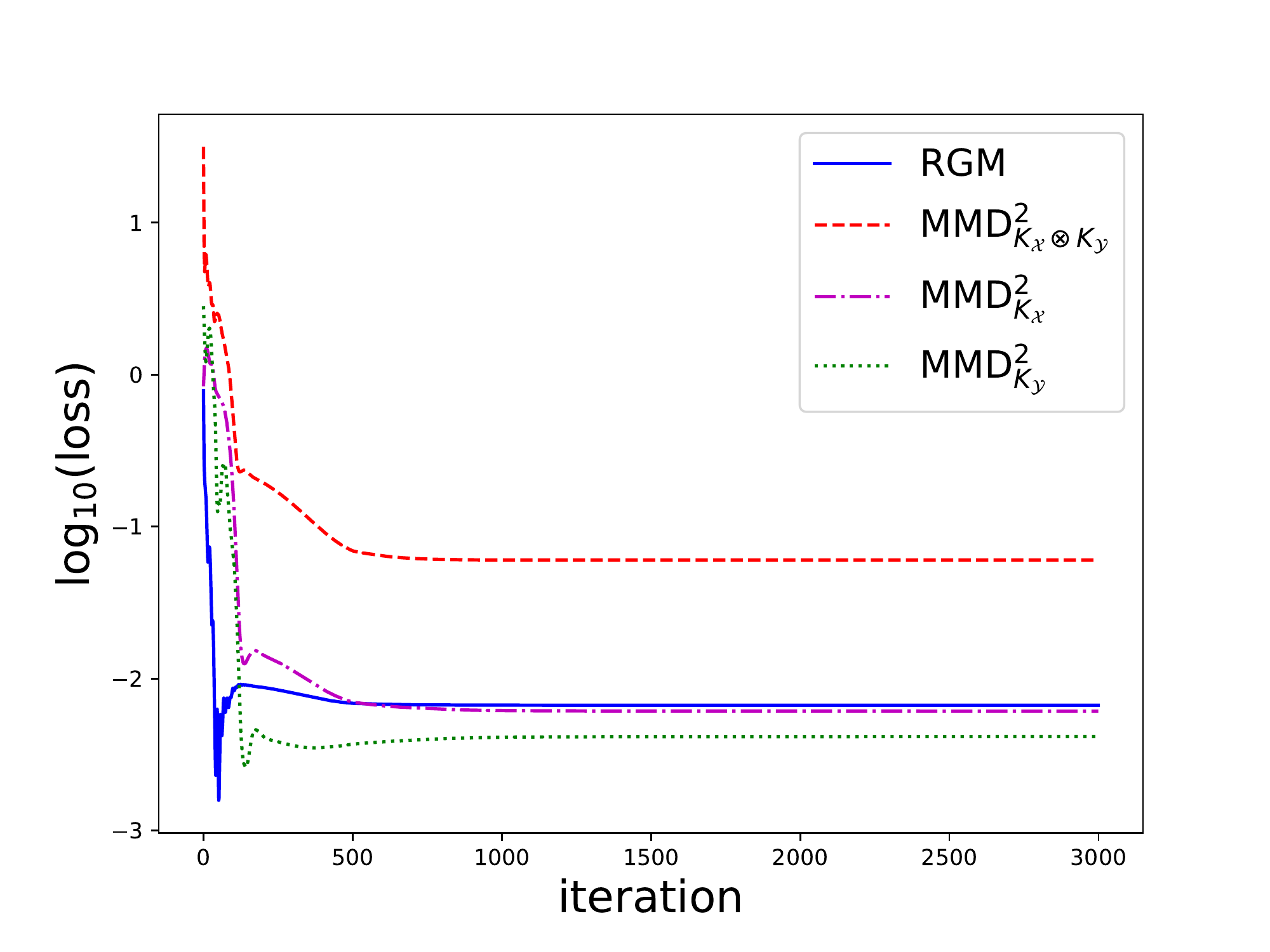}}}
	\subfloat[MNIST]{{\includegraphics[width=8cm]{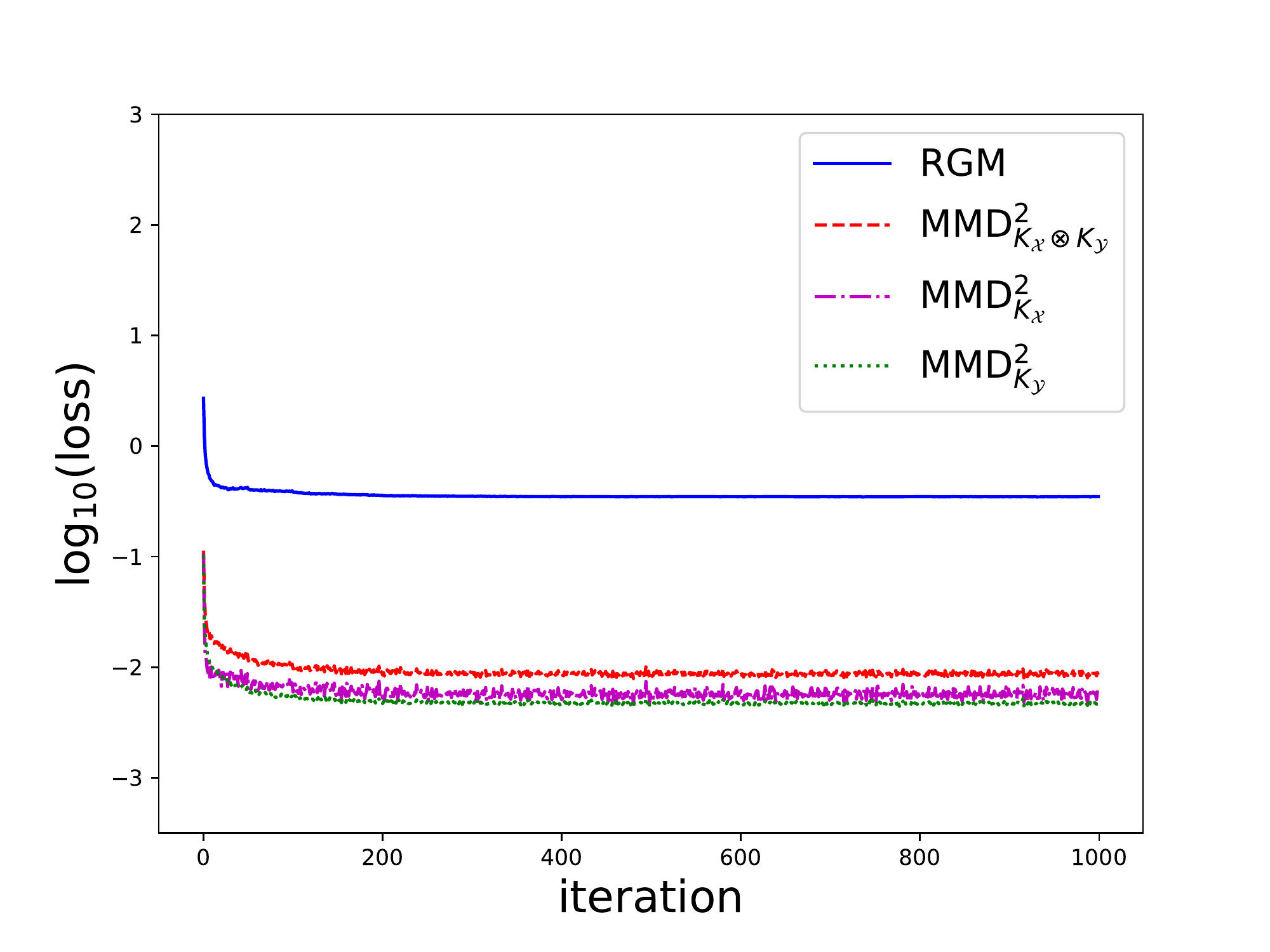}}}
    \caption{Training curves for the experiments: (a) the Gaussian experiment, (b) the MNIST experiment with $\cX = \R^2$.}
    \label{fig:training-curves}
\end{figure}

\subsection{MNIST}
\paragraph{$\R^2$ and MMDs} 
For numerical stability, we encode a variant of empirical estimate \eqref{eqn:1} as our loss: 
\begin{align}
	&\lambda_1\cdot\frac{1}{m n} \sum_{i = 1}^{m} \sum_{j = 1}^{n} (c_{\cX}(x_i, B(y_j)) - c_{\cY}(F(x_i), y_j))^2  + \mathrm{MMD}_{K_{\cX} \otimes K_{\cY}}^2((\mathrm{Id}, F)_{\#} \widehat{\mu}_m, (B, \mathrm{Id})_{\#} \widehat{\nu}_n)\nonumber\\
	& + \lambda_2 \cdot \mathrm{MMD}_{K_{\cX}}^2(\widehat{\mu}_m, B_{\#} \widehat{\nu}_n) + \lambda_3 \cdot \mathrm{MMD}_{K_{\cY}}^2(F_{\#} \widehat{\mu}_m, \widehat{\nu}_n) \;. \label{eqn:mnist}
\end{align}
We choose tuning parameters $(\lambda_1,\lambda_2,\lambda_3) = (0.01, 1,1)$. 
For fully connected neural networks $F$ and $B$, we apply the rectified linear unit (ReLU) activation function $\sigma(x) = \max(x, 0)$ to all three hidden layers of $F$ and $B$, and an additional tangent hyperbolic (tanh) function $\tanh(x) = (e^x - e^{-x})/(e^x + e^{-x})$ to the output layer of $F$, both of which are elementwise activation functions. To put it explicitly, $y = F(x)$ is defined as
\begin{align*}
    h_0 &= x, \quad x\in\R^2\\
    h_l &= \sigma(W_l h_{l-1}+b_l), \quad l = 1,2\\
    y &= \tanh(W_3h_2 + b_3)
\end{align*}
with
$W_1 \in \R^{50\times 2}, W_2\in \R^{50\times 50}, W_3\in\R^{784\times 50},
b_1, b_2 \in\R^{50\times1}, b_3\in \R^{784\times 1}$.
Similarly, $\widetilde{x} = B(\widetilde{y})$ is defined as
\begin{align*}
    \widetilde{h}_0 &= \widetilde{y}, \quad \widetilde{y}\in\R^{784}\\
    \widetilde{h}_l &= \sigma(\widetilde{W}_l\widetilde{h}_{l-1}+\widetilde{b}_l), \quad l = 1,2\\
    \widetilde{x} &= \widetilde{W}_3\widetilde{h}_2 + \widetilde{b}_3 
\end{align*}
with $\widetilde{W}_1 \in \R^{50\times 784}, \widetilde{W}_2\in \R^{50\times 50}, \widetilde{W}_3\in\R^{2\times 50},
\widetilde{b}_1, \widetilde{b}_2 \in\R^{50\times1}, \widetilde{b}_3\in \R^{2\times 1}$.
The training set is randomly devided into minibatches of size $256$, for which we run Adam again for 1000 iterations. The learning rate at the initial iteration is $0.005$ and we halve it after every $500$ iterations.
Figure~\ref{fig:snapshot} shows the generated images during the training process. 

\begin{figure}[ht]
    \centering
    \subfloat[Before training]{\includegraphics[width=5cm]{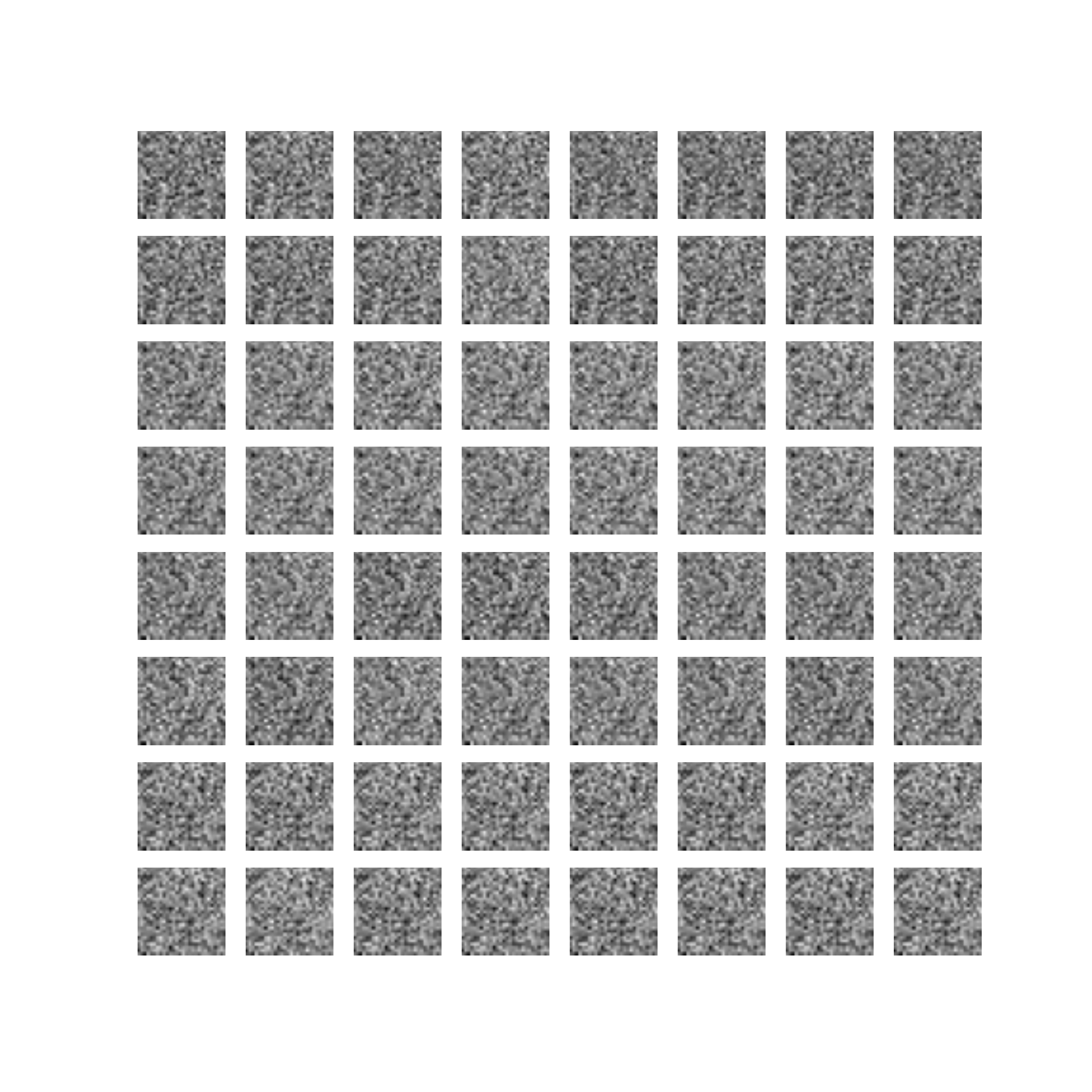}\hspace{-2em} }%
    \subfloat[After 20 iterations]{\includegraphics[width=5cm]{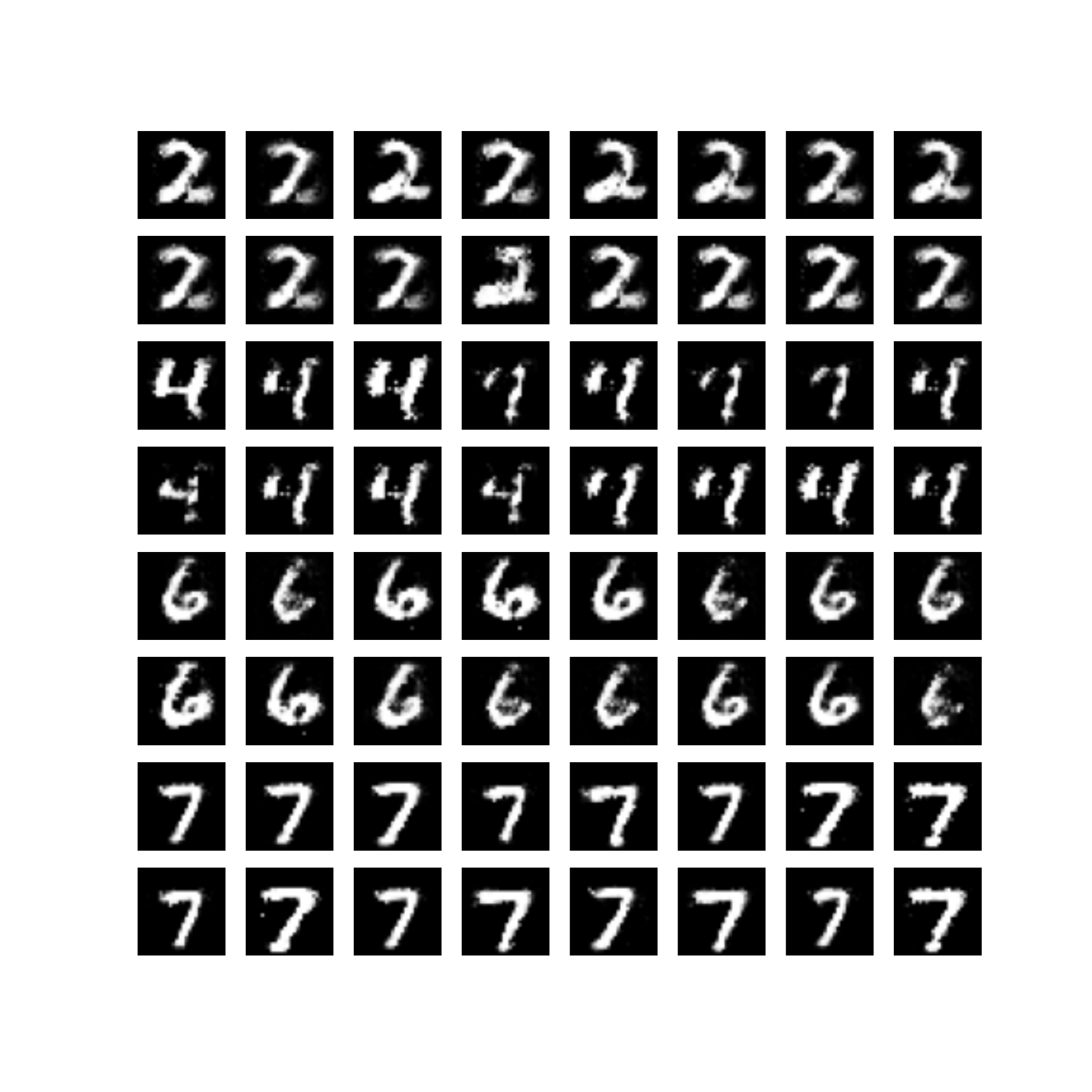}\hspace{-2em} }\\%
    \subfloat[After 50 iterations]{\includegraphics[width=5cm]{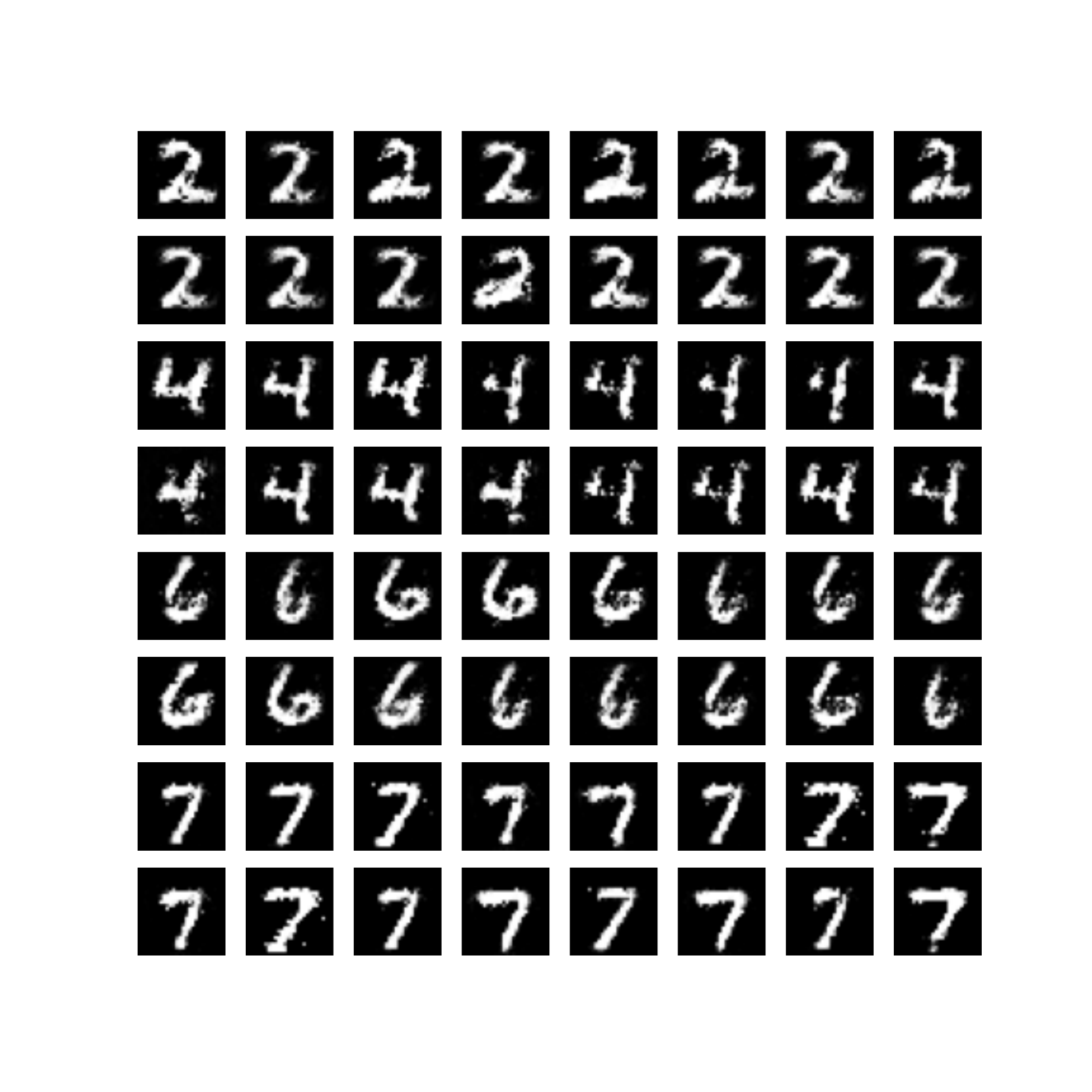}\hspace{-2em} }%
    \subfloat[After 1000 iterations]{\includegraphics[width=5cm]{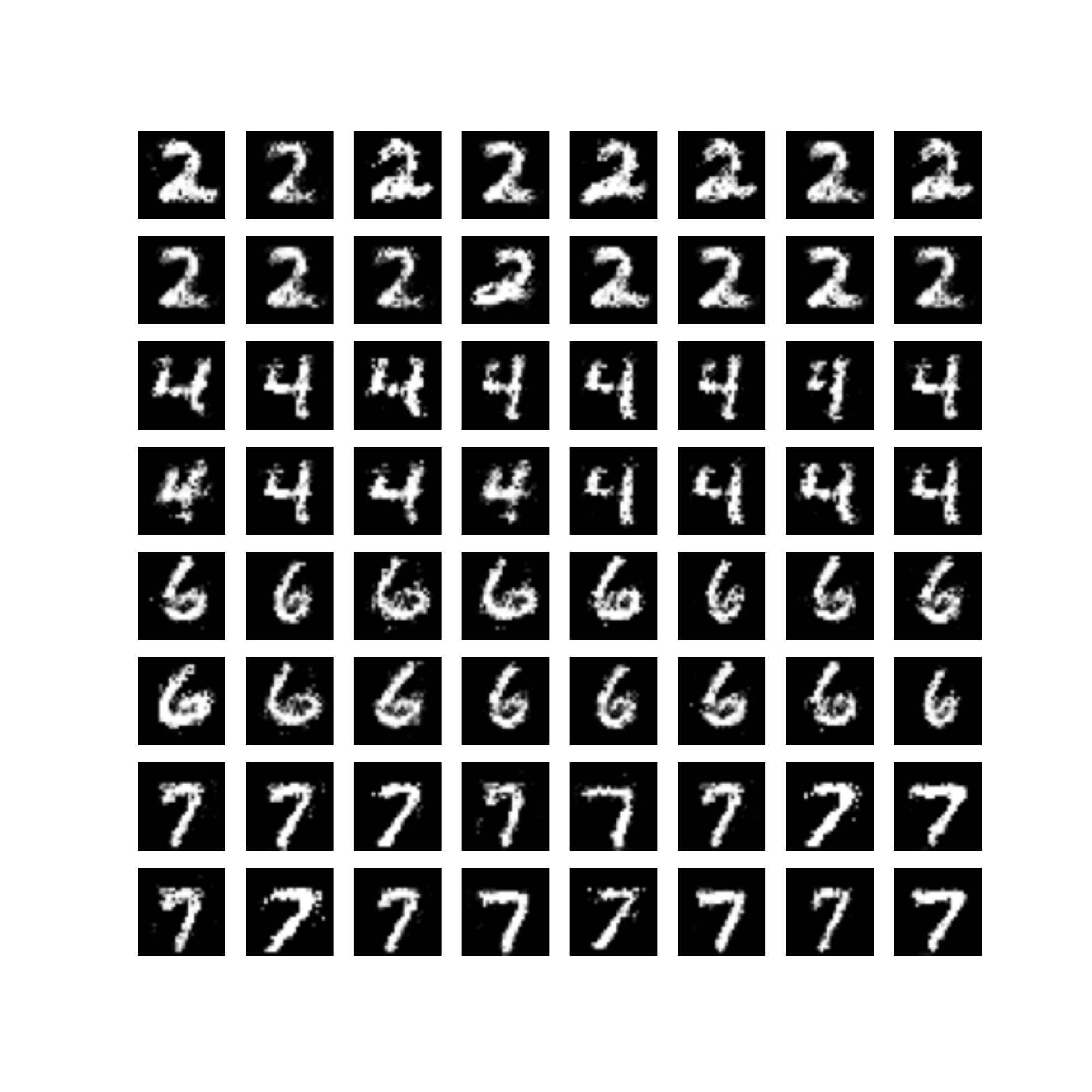}\hspace{-2em} }%
    \caption{Generated images on MNIST data for digits $2, 4, 6, 7$ during the training process, under the experimental setup with $\R^2$ and MMD discrepancy. }
    \label{fig:snapshot}
\end{figure}

\paragraph{$\R^4$ and Sinkhorn divergences}

We implement the Sinkhorn divergence with squared Euclidean cost by using GeomLoss \citep{feydy2019interpolating}. 
Concretely, we first define the entropic regularized Kantorovich problem between $\widehat{\mu}_m$ and $B_{\#} \widehat{\nu}_n$ on some Euclidean space $\cX$
\begin{equation*}
	W_{2, \epsilon}^2(\widehat{\mu}_m,B_{\#} \widehat{\nu}_n) = \min_{\gamma \in \Pi(\widehat{\mu}_m, B_{\#} \widehat{\nu}_n)} \sum_{i = 1}^m \sum_{j = 1}^n \gamma_{ij}\left(\|x_i - B(y_j)\|^2+ \epsilon\log(\gamma_{ij})\right)
\end{equation*}
where $\gamma$ is a coupling matrix and $\gamma_{ij}$ denotes its $(i,j)$ element. 
Then the Sinkhorn divergence between two empirical measures is defined by $s_{\epsilon, \cX}(\widehat{\mu}_m,B_{\#} \widehat{\nu}_n) = W_{2, \epsilon}^2(\widehat{\mu}_m,B_{\#} \widehat{\nu}_n) - \frac{1}{2}W_{2, \epsilon}^2(\widehat{\mu}_m,\widehat{\mu}_m) - \frac{1}{2}W_{2, \epsilon}^2(B_{\#} \widehat{\nu}_n,B_{\#} \widehat{\nu}_n)$. 
We encode our loss by replacing MMDs in \eqref{eqn:mnist} with Sinkhorn divergences
\begin{align*}
	&\lambda_1\cdot\frac{1}{m n} \sum_{i = 1}^{m} \sum_{j = 1}^{n} (c_{\cX}(x_i, B(y_j)) - c_{\cY}(F(x_i), y_j))^2  + s_{\epsilon, \cX\times\cY}((\mathrm{Id}, F)_{\#} \widehat{\mu}_m, (B, \mathrm{Id})_{\#} \widehat{\nu}_n)\\
	& + \lambda_2 \cdot s_{\epsilon,\cX}(\widehat{\mu}_m, B_{\#} \widehat{\nu}_n) + \lambda_3 \cdot s_{\epsilon, \cY}(F_{\#} \widehat{\mu}_m, \widehat{\nu}_n) \;, 
\end{align*}
and choose tuning parameters $(\lambda_1,\lambda_2,\lambda_3) = (1,1,1)$. 
The Sinkhorn parameter $\epsilon$ is set to be $0.0001$ for all three discrepancy measures. 
Again, $F\colon \R^4\to\R^{784}$ and $B\colon \R^{784}\to\R^4$ are parametrized by fully connected neural networks with three hidden layers, whose activation functions are same as the MMD case. 
The rest of the setups, including the choice of optimizer, number of iterations, and batchsize, are same as the MMD case above.

\paragraph{Comparison between RGM and GW}
Lastly, let us estimate the gap between the RGM distance and the GW distance in the MNIST example ($\R^2$ and MMDs) based on the discussions in Section \ref{sec:comparison-with-GW}. Recall that we have obtained a minimizer $(\widehat{F}, \widehat{B})$ of \eqref{eqn:1} over $\cF \times \cB$ using samples $\{x_i\}_{i = 1}^{20000}$ and $\{y_j\}_{j = 1}^{20000}$ from $\mu = N(0, I_2)$ and $\nu =$ the distribution of the four digits, respectively; though this is a local minimizer as the optimization problem may be nonconvex, we will assume that this is indeed a global minimizer as in Section \ref{sec:comparison-with-GW}. Letting $\widehat{\mu}_m$ and $\widehat{\nu}_n$ be the empirical measures constructed by $\{x_i\}_{i = 1}^{20000}$ and $\{y_j\}_{j = 1}^{20000}$, respectively ($m = n = 20000$), we obtain the following quantities:
\begin{align*}
	\widehat{C}_0(\widehat{F}, \widehat{B}) \overset{\text{Section \ref{sec:statistical-theory}}}{=} \int (c_{\cX}(x, \widehat{B}(y)) - c_{\cY}(\widehat{F}(x), y))^2 \dd{\widehat{\mu}_m \otimes \widehat{\nu}_n} & \approx 0.348 \;, \\
	m_1 \coloneqq \mathrm{MMD}_{K_{\cX} \otimes K_{\cY}}^2((\mathrm{Id}, \widehat{F})_{\#} \widehat{\mu}_m, (\widehat{B}, \mathrm{Id})_{\#} \widehat{\nu}_n) &\approx 2.192 \times 10^{-3} \;,\\
	m_2 \coloneqq \mathrm{MMD}_{K_{\cX}}^2(\widehat{\mu}_m, \widehat{B}_{\#} \widehat{\nu}_n) &\approx 8.261 \times 10^{-5} \;,\\ 
	m_3 \coloneqq \mathrm{MMD}_{K_{\cY}}^2(\widehat{F}_{\#} \widehat{\mu}_m, \widehat{\nu}_n) &\approx 1.437 \times 10^{-3} \;.
\end{align*}

Hence, $C(\widehat{\mu}_m, \widehat{\nu}_n, \widehat{F}, \widehat{B}) = \widehat{C}_0(\widehat{F}, \widehat{B}) + \sum_{k = 1}^{3} \lambda_k m_k = 0.719$, where $C$ is defined in \eqref{eqn:cost-def}. Recall from Section \ref{sec:statistical-theory} that
\begin{equation*}
	\begin{split}
		|C(\widehat{\mu}_m, \widehat{\nu}_n, \widehat{F}, \widehat{B}) - C(\mu, \nu, \widehat{F}, \widehat{B})|
		& \le \sup_{(F, B) \in \cF \times \cB} |C(\widehat{\mu}_m, \widehat{\nu}_n, F, B) - C(\mu, \nu, F, B)| \\
		& \precsim \cM(\cF, \cB, m, n, \delta)
	\end{split}
\end{equation*}
holds with probability at least $1 - \delta$, where $\cM(\cF, \cB, m, n, \delta)$ is defined in Theorem \ref{thm:stat}. Assuming that this complexity measure is sufficiently small for $m = n = 20000$, we may roughly say $C(\mu, \nu, \widehat{F}, \widehat{B}) \approx C(\widehat{\mu}_m, \widehat{\nu}_n, \widehat{F}, \widehat{B}) = 0.719$. In the same vein, Theorem \ref{thm:stat} indicates
\begin{equation*}
	\inf_{(F, B) \in \cF \times \cB} C(\mu, \nu, F, B) \approx C(\mu, \nu, \widehat{F}, \widehat{B}) \approx 0.719\;. 
\end{equation*}
Now, we combine this result with the discussion in Section \ref{sec:comparison-with-GW}. First, by definition, 
\begin{equation*}
	\inf_{(F, B) \in \cF \times \cB} C(\mu, \nu, F, B) = \inf_{(F, B) \in \cF \times \cB} L_{\lambda_1, \lambda_2, \lambda_3}(F, B) \;. 
\end{equation*}
We have derived in Section \ref{sec:comparison-with-GW} that 
\begin{equation*}
	\mathrm{RGM}(\mu, \nu)^2 = C_0(F^\star, B^\star) = L_{\lambda_1, \lambda_2, \lambda_3}(F^\star, B^\star) = \inf_{(F, B) \in \cF \times \cB} L_{\lambda_1, \lambda_2, \lambda_3}(F, B) 
\end{equation*}
if $\cI(\mu, \nu) \subseteq \cF \times \cB$ and the minimizer $(F^\star, B^\star)$ defined in \eqref{eq:minimizer-lag} satisfies $(F^\star, B^\star) \in \cI(\mu, \nu)$. Therefore, under these assumptions, we can roughly estimate
\begin{equation*}
	\mathrm{RGM}^2(\mu, \nu) \approx 0.719.
\end{equation*}

Lastly, let us estimate the lower bounds, FLB and SLB, on the GW distance as in Section \ref{subsec:comparion-experiments}. Due to the computational complexity, we will use subsets of the training data, say $\{x_i\}_{i = 1}^{2000}$ and $\{y_j\}_{j = 1}^{2000}$, to construct plug-in estimators $\mathrm{FLB}(\widehat{\mu}_{2000}, \widehat{\nu}_{2000})$ and $\mathrm{SLB}(\widehat{\mu}_{2000}, \widehat{\nu}_{2000})$; using POT again, we have
\begin{align*}
	\mathrm{FLB}^2(\mu, \nu) \approx \mathrm{FLB}^2(\widehat{\mu}_{2000}, \widehat{\nu}_{2000}) \approx 0.006 \;, \\
	\mathrm{SLB}^2(\mu, \nu) \approx \mathrm{SLB}(\widehat{\mu}_{2000}, \widehat{\nu}_{2000}) \approx 0.148 \;.
\end{align*}

Therefore, we can roughly conclude
\begin{equation*}
	1 \le \frac{\mathrm{RGM}(\mu, \nu)}{\mathrm{GW}(\mu, \nu)} \le \sqrt{\frac{0.719}{0.148}} = 2.204 \;.
\end{equation*}

\end{document}